\documentclass[a4paper,11pt]{article}

\usepackage{vmargin}
 
\setmarginsrb{1in}{1in}{1in}{1in}{0pt}{0pt}{0pt}{6mm}

\title{Characterizing the easy-to-find subgraphs from the viewpoint of polynomial-time algorithms, kernels, and Turing kernels
\thanks{This work was partially supported by the European Research Council through starting grant 306992 ``Parameterized Approximation'' and grant 280152 ``PARAMTIGHT: Parameterized complexity and the search for tight complexity results'' and OTKA grant NK105645.}}

\author{Bart M. P. Jansen\thanks{University of Bergen, Norway, bart.jansen@ii.uib.no} \\
\and
D\'{a}niel Marx \thanks{Institute for Computer Science and Control, Hungarian Academy of Sciences (MTA SZTAKI), Budapest,
Hungary, dmarx@cs.bme.hu}
}

\usepackage[tight]{subfigure} 
\usepackage{float}

\usepackage{tikz}
\usetikzlibrary{shapes}
\usetikzlibrary{calc}
\usetikzlibrary{decorations.pathreplacing}
\usetikzlibrary{decorations.pathmorphing}

\usepackage{hyperref}
\usepackage{boxedminipage}
\usepackage{xspace}
\usepackage{comment}
\usepackage{amsmath}
\usepackage{amssymb}
\usepackage{shadethm}
\usepackage{amsthm}
\usepackage{thm-restate}
\usepackage{graphicx} 	
\usepackage[shortlabels,inline]{enumitem}
\setlist{noitemsep}
\setlist[enumerate,1]{label=(\arabic*), ref=(\arabic*)}
\usepackage{url}
\usepackage{tabularx}

\newcommand{\cqed}{\renewcommand{\qedsymbol}{$\lrcorner$}} 
\newcommand{\retheorem}{theorem}
\newcommand{\reabctheorem}{abctheorem}
\newcommand{\relemma}{lemma}
\setlist{noitemsep}

\newcommand{\executeiffilenewer}[3]{%
\ifnum\pdfstrcmp{\pdffilemoddate{#1}}%
{\pdffilemoddate{#2}}>0%
{\immediate\write18{#3}}\fi%
} 

\newcommand{\svg}[2]{\def\svgwidth{#1}%
\executeiffilenewer{#2.svg}{#2.pdf}%
{inkscape -z -D --file=#2.svg %
--export-pdf=#2.pdf --export-latex}%
{\input{#2.pdf_tex}}}

\usepackage{xcolor}
\definecolor{dark-red}{rgb}{0.4,0.15,0.15}
\definecolor{dark-blue}{rgb}{0.15,0.15,0.4}
\definecolor{medium-blue}{rgb}{0,0,0.5}
\definecolor{gray}{rgb}{0.5,0.5,0.5}

\hypersetup{
    colorlinks, linkcolor={dark-red},
    citecolor={dark-blue}, urlcolor={medium-blue}
}

\theoremstyle{plain}
\newtheorem{theorem}{Theorem}[section]
\newtheorem{proposition}[theorem]{Proposition}
\newtheorem{lemma}[theorem]{Lemma}

\newtheorem{corollary}[theorem]{Corollary}

\newtheorem{claim}[theorem]{Claim}

\theoremstyle{definition}

\newtheorem{definition}[theorem]{Definition}
\newtheorem{observation}[theorem]{Observation}

\newenvironment{claimproof}{\begin{proof}\renewcommand{\qedsymbol}{\claimqed}}{\end{proof}\renewcommand{\qedsymbol}{\plainqed}}

\let\plainqed\qedsymbol

\newcommand{\eqvr}[0]{\ensuremath{\mathcal{R}}\xspace}

\newcommand{\containment}[0]{NP~$\subseteq$~coNP$/$poly\xspace}
\newcommand{\ncontainment}[0]{NP~$\not \subseteq$~coNP$/$poly\xspace}
\newcommand{\Oh}{\mathcal{O}}
\newcommand{\yes}{\textsc{yes}\xspace}
\newcommand{\no}{\textsc{no}\xspace}

\newcommand{\F}{\ensuremath{\mathcal{F}}\xspace}
\newcommand{\G}{\ensuremath{\mathcal{G}}\xspace}
\newcommand{\Q}{\ensuremath{\mathcal{Q}}\xspace}
\renewcommand{\L}{\ensuremath{\mathcal{L}}\xspace}
\newcommand{\I}{\ensuremath{\mathcal{I}}\xspace}

\renewcommand{\P}{\ensuremath{\mathcal{P}}\xspace}
\renewcommand{\S}{\ensuremath{\mathcal{S}}\xspace}
\newcommand{\N}{\ensuremath{\mathcal{N}}\xspace}

\newcommand{\kPath}{\textsc{Long Path}\xspace}
\newcommand{\kFPacking}{\textsc{\F-Packing}\xspace}
\newcommand{\kFSubgraphTest}{\textsc{\F-Subgraph Test}\xspace}
\newcommand{\FPacking}{\textsc{\F-Packing}\xspace}
\newcommand{\HPacking}{\textsc{$H$-Packing}\xspace}
\newcommand{\FSubgraphTest}{\textsc{\F-Subgraph Test}\xspace}
\newcommand{\Packing}{\textsc{Packing}\xspace}
\newcommand{\SubgraphTest}{\textsc{Subgraph Test}\xspace}
\newcommand{\kPacking}{\textsc{Packing}\xspace}
\newcommand{\kSubgraphTest}{\textsc{Subgraph Test}\xspace}
\newcommand{\nRegularExactSetCover}{\textsc{Uniform Exact Set Cover ($n$)}\xspace}
\newcommand{\nExactSetCover}{\textsc{Exact Set Cover ($n$)}\xspace}
\newcommand{\XTC}{\textsc{Exact Cover by 3-Sets}\xspace}

\newcommand{\TrianglePackingCanonical}{\textsc{Triangle Packing in a Canonical Subgraph}\xspace}
\newcommand{\PathPackingCanonical}{\textsc{$P_3$ Packing in a Canonical Subgraph}\xspace}
\newcommand{\bipRamsey}{\ensuremath{R_{\mathrm{B}}}}

\newcommand{\Ffamily}[2]{\ensuremath{\mathcal{F}_{\textup{\sf #1}}^{#2}}\xspace}
\newcommand{\Fpath}{\Ffamily{Path}{}}
\newcommand{\Fclique}{\Ffamily{Clique}{}}
\newcommand{\Fbiclique}{\Ffamily{Biclique}{}}
\newcommand{\Fdoublebroom}[1]{\Ffamily{2-broom}{#1}}
\newcommand{\Ffountain}[1]{\Ffamily{Fountain}{#1}}
\newcommand{\Flongfountain}[2]{\Ffamily{LongFountain}{#1,#2}}
\newcommand{\Foperahouse}[1]{\Ffamily{OperaHouse}{#1}}
\newcommand{\Fsubdivstar}{\Ffamily{SubDivStar}{}}
\newcommand{\Fsubdivtree}[1]{\Ffamily{SubDivTree}{#1}}
\newcommand{\Fdiamondfan}{\Ffamily{DiamondFan}{}}

\newcommand{\Graph}[2]{\textup{\sf #1}(#2)}

\newcommand{\Gpath}[1]{\Graph{Path}{#1}}
\newcommand{\Gclique}[1]{\Graph{Clique}{#1}}
\newcommand{\Gbiclique}[1]{\Graph{Biclique}{#1}}
\newcommand{\Gdoublebroom}[2]{\Graph{2-broom}{#1,#2}}
\newcommand{\Gfountain}[2]{\Graph{Fountain}{#1,#2}}
\newcommand{\Glongfountain}[3]{\Graph{LongFountain}{#1,#2,#3}}
\newcommand{\Goperahouse}[2]{\Graph{OperaHouse}{#1,#2}}
\newcommand{\Gsubdivstar}[1]{\Graph{SubDivStar}{#1}}
\newcommand{\Gsubdivtree}[2]{\Graph{SubDivTree}{#1,#2}}
\newcommand{\Gdiamondfan}[1]{\Graph{DiamondFan}{#1}}

\newcommand{\problemdef}[3]
{
\begin{quote}
\textsc{#1}\\
\textbf{Input:} #2\\
\textbf{Question:} #3
\end{quote}
}

\newcommand{\parproblemdef}[4]
{
\begin{quote}
\textsc{#1}\\
\textbf{Instance:} #2\\
\textbf{Parameter:} #3\\
\textbf{Question:} #4
\end{quote}
}

\begin{document}

\maketitle

\begin{abstract}
  We study two fundamental problems related to finding subgraphs: (1) given
  graphs $G$ and $H$, \SubgraphTest asks if $H$ is isomorphic to a
  subgraph of $G$, (2) given graphs $G$, $H$, and an integer $t$,
  \Packing asks if $G$ contains $t$ vertex-disjoint subgraphs
  isomorphic to $H$. For every graph class $\F$, let \FSubgraphTest
  and \FPacking be the special cases of the two problems where $H$ is
  restricted to be in $\F$. Our goal is to study which classes $\F$
  make the two problems tractable in one of the following senses:
\begin{itemize}
\item (randomized) polynomial-time solvable,
\item admits a polynomial (many-one) kernel (that is, has a polynomial-time preprocessing procedure that creates an equivalent instance whose size is polynomially bounded by the size of the solution), or 
\item admits a polynomial Turing kernel (that is, has an adaptive polynomial-time procedure that
  reduces the problem to a polynomial number of instances, each of
  which has size bounded polynomially by the size of the solution).
\end{itemize}
To obtain a more robust setting, we restrict our attention to hereditary
classes $\F$.

It is known that if every component of every graph in $\F$ has at most
two vertices, then \FPacking is polynomial-time solvable, and NP-hard
otherwise. We identify a simple combinatorial property (every
component of every graph in $\F$ either has bounded size or is a
bipartite graph with one of the sides having bounded size) such that
if a hereditary class $\F$ has this property, then \FPacking admits a
polynomial kernel, and has no polynomial (many-one) kernel otherwise, unless the
polynomial hierarchy collapses.  Furthermore, if $\F$ does not have this
property, then \FPacking is either \textup{WK[1]}-hard, \textup{W[1]}-hard, or \textsc{Long
  Path}-hard, giving evidence that it does not admit polynomial Turing
kernels either.

For \FSubgraphTest, we show that if every graph of a hereditary class
$\F$ satisfies the property that it is possible to delete a bounded
number of vertices such that every remaining component has size at
most two, then \FSubgraphTest is solvable in randomized polynomial time
and it is NP-hard otherwise. We introduce a combinatorial property
called $(a,b,c,d)$-splittability and show that if every graph in a
hereditary class $\F$ has this property, then \FSubgraphTest admits a
polynomial Turing kernel and it is \textup{WK[1]}-hard, \textup{W[1]}-hard, or \textsc{Long
  Path}-hard otherwise. We do not give a complete characterization of the
cases when \FSubgraphTest admits polynomial many-one kernels, but show
examples that this question is much more fragile than the
characterization for Turing kernels.

\end{abstract}

\newpage

\setcounter{tocdepth}{2}
\tableofcontents
\newpage

\section{Introduction}
Many classical algorithmic problems on graphs can be defined in terms
of finding a subgraph that is isomorphic to a certain pattern
graph. For example, the polynomial-time solvable problem of finding
perfect matchings and the NP-hard \textsc{Hamiltonian Cycle} and \textsc{Clique}
problems arise this way. The goal of the paper is to understand which
pattern graphs make this problem easy with respect to polynomial-time
solvability and polynomial-time preprocessing. 

Given graphs~$G$ and~$H$, \SubgraphTest asks if~$G$ has a subgraph
isomorphic to the pattern~$H$. Observe that, for every fixed pattern graph $H$,
\SubgraphTest is polynomial-time solvable, as we can test each of the
$|V(G)|^{|V(H)|}$ mappings from the vertices of $H$ to the vertices of
$G$, resulting in a polynomial-time algorithm. Therefore,
studying the restrictions of \SubgraphTest to fixed $H$ does not allow
us to make a distinction between easy and hard patterns. We can get a
more useful framework if we restrict \SubgraphTest to a fixed {\em
  class} of patterns.  For every graph class $\F$, let \FSubgraphTest
be the special case of the problem where $H$ is restricted to be in
$\F$. For example, if $\F$ is the set of all matchings (1-regular
graphs), then \FSubgraphTest is the polynomial-time solvable maximum matching
problem; if $\F$ is the set of all cliques, then \FSubgraphTest is the
NP-hard \textsc{Clique} problem. Our goal is to understand which
classes $\F$ make \FSubgraphTest tractable.

We also investigate a well-studied and natural variant of finding
subgraphs. Given graphs $G$ and $H$, and an integer $t$, \Packing asks
if $G$ has $t$ vertex-disjoint subgraphs isomorphic to $H$. Unlike for
\SubgraphTest, now it makes sense to define the problem \HPacking for
a fixed graph $H$: for example, $K_2$-\textsc{Packing} is the
polynomial-time solvable maximum matching problem and
$K_3$-\textsc{Packing} is the NP-hard vertex-disjoint triangle packing
problem. We also define the more general \FPacking problem, where $H$
is restricted to be a member of $\F$.

\textbf{Kernels and Turing kernels.} Besides looking at the
polynomial-time solvability of these problems, we also explore the
possibility of efficient preprocessing algorithms, as defined by the
notion of polynomial kernelization in parameterized complexity
\cite{DowneyF13,FlumG06,LokshtanovMS12}. We can naturally
associate a parameter $k$ to each instance measuring the size of the
solution we are looking for, that is, we define the parameter
$k:=|V(H)|$ for \kSubgraphTest and $k:=t\cdot |V(H)|$ for \kPacking. We
say that a problem with parameter $k$ is {\em fixed-parameter
  tractable (FPT)} if it is solvable in time $f(k)\cdot n^{O(1)}$ for
some computable function $f$. The fixed-parameter tractability of
various cases of \SubgraphTest is a classical topic of the
parameterized complexity literature. It is known that \kFSubgraphTest
is FPT if \F is the set of paths
\cite{AlonYZ95,Bjorklund10,KneisMRR06,DBLP:journals/siamcomp/WilliamsW13}
and, more generally, if \F is a set of graphs of bounded treewidth
\cite{AlonYZ95,DBLP:journals/jcss/FominLRSR12}. The case where~\F is the set of 
all bicliques (complete bipartite graphs), corresponding to the \textsc{Biclique} problem, 
was a tantalizing open problem for many years. In a recent breakthrough result, Binkai Lin~\cite{Lin14} proved
that \textsc{Biclique} is W[1]-hard.

In this paper, we study only a specific aspect of fixed-parameter
tractability. A {\em polynomial (many-one) kernelization} is a
polynomial-time algorithm that creates an equivalent instance whose
size is polynomially bounded by the parameter $k$. Intuitively, a
kernelization is a preprocessing algorithm that does not solve the
problem, but assuming that the parameter value is ``small'' compared
to the size of the input, creates a compact equivalent instance by
somehow getting rid of irrelevant parts of the input. In the case of
\kSubgraphTest, we want to create an equivalent instance with size
bounded by $|V(H)|^{O(1)}$: if the pattern $H$ is small compared to
$G$, we want to compress the instance to a ``hard core'' that has size
comparable to $H$. In recent years, the existence of polynomial
kernelization for various parameterized problems has become a
thoroughly investigated subject. In 2008, Bodlaender et
al.~\cite{BodlaenderDFH09} built on a theorem by Fortnow and Santhanam~\cite{FortnowS11} to introduce the lower bound technology of
OR-compositions, which allows us to show that certain parameterized
problems do not admit polynomial kernels, unless \containment
and the polynomial-time hierarchy collapses to the third level~\cite{Yap83}. In particular, they showed that \kPath (given an undirected graph~$G$ and integer~$k$, does~$G$ contain a simple path of length~$k$?) does not
admit a polynomial kernel under this complexity assumption. This work
has been followed by a flurry of results refining this technology
\cite{BodlaenderJK14,DellM12,DellM10,Drucker12,HermelinW12} and using
it to prove negative results for concrete parameterized problems
(e.g.,
\cite{Binkele-RaibleFFLSV12,BodlaenderFLPST09,BodlaenderTY11,CyganKPPW12,DomLS09,FominLMS12,JansenB11,JansenK13,JansenK12,Kratsch13,KratschPRR12},
see also the recent survey of Lokshtanov et
al.~\cite{LokshtanovMS12}). We continue this line of research by
trying to characterize which \kFSubgraphTest and \kFPacking problems
admit polynomial kernels.

A natural, but less understood variant of kernelization is Turing
kernelization. In a Turing kernelization, instead of creating a single
compact instance in polynomial time, we want to solve the instance in
polynomial time having access to an oracle solving instances of size
$k^{O(1)}$ in constant time. This form of kernelization can be also
thought of as some kind of preprocessing: we want to spend polynomial
time to preprocess the instance in such a way that the time-consuming
part of the work needs to be done on compact instances. While Turing
kernelization may seem much more powerful than many-one kernels, there
are only a handful of examples where Turing kernelization is possible,
but many-one kernelization is not~\cite{AmbalathBHKMPR10,Binkele-RaibleFFLSV12,Jansen14,SchaferKMN12,ThomasseTV13}. On the other hand, the lower bound technology introduced
by Fortnow and Santhanam \cite{FortnowS11} and Bodlaender et
al.~\cite{BodlaenderDFH09} {\em does not} say anything about the
possibility of Turing kernels and therefore we know very little about
the limits of Turing kernelization. In fact, even the basic question
whether \kPath admits a Turing kernel is open (cf.~\cite{Jansen14}). Hermelin et
al.~\cite{HermelinKSWW13} tried to deal with this situation by
developing a completeness theory based on certain fundamental
satisfiability problems that can be shown to be fixed-parameter
tractable by simple branching argument, but for which the existence of polynomial
(Turing) kernels is unlikely. They introduced the notion of
\textup{WK[1]}-hardness, which can be interpreted as evidence that the problem
is unlikely to admit a polynomial Turing kernel.\footnote{It is known~\cite[Lemma 2]{HermelinKSWW13} that the existence of a polynomial-size many-one kernel for a \textup{WK[1]}-hard problem implies \containment.}
Unfortunately,
Hermelin et al.~\cite{HermelinKSWW13} were unable to prove any hardness
result for \kPath; its \textup{WK[1]}-hardness remains an open
question. In this paper, we are working under the assumption that
\kPath admits no polynomial Turing kernel and interpret the existence of a polynomial-parameter transformation from \kPath to our problem as evidence for the nonexistence of polynomial Turing kernels. Problems for which such a transformation exists will be called \emph{\kPath-hard}.

\textbf{Our results.}  In this paper, we restrict our study of
\FPacking and \FSubgraphTest to hereditary classes $\F$, that is, to
classes that are closed under taking induced subgraphs.

The polynomial-time solvability of \HPacking is
well understood: if every component of $H$ has at most two vertices,
then it is a matching problem (hence polynomial-time solvable) and
Kirkpatrick and Hell \cite{KirkpatrickH78} proved that \HPacking is
NP-hard for every other $H$. It follows that \FPacking is
polynomial-time solvable if every component of every graph in $\F$ has
at most two vertices, and is NP-hard otherwise.
For every fixed $H$, we can formulate \HPacking as a
special case of finding $t$ disjoint sets of size $|V(H)|$ each. Hence the problem 
admits a polynomial kernel of size $t^{O(|V(H)|)}$ using, for example, standard sunflower
kernelization arguments~\cite[Appendix A]{DellM12}. However, the exponent of the bound on the
kernel size depends on the size of~$H$. Therefore, it does not follow that
\FPacking admits a polynomial kernel for every fixed class $\F$, as
$\F$ may contain arbitrarily large graphs.

Our first result characterizes those hereditary classes $\F$ for which
\FPacking admits a polynomial kernel. Interestingly, it seems that
Turing kernels are not more powerful for this family of problems: we
get the same positive and negative cases with respect to both
notions. Let us call a connected bipartite graph {\em $b$-thin} if the
smaller partite class has size at most $b$. We say that a graph $H$
is {\em $a$-small/$b$-thin} if every component of $H$ either has at
most~$a$ vertices, or is a~$b$-thin bipartite graph (we emphasize that
it is possible that $H$ has components of both types).
A graph class $\F$ is small/thin if there are $a,b\ge 0$ such that every graph in $\F$ is $a$-small/$b$-thin.
\begin{restatable}{\reabctheorem}{restateintropacking}\label{theorem:intro:packing}
  Let $\F$ be a hereditary class of graphs.  If $\F$ is small/thin, then
  \kFPacking admits a polynomial (many-one) kernel. If $\F$ does
  not have this property, then \kFPacking admits no polynomial kernel,
  unless \containment, and moreover it is also \textup{WK[1]}-hard, \textup{W[1]}-hard, or \kPath-hard.
\end{restatable}
Theorem~\ref{theorem:intro:packing} gives a complete characterization
of the hereditary families for which \kFPacking admits a polynomial
kernel. It is well known that many problems related to packing small graphs/objects
admit polynomial kernels (most of the research is therefore on
understanding the exact degree of the polynomial bound \cite{Abu-Khzam10,ChenLSZ07,DellM12,HermelinW12,Moser09}), but we are
not aware of any previous result showing that thin bipartite graphs have
similar good properties. This revelation about thin bipartite graphs
highlights the importance of looking for dichotomy theorems such as
Theorem~\ref{theorem:intro:packing}: while proving a complete
characterization of the positive and negative cases, we {\em
  necessarily} have to uncover all the important algorithmic ideas
relevant to the family of problems we study. Indeed, our goal was not
to prove a result specific to the kernelization of thin bipartite graphs,
but it turned out that one cannot avoid proving this result in a complete
characterization. The negative part of
Theorem~\ref{theorem:intro:packing} shows that these two algorithmic
ingredients (handling small components and thin bipartite graphs)
cover all the relevant algorithmic ideas and {\em
  any} hereditary class $\F$ that cannot be handled by these ideas
leads to a hard problem.

For \FSubgraphTest, we first prove a dichotomy theorem characterizing
the randomized polyno\-mial-time solvable and NP-hard cases. We say that
$\F$ is {\em matching-splittable} if there is a constant $c$ such that
every $H\in \F$ has a set $S$ of at most $c$ vertices such that every
component of $H- S$ has at most 2 vertices.
\begin{restatable}{\reabctheorem}{restateintropolysubgraph}\label{theorem:intro:polysubgraph}
  Let $\F$ be a hereditary class of graphs. If $\F$ is matching-splittable,
  then \FSubgraphTest can be solved in randomized polynomial time. If
  $\F$ does not have this property, then \FSubgraphTest is NP-hard.
\end{restatable}
The reason why randomization appears in
Theorem~\ref{theorem:intro:polysubgraph} is the following. Given
graphs $G$ and $H\in \F$, first we try every possible location where
the set $S\subseteq V(H)$ can appear in $V(G)$ in a solution; as
$|S|\le c$, there are $|V(G)|^c$ possibilities to try. Having fixed the
location of $S$, we need to locate every component of $H-S$. As each such component is an edge or a single vertex, this looks
like a matching problem, but here we have an additional restriction on
how the endpoints of the edges should be attached to $S$. We can
encode these neighborhood conditions using a bounded number of colors
and get essentially a colored matching problem, which can be solved in
randomized polynomial time using the algorithm of Mulmuley, Vazirani,
and Vazirani \cite{MulmuleyVV87} for finding perfect matchings of
exactly a certain weight. The negative side of
Theorem~\ref{theorem:intro:polysubgraph} can be obtained by observing
(using an application of Ramsey arguments) that if $\F$ is not
matching-splittable, then $\F$ contains all cliques,
all bicliques, all disjoint unions of triangles, or all disjoint unions of length-two paths; in each case, the problem is NP-hard. The authors are somewhat
puzzled that the clean characterization of
Theorem~\ref{theorem:intro:polysubgraph} has apparently not been
observed so far in the literature: it is about the classical question
of polynomial-time solvability of finding subgraphs and the proof uses
techniques that are decades old. We may attribute this to the fact
that while dichotomy theorems for fixed classes \F of graphs exist 
(e.g.,
\cite{DBLP:conf/icalp/ChenTW08,DBLP:journals/tcs/DalmauJ04,1206036,380867,DBLP:journals/tcs/KhotR02,LewisY80,DBLP:journals/siamcomp/Yannakakis81a}),
perhaps it is not yet widely realized that such results are possible
and aiming for them is a doable goal. We hope our paper contributes to
the more widespread recognition of the feasibility of this line of
research.

In Theorem~\ref{theorem:intro:packing}, we have observed that Turing
kernels are not more powerful than many-one kernels for \kFPacking. The
situation is different for \kFSubgraphTest: there are classes $\F$ for
which \kFSubgraphTest admits a polynomial Turing kernel, but has no
polynomial many-one kernel, unless \containment. We characterize the
classes $\F$ that admit polynomial Turing kernels the following way. We
say that a graph $H$ is $(a,b,c,d)$-splittable, if there is a set
$S$ of at most $c$ vertices such that every component of $H-S$ either has size at most $a$ or is a~$b$-thin bipartite graph with the
additional restriction that the closed neighborhoods of all but $d$
vertices are universal to~$S$ (see Section~\ref{sec:char-prop} for details).
\begin{restatable}{\reabctheorem}{restatemainsubgraph}
\label{theorem:intro:turingsubgraph}
  Let $\F$ be a hereditary class of graphs. If there are $a,b,c,d\ge
  0$ such that every $H\in \F$ is $(a,b,c,d)$-splittable, then
  \kFSubgraphTest admits a polynomial Turing kernel. If $\F$ does not
  have this property, then \kFSubgraphTest is \textup{WK[1]}-hard, \textup{W[1]}-hard, or
  \kPath-hard.
\end{restatable}
In the algorithmic part of Theorem~\ref{theorem:intro:turingsubgraph},
the first step is to guess the location of the set $S\subseteq V(H)$
in $V(G)$, giving $|V(G)|^c$ possibilities (this is the reason why in
general our Turing kernel is not a many-one kernel). For each guess, locating the
components of $H- S$ in $G$ is similar to
Theorem~\ref{theorem:intro:packing}, as we have to handle small
components and thin bipartite components, but here we have the
additional technicality that we have to ensure that these components
are attached to $S$ in a certain way.

For many-one kernels, we do not have a characterization similar to
Theorem~\ref{theorem:intro:turingsubgraph}.  We present some concrete
positive and negative results showing that a complete characterization
of \kFSubgraphTest with respect to many-one kernels would be much more
delicate than Theorem~\ref{theorem:intro:turingsubgraph}. The simple
algorithmic idea used in Theorem~\ref{theorem:intro:turingsubgraph},
guessing the location of $S$, fails for many-one kernels
and it seems that we have to make extreme efforts (whenever it is
possible at all) to replace this step with adhoc arguments.

\textbf{Our techniques.} The proofs of
Theorems~\ref{theorem:intro:packing}--\ref{theorem:intro:turingsubgraph}
all follow the same pattern. First, we define a certain
graph-theoretic property and devise an algorithm for the case when
$\F$ has this property.  As described above, the algorithmic part of
Theorem~\ref{theorem:intro:polysubgraph} is based on the randomized
matching algorithm of Mulmuley, Vazirani, and Vazirani
\cite{MulmuleyVV87}.  For Theorems~\ref{theorem:intro:packing} and
\ref{theorem:intro:turingsubgraph}, the algorithm is a marking
procedure: for each component, we mark a bounded number of vertices
such that we can always find a copy of this component using only these
vertices even if the other components already occupy an unknown but small set of
vertices.  Therefore, if there is a solution, then
there is a solution using only this set of marked vertices. The kernel is
obtained by restricting the graph to this set of vertices. For small
components, we use the Sunflower Lemma of Erd\H os and Rado
\cite{MR22:2554} (similarly as it is used in the kernelization of
other packing problems, cf.~\cite{DellM12}). For thin bipartite graphs,
the marking procedure is a branching algorithm specifically designed
for this class of graphs. At some point in the algorithm, we crucially
use that the component is $b$-thin: we find a biclique with
$b$ vertices on one side and many vertices on the other side, and then
we argue that the component is a subgraph of this biclique.

For the hardness results of
Theorems~\ref{theorem:intro:packing}--\ref{theorem:intro:turingsubgraph},
first we prove that if $\F$ does not have the stated property, then
$\F$ contains every graph from one of the basic families of hard
graphs. These hard families include cliques, bicliques, paths, odd
cycles with a high-degree vertex, and subdivided stars (see Section~\ref{sec:hard-families}). To prove that
a hard family appears in $\F$, we use Ramsey results (including a
recent path vs.~induced path vs.~biclique result of Atminas, Lozin,
and Razgon~\cite{AtminasLR12}) and a graph-theoretic analysis of what,
for example, a large nonbipartite graph without large cliques and long
induced paths can look like.  For each hard family, we then claim a
lower bound on the problem. Most of these lower bounds take the form
of a relatively standard polynomial-parameter transformation from \textsc{Set Cover} parameterized by the size of the universe; here the
value of our contribution is not in the details of the 
reduction, but in realizing that these are the hard families of graphs
whose hardness exhaustively explain the hard cases of the
problem. 

The basic technique to obtain negative evidence for the existence of
many-one kernels is the method of \emph{OR-cross-composition}~\cite{BodlaenderJK14}, which refines the original OR-composition framework~\cite{BodlaenderDFH09}. An OR-cross-composition of a classical problem~$L$ into a parameterized problem~$\Q$ is a polynomial-time embedding of a series of~$t$ length-$n$ instances~$x_1, \ldots, x_t$ of~$L$ into a single instance~$x^*$ of~$\Q$ with parameter value poly($n$), such that~$x^* \in \Q \Leftrightarrow \bigvee _{i=1}^t x_i \in L$. If~$L$ is NP-hard, such a construction is known to rule out the existence of polynomial kernels for \Q under the assumption that \ncontainment.  
The negative results that we present for the 
existence of many-one kernels for \FSubgraphTest use a specific form
of this technique that we name {\em OR-cross-composition by reduction with a
  canonical template.} The idea is to start from an NP-hard graph problem~$L$ for which a family of polynomial-size \emph{canonical template graphs} exists, such that for every~$n$, the instances of length~$n$ are induced subgraphs of the $n$-th graph in this family. This allows length-$n$ inputs~$x_1, \ldots, x_t$ to be merged into one through their common canonical supergraph of size poly($n$), as opposed to the trivial~$t \cdot n$, which facilitates an OR-cross-composition. Canonical template graphs were first used for this purpose by Bodlaender et al.~\cite[Theorem 11]{BodlaenderJK12c} to prove a kernel lower bound for a structural parameterization of \textsc{Path with Forbidden Pairs}.


\section{Outline}\label{sec:outline}
In this section we present a more detailed overview of the results of the paper. We also describe the main technical parts of the proofs. The proofs of Theorems~\ref{theorem:intro:packing}--\ref{theorem:intro:turingsubgraph} all follow the same pattern:
\begin{enumerate*}
\item We define the property separating the positive and negative cases.
\item We prove an algorithmic result for the positive cases.
\item We prove a purely combinatorial result stating that if a class $\F$ does not satisfy the property, then $\F$ is a superset of one of the classes 
appearing on a short list of basic hard classes.
\item We prove a hardness result for each basic hard class on the list.
\end{enumerate*}
The structure of this section follows these steps: for each step, we
go through the relevant definitions and state the results proved later in
the paper.

\subsection{Characterizing properties}
\label{sec:char-prop}We say that a graph is {\em $c$-matching-splittable} if there is a set
$S\subseteq V(H)$ of at most $c$ vertices such that every component of
$H- S$ has at most two vertices. We say that a class $\F$
of graphs is $c$-matching-splittable if every $H\in \F$ has this
property, and we say that $\F$ is matching-splittable if $\F$ is
$c$-matching-splittable for some $c\ge 0$. In
Theorem~\ref{theorem:intro:polysubgraph}, this is the condition for
randomized polynomial-time solvability. Clearly, a matching is
0-matching-splittable and a matching plus a universal vertex is
1-matching-splittable. On the other hand, the class containing the
disjoint unions of arbitrarily many triangles is not
$c$-matching-splittable for any $c\ge 0$, as $S$ would need to contain
at least one vertex from each triangle.

In Theorem~\ref{theorem:intro:packing}, the condition that we need is
that every component is either small or a thin bipartite graph. We say
that a graph $H$ is {\em $a$-small/$b$-thin} if every component of $H$
has at most~$a$ vertices or is a $b$-thin bipartite graph (that is, a bipartite graph with one of the partite classes having size at most $b$). Note that
$H$ can have both types of components. For example, if $H$ is the
disjoint union of an arbitrary number of triangles and stars of
arbitrary size, then it is 3-small/1-thin. We say that class $\F$ is
$a$-small/$b$-thin if every graph $H\in \F$ has this property and say
that $\F$ is small/thin if it is $a$-small/$b$-thin for some $a,b\ge
0$.
The characterization property that we need for
Theorem~\ref{theorem:intro:turingsubgraph} is a somewhat technical generalization of being $a$-small/$b$-thin.
\begin{definition}\label{def:splittable}
  We say that a graph $H$ is \emph{$(a,b,c,d)$-splittable} if it has a 
  vertex set~$S \subseteq V(H)$ of size at most~$c$ such that:
\begin{enumerate}
	\item each connected component of~$H - S$ on more than~$a$ vertices is bipartite and has a partite class of size at most~$b$, and
	\item in each connected component~$C$ of~$H - S$, the number of vertices whose closed neighborhood in~$G[C]$ is \emph{not} universal to~$N_H(C) \cap S$ is at most~$d$.
\end{enumerate}
We say that such a set~$S \subseteq V(H)$ \emph{realizes}
the~$(a,b,c,d)$-split of~$H$. Family~\F is $(a,b,c,d)$-splittable if
every $H\in \F$ is a $(a,b,c,d)$-splittable. Family~\F is splittable
if there are constants~$a,b,c,d$ such that~\F is
$(a,b,c,d)$-splittable.
\end{definition}

\noindent Observe that being $a$-small/$b$-thin is exactly the same as being
$(a,b,0,0)$-splittable and being $c$-matching-splittable is exactly
the same as being $(2,0,c,2)$-splittable. We prefer to use
the terms $a$-small/$b$-thin and $c$-matching-splittable for these
special cases, as they are more descriptive.

Given an $a$-small/$b$-thin
graph $H$, adding a set $S$ of $c$ universal vertices results in an
$(a,b,c,0)$-splittable graph $H'$. If $C$ is a component of $H$ having
at most $a$ vertices and we remove from $H'$ any set of edges between
$C$ and $S$, then the resulting graph is $(a,b,c,a)$-splittable. The closed neighborhoods of the $a$ vertices in~$C$ may no longer be universal to~$N_H(C) \cap S$ after the edge removals, which is compensated by the fourth entry in the tuple. Let now $C$ be
a $b$-thin bipartite component of $H$, let $A$ be the smaller side and
let $B$ be the larger side of $C$. Observe that Definition~\ref{def:splittable} not only requires that all but $d$ vertices of $C$ 
are universal to $N_H(C)\cap S$, but even the closed neighborhoods in~$G[C]$ have to be universal. Therefore, removing even a single edge between a vertex $v$ of $C$ and $S$ can ruin the property, as it ``contaminates'' all the neighbors of $v$.
If we remove a single edge between
some $x\in B$ and $S$, then the graph is still
$(a,b,c,b+1)$-splittable: there are at most $b+1$ vertices in $C$
whose neighborhood is not universal to $S$, namely $x$ and some of the
vertices of $A$. On the other hand, if we remove a single edge between
some $y\in A$ and $S$, then the graph may not be
$(a,b,c,d)$-splittable for arbitrary large $d$: if $y$ has degree $d$,
then $y$ and all its neighbors have the property that their closed
neighborhoods are not universal to $S$. 

Note that the definition {\em does not} require that the closed
neighborhood of (all but $d$ of) the vertices are universal to $S$, it
requires universality only to $N_H(C)\cap S$. Suppose that $H_1$
and $H_2$ are two graphs with $S_i$ realizing an
$(a,b,c,d)$-split of $H_i$ for $i=1,2$. The disjoint union of
$H_1$ and $H_2$ is $(a,b,2c,d)$-splittable, as realized by $S_1\cup
S_2$: the vertices in a $b$-thin component of $H_1$ need to be
universal only to (a certain part of) $S_1$, as $C$ has no edge to $S_2$.
 
\subsection{Algorithms}
In the algorithmic part of Theorem~\ref{theorem:intro:polysubgraph},
we need to solve \SubgraphTest in the case that~$H$ is $c$-matching-splittable for some set $S$ of at most $c$ vertices. As
described in the introduction, we guess the location of $S$ and then solve
the resulting constrained matching problem. The main technical engine
in the algorithm is the classic algebraic matching algorithm due to Mulmuley, Vazirani, and Vazirani~\cite{MulmuleyVV87}. It can be used to obtain randomized algorithms for various colored versions of matching (see, for example, \cite{DBLP:journals/ipl/Marx04,MarxP14}). We need the following variant.
\begin{restatable}{\retheorem}{restatecoloredmatching}
\label{theorem:coloredmatching}
  Given a multigraph $G$ with a (not necessary proper) coloring of the
  edges with a set $C$ of colors and function $f:C\to \mathbb{Z}^+$,
  there is a randomized algorithm with false negatives that decides in
  time $(|V(G)|+|E(G)|)^{O(|C|)}$ if $G$ has a matching containing exactly $f(i)$
  edges of color $i$ for every $i\in C$.
\end{restatable}

By a randomized algorithm with false negatives, we mean an algorithm that is always correct on \no-instances, but which may incorrectly reject a \yes-instance with probability at most~$\frac{1}{2}$. Equipped with Theorem~\ref{theorem:coloredmatching}, we prove the algorithmic part of Theorem~\ref{theorem:intro:polysubgraph} in Section \ref{section:polynomialvsnpcomplete:upperbounds}.
\begin{restatable}{\retheorem}{restatesubgraphtestalg}
\label{theorem:fsubgraphtest:alg}
  \FSubgraphTest is (randomized) polynomial-time solvable if $\F$ is matching-splittable.
\end{restatable}

The polynomial kernel in the positive part of
Theorem~\ref{theorem:intro:packing} is obtained by a marking procedure
that finds a polynomially bounded subset of vertices in $G$ that
surely contains a solution, if a solution exists at all. Let us first
explain briefly how the standard technique of sunflowers can be used for this
marking procedure if every component of $H$ has at most $a$ vertices. 
We need the Sunflower Lemma of Erd\H os and Rado \cite{MR22:2554}. A
collection $\S$ of sets is called a {\em sunflower} if the pairwise
intersection $S_1\cap S_2$ is the same set $C$ for any two distinct
$S_1,S_2\in \S$. Then this intersection $C$ is the {\em core} of the
sunflower; the sets $S\setminus C$ for $S\in \S$ are the {\em petals}
of the sunflower.
\begin{lemma}[{\cite{MR22:2554}, cf.~\cite[Lemma 9.7]{FlumG06}}] \label{lemma:sunflowers}
Let $k$ and $m$ be nonnegative integers and let~$\S$ be a system of sets of size at most~$m$ over a universe~$U$. If~$|\S| \geq m!(k-1)^{m}$, then there is a sunflower in~$\S$ with~$k$ petals. Furthermore, for every fixed~$m$ there is an algorithm that computes such a sunflower in time polynomial in~$(k + |\S|)$.
\end{lemma}

Let $H$, $G$, and $t\ge 1$ form an instance of \Packing; the solution
we are looking for has $k:=t\cdot |V(H)|$ vertices. Let $C$ be a
component of $H$ having size at most $a$. First, we enumerate every
subset of $|V(C)|\le a$ vertices in $G$ where $C$ appears; the length of this
list is polynomial in the size of $G$ if $a$ is a fixed constant. We
would like to reduce the length of this list: we would like to have a
shorter list of candidate locations where $C$
can appear in a solution, such that the length of the list is polynomially bounded in $k$. We argue the following way. As long as the
length of the list is at least $a!(k+1)^a$, we can find a sunflower
with $k+2$ petals among the sets in the list. We claim that we can
choose any set $S$ from this sunflower and throw it out of the
list. Suppose that there is a solution where the component $C$ is mapped exactly to this set $S\subseteq V(G)$. As the solution uses only $k$
vertices of $G$ and the petals of the sunflower are disjoint, there is
another set $S'$ among the remaining $k+1$ sets of the sunflower whose
petal is disjoint from the solution. Therefore, we can modify the
solution such that $C$ is mapped to $S'$ instead of $S$, which means
that the set $S$ cannot be essential to the solution and can be safely
removed from the list of candidate locations for $C$. Repeating this
argument, we eventually get a list of at most $a!(k+1)^a$ candidate
locations for each component of $H$, thus we can reduce the problem to
an induced subgraph of $G$ whose size, for a fixed constant $a$, is
polynomial in $k$.

If $H$ has $b$-thin components, then the Sunflower Lemma cannot be
applied, as the size of such a component can be arbitrarily large (and
it is the size of the component that appears in the exponent in the
argument above). Therefore, in Section~\ref{sec:repr-sets-thin}, we
develop a marking procedure specifically designed for thin bipartite
graphs. As an illustration, we present here the main idea on the
special case of packing thin bicliques, that is, on graphs
$K_{b,\ell}$ for some fixed $b\ge 1$. The crucial ingredient for the
kernel for biclique packing is the following lemma.

\begin{lemma} \label{lemma:bicliquemarking:specialcase}
For every fixed~$b$ there is a polynomial-time algorithm that, given a graph~$G$ and integers~$\ell > b$ and~$k \geq \ell + b$, computes a set~$X$ of size~$\Oh(k^{4b})$ such that for every~$Z \subseteq V(G)$ of size at most~$k$, if~$G - Z$ contains a $K_{b,\ell}$ subgraph, then~$G[X] - Z$ contains a $K_{b,\ell}$ subgraph.
\end{lemma}

Before proving the lemma, we show how it leads to a polynomial kernel for biclique packing. To reduce the size of an instance that asks whether~$G$ contains~$t$ disjoint $K_{b,\ell}$ subgraphs for~$\ell > b$, we define~$k := t \cdot (b + \ell)$ and invoke the lemma to compute a set~$X$ of size~$\Oh(k^{4b})$. We then output~$G[X]$ as the kernelized instance. If~$G$ contains a packing of~$t$ disjoint $K_{b,\ell}$ subgraphs, then while the packing contains a biclique~$C$ using a vertex in~$V(G) \setminus X$, we let~$Z$ be the~$(t-1)(b + \ell)$ other vertices in the packing, apply the guarantee of the lemma to find a biclique model~$C'$ in~$G[X]$ avoiding~$Z$, and replace~$C$ in the packing by~$C'$. Iterating the argument results in a packing of bicliques in~$G[X]$, proving that the reduced instance is equivalent to the original one. 

To facilitate a recursive algorithm, we actually prove a generalization of Lemma~\ref{lemma:bicliquemarking:specialcase}. To state the generalization we need the following terminology. For disjoint sets~$A', B' \subseteq V(G)$ and~$\ell > b$ we say that a~$K_{b,\ell}$ subgraph in~$G$ extends~$(A',B')$ if the side-$b$ partite class is a superset of~$A'$ and the size-$\ell$ partite class is a superset of~$B'$. 

\begin{lemma} \label{lemma:bicliquemarking:specialcase:general}
For every fixed~$b$ there is a polynomial-time algorithm that, given a graph~$G$, integers~$\ell > b$ and~$k \geq \ell + b$, and disjoint sets~$A', B' \subseteq V(G)$ of size at most~$b$, computes a set~$X$ of size at most~$(3k^2)^{2b - |A' \cup B'|}$ such that for every~$Z \subseteq V(G)$ of size at most~$k$, if~$G - Z$ contains a $K_{b,\ell}$ subgraph that extends~$(A',B')$, then~$G[X] - Z$ contains a $K_{b,\ell}$ subgraph.
\end{lemma}
\begin{proof}
The main idea behind the algorithm is to make progress in recursive calls by increasing the size of~$A' \cup B'$, thereby restricting the type of bicliques that have to be preserved in the set~$X$. Throughout the proof we use the fact that if~$Z \subseteq V(G)$ and there is a $K_{b,\ell}$-subgraph in~$G - Z$ that extends~$(A',B')$, then~$Z \cap (A' \cup B') = \emptyset$, the size-$b$ partite class consists of common neighbors of~$B'$, while the size-$\ell$ partite class consists of common neighbors of~$A'$. Let us point out that the lemma requires that $G[X]-Z$ contains an $K_{b,\ell}$-subgraph, but it does not require it to extend $(A',B')$.

\textbf{Case 1.} If~$|A'| = b$, then we choose~$X$ as~$A' \cup B'$ together with~$k + \ell$ common neighbors of~$A'$ (or less, if there are fewer), for a total size of at most~$2b + (k + \ell) \leq 3k$. Let~$Z \subseteq V(G)$ have size at most~$k$. If there is a $K_{b,\ell}$-subgraph~$H$ in~$G - Z$ that extends~$(A',B')$, then all vertices in the size-$\ell$ partite class are common neighbors of~$A'$. If all vertices of~$H$ are contained in~$G[X]$, then the biclique subgraph~$H$ also exists in~$G[X]-Z$. If not, then the set~$A'$ had at least~$k + \ell$ common neighbors (otherwise they were all preserved in~$X$). Since~$Z$ contains at most~$k$ of them, any~$\ell$ of the remaining vertices in~$X$ combines with~$A'$ to form a~$K_{b,\ell}$-subgraph in~$G[X] - Z$.

\textbf{Case 2.a.} If~$|B'| = b$,~$|A'| < b$, and the set~$B'$ has at least~$k + \ell$ common neighbors, then we choose~$X$ containing~$k+\ell$ of these common neighbors together with~$B'$ itself. For any~$Z \subseteq V(G)$ of size at most~$k$, if a biclique extending~$(A',B')$ exists in~$G - Z$ then~$Z$ avoids at least~$\ell$ common neighbors of~$B'$ in~$X$. Together with~$B'$, these form a~$K_{b,\ell}$ subgraph in~$G[X] - Z$. Note that this $K_{b,\ell}$ does not extend $(A',B')$, but this is not required by the lemma.

\textbf{Case 2.b.} If~$|B'| = b$,~$|A'| < b$, and the set~$B'$ has less than~$k + \ell \leq 2k$ common neighbors~$T := \bigcap _{v \in B'} N_G(v)$, then a $K_{b,\ell}$-subgraph extending~$(A',B')$ has its size-$b$ side within~$T$. For each~$a \in T \setminus (A' \cup B')$, add~$a$ to~$A'$ and recurse. Let~$X$ be the union of the recursively computed sets. If there is a biclique in~$G - Z$ extending~$(A',B')$, then there is an~$a \in T \setminus (A' \cup B')$ such that it extends~$(A' \cup \{a\}, B')$, and the correctness guarantee for that recursive call yields a biclique in~$G[X] - Z$. The measure~$2b - |A' \cup B'|$ drops in each recursive call and we recurse on at most~$2k$ instances, giving a bound of~$2k \cdot (3k^2)^{2b - |A' \cup B'| - 1} \leq (3k^2)^{2b - |A' \cup B'|}$ on~$|X|$.

\textbf{Case 3.} In the remaining cases we have~$|A'|, |B'| < b$. We greedily compute a maximal set of $K_{b,\ell}$ subgraphs that extend~$(A',B')$ and pairwise intersect only in~$A' \cup B'$. Since~$b$ is constant, this can be done in polynomial time by guessing all possible locations for the remaining vertices in the size-$b$ partite class and testing whether the resulting vertices are adjacent to~$B'$ and have sufficient common neighbors to realize the other partite class. Two things can happen.

\textbf{Case 3.a.} If we find~$k + 1$ distinct $K_{b,\ell}$ subgraphs that pairwise intersect only in~$(A',B')$, then we output~$X$ containing the union of these subgraphs, which has size at most~$(k+1)(\ell + b) \leq 2k^2$. If a $K_{b,\ell}$-subgraph extending~$(A',B')$ exists in~$G - Z$ for some~$Z \subseteq V(G)$ of size~$k$, then~$Z$ intersects at most~$k$ of the extensions. Hence one extension avoids~$Z$ and combines with~$A',B'$ to form a~$K_{b,\ell}$-subgraph in~$G[X] - Z$.

\textbf{Case 3.b.} If there are at most~$k$ of such extensions, then let~$T$ contain the at most~$k (\ell + b) \leq k^2$ vertices in their union. By the maximality of the packing, any extension of~$(A',B')$ uses a vertex in~$T \setminus (A' \cup B')$. For each~$v \in T \setminus (A' \cup B')$, recurse twice: once for adding~$v$ to~$A'$ and once for adding~$v$ to~$B'$. We let~$X$ be the union of the recursively computed sets. If there is a $K_{b,\ell}$ subgraph in~$G - Z$ for some~$Z \subseteq V(G)$ of size at most~$k$, then it extends~$(A' \cup \{v\}, B')$ or~$(A', B' \cup \{v\})$ for some~$v \in T \setminus (A' \cup B')$. The correctness guarantee for that branch of the recursion guarantees the existence of~$K_{b,\ell}$ in~$G[X] - Z$. As the measure~$2b - |A' \cup B'|$ drops in each recursive call, while we branch in at most~$2|T| \leq 2k(\ell + b) \leq 2k^2$ directions, the size of~$X$ is bounded by~$2k^2 \cdot (3k^2)^{2b - |A' \cup B'| - 1} \leq (3k^2)^{2b - |A' \cup B'|}$.
\end{proof}

The generalization from $b$-thin bicliques to general $b$-thin
bipartite graphs makes the scheme described above much more
technical. Let us point out that the large side of a $b$-thin
bipartite graph can be partitioned into at most $2^b$ classes
according to its neighborhood in the small side. Therefore,
intuitively, a $b$-thin bipartite graph can be seen as $2^b$ different
$b$-thin bicliques joined together, which makes it plausible that such
a generalization exists.

\begin{restatable}{\retheorem}{restatepackingkernel}\label{theorem:kernel:packing}
  If $\F$ is a hereditary class of graphs that is small/thin, then
  \FPacking admits a polynomial many-one kernel.
\end{restatable}

For the algorithmic part of
Theorem~\ref{theorem:intro:turingsubgraph}, we have to guess the
location of the set $S$ realizing the $(a,b,c,d)$-split and then take
into account the universality restrictions. This introduces another
layer of technical difficulties, but no new conceptual ideas are
needed. Moreover, because of this guessing step, the kernel is no longer
many-one, but it is a Turing kernel.
\begin{restatable}{\retheorem}{restatesubgraphturing} \label{theorem:kernel:subgraph}
  If $\F$ is a hereditary class of graphs that is splittable, then
  \FSubgraphTest admits a polynomial Turing kernel.
\end{restatable}

\subsection{Hard families}
\label{sec:hard-families}
We define several specific classes of graphs and show hardness results for these classes. Then we show that if a class does not have the property of, say, being splittable, then it is a superset of at least one hard class, hence hardness follows for every class that does not have this property.

First, we define the following graphs (see Figure~\ref{fig:basic} on page~\pageref{fig:basic}).
\begin{itemize}[noitemsep]
	\item $\Gpath{\ell}$ is the path of length $\ell$, which consists of~$\ell$ edges and ~$\ell + 1$ vertices. It is sometimes denoted~$P_{\ell+1}$ for brevity.
        \item $\Gclique{n}$ is the clique on~$n$ vertices (while describing hard families, we use $\Gclique{n}$ instead of the more standard $K_n$ for consistency of notation).
        \item $\Gbiclique{n}$ is the balanced biclique~$K_{n,n}$ on $n+n$
          vertices.
	\item $\Gdoublebroom{s}{n}$ is obtained from a length-$s$ path by adding $n$ pendant vertices to each of the two endpoints of the path.
	\item $\Goperahouse{s}{n}$ is obtained from a length-$s$ path by adding $n$ vertices that are adjacent to both endpoints of the path.
	\item $\Gfountain{s}{n}$ is obtained from a length-$s$ cycle by adding $n$ pendant vertices to one vertex on the cycle.
	\item $\Glongfountain{s}{t}{n}$ is obtained from a length-$s$ cycle by adding a path of length~$t$, identifying one endpoint with a vertex on the cycle and adding $n$ pendant vertices to the other endpoint.
	\item $\Gsubdivstar{n}$ is obtained from a star with $n$ leaves by subdividing each edge once.
	\item $\Gsubdivtree{s}{n}$ is obtained from a star with $n$ leaves by subdividing each edge $s-1$ times and attaching $n$ pendant vertices to each leaf.
	\item $\Gdiamondfan{n}$ is obtained from $n$ copies of $K_{2,n}$ by taking one degree-$n$ vertex from each copy and identifying them into a single vertex.
        \end{itemize}
\begin{figure}[t]
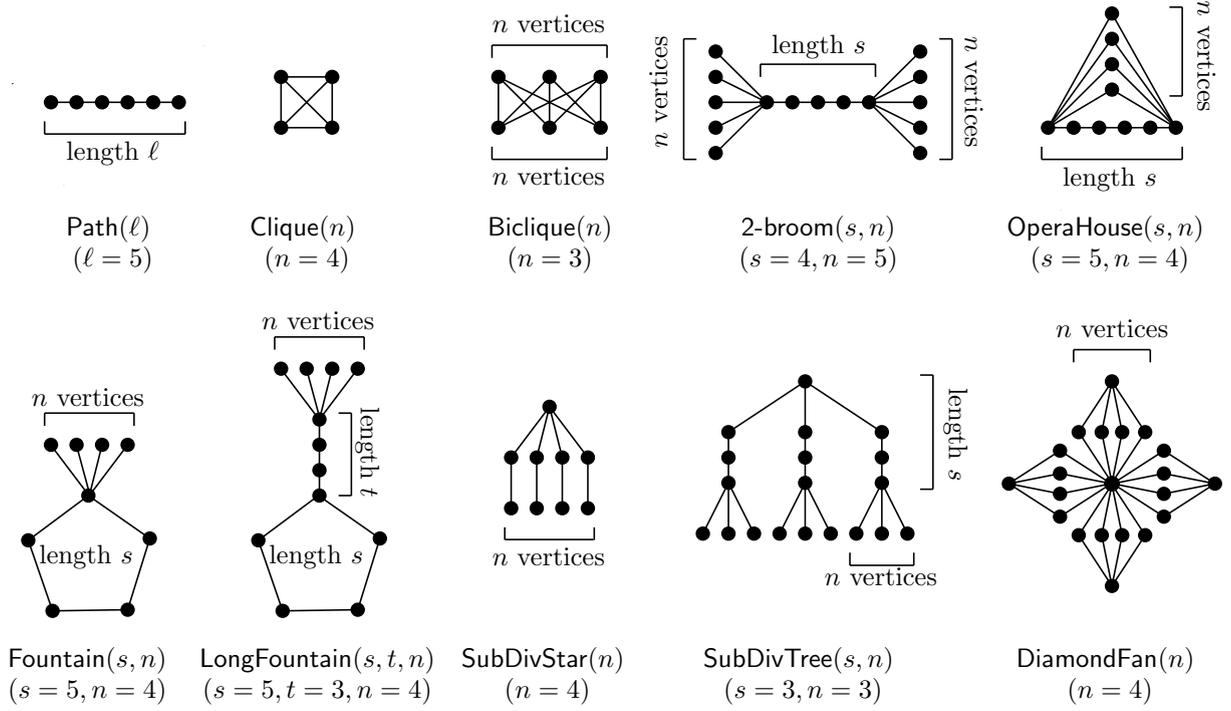

\begin{center}
{\small  \svg{\linewidth}{families}}
\end{center}
\caption{Basic families of graphs.}\label{fig:basic}
\end{figure}

We can define families of these graphs the obvious way:
\[
\begin{array}{ll}
\Fpath=\{\Gpath{i}\mid i\ge 1\}&
\Fclique=\{\Gclique{i}\mid i\ge 1\}\\
\Fbiclique=\{\Gbiclique{i}\mid i\ge 1\}&
\Fdoublebroom{s}=\{\Gdoublebroom{s}{i}\mid i\ge 1\}\\
\Ffountain{s}=\{\Gfountain{s}{i}\mid i\ge 1\}&
\Flongfountain{s}{t}=\{\Glongfountain{s}{t}{i}\mid i\ge 1\}\\
\Foperahouse{s}=\{\Goperahouse{s}{i}\mid i\ge 1\}&
\Fsubdivstar=\{\Gsubdivstar{i}\mid i\ge 1\}\\
\Fsubdivtree{s}=\{\Gsubdivtree{s}{i}\mid i\ge 1\}&
\Fdiamondfan=\{\Gdiamondfan{i}\mid i\ge 1\}
\end{array}
\]

To prove that a hard family is contained in every class not satisfying
a certain property, we use arguments based on Ramsey theory. The following lemma (proved in Section~\ref{sec:polyn-time-solv}) characterizes hereditary classes that are not matching-splittable. We define $n\cdot H$ to be the graph that contains $n$ disjoint copies of $H$. (Recall that $P_3$ is the path on 3 vertices.) 
\begin{restatable}{\retheorem}{restateramseymatching}
\label{theorem:ramsey:matchingsplittable}
Let $\F$ be a hereditary graph family that is not matching splittable. Then at least one of the following holds:
\begin{enumerate*}
\item $\F$ is a superset of \Fclique.
\item $\F$ is a superset of \Fbiclique.
\item $\F$ contains $n\cdot K_3$ for every $n\ge 1$.
\item $\F$ contains $n\cdot P_3$ for every $n\ge 1$.
\end{enumerate*}
\end{restatable}

Observe that Theorem~\ref{theorem:ramsey:matchingsplittable} is a tight
characterization of matching-splittable graphs: the converse statement
is also true, that is, if any of the four statements is true for $\F$,
then it is not matching-splittable. Clearly, large cliques and large
bicliques are not $c$-matching-splittable for constant $c$. Moreover,
if every component of a graph has three vertices (that is, it is either a $K_3$
or $P_3$), then at least one vertex has to be deleted from each
component to decrease the size of every component to at most two
vertices, hence $\F$ cannot be $c$-matching-splittable for constant
$c$ in the last two cases either.

In Section~\ref{sec:packing-ramsey}, we characterize hereditary classes that are not small/thin.
\begin{restatable}{\retheorem}{restateramseysmallthin} \label{theorem:ramsey:1} Let~$\F$ be a
  hereditary graph family that is not small/thin.
Then $\F$ is a superset of at least one of the following families:
\begin{enumerate*}
\item \Fpath,
\item \Fclique,
\item \Fbiclique,
\item \Ffountain{s} for some odd integer $s\ge 3$,
\item \Flongfountain{s}{t} for some odd integer $s\ge 3$ and integer $t\ge 1$,
\item \Foperahouse{s} for some odd integer $s\ge 1$,
\item \Fsubdivstar, or
\item \Fdoublebroom{s} for some odd integer $s\ge 1$.
\end{enumerate*}
\end{restatable}
Again, the characterization is tight: we can observe that if $\F$ is a
superset of any of these families, then there is no $a,b\ge 0$ such
that $\F$ is $a$-small/$b$-thin. Note that we cannot leave out any of
the eight items from the list: the hereditary closure of, say,
\Flongfountain{5}{2} is not the superset of any of the classes
described in the remaining seven items.

Finally, in Section~\ref{sec:subgraph-ramsey}, we characterize graphs that are not splittable. 
\begin{restatable}{\retheorem}{restateramseysplittable}
\label{theorem:ramsey:separator}
Let~$\F$ be a hereditary graph family that is not splittable. Then at least one of the following holds:
\begin{enumerate*}
\item $\F$ is a superset of $\Fpath$,
\item $\F$ is a superset of $\Fclique$,
\item $\F$ is a superset of $\Fbiclique$,
\item $\F$ contains $n\cdot \Gsubdivstar{n}$ for every $n\ge 1$,
\item there is an odd $s\ge 3$ such that $\F$ contains $n\cdot \Gfountain{s}{n}$ for every $n\ge 1$,
\item there is an odd $s\ge 1$ such that $\F$ contains $n\cdot \Goperahouse{s}{n}$ for every $n\ge 1$,
\item there is an odd $s\ge 1$ such that $\F$ contains $n\cdot \Gdoublebroom{s}{n}$ for every $n\ge 1$,
\item there is an odd $s\ge 3$ and arbitrary $t\ge 1$ such that $\F$ contains $n\cdot \Glongfountain{s}{t}{n}$ for every $n\ge 1$,
\item $\F$ is a superset of $\Fsubdivtree{s}$ for some integer $s\ge 1$, or
\item $\F$ is a superset of $\Fdiamondfan$.
\end{enumerate*}
\end{restatable}
We can again verify that the characterization is tight. In particular, let us show that
$\Gsubdivtree{s}{n}$ is not $(a,b,c,d)$-splittable if
$n>a+b+c+d$. Suppose that $S$ realizes the $(a,b,c,d)$-split. As the
graph is not $b$-thin and has more than $a$ vertices, we have that $S$
is not empty. By the pigeonhole principle, there is a vertex $v$ with
degree $n+1$ that is not in $S$, and none of its degree-1 neighbors
are in $S$ either. Then $v$ is in a component of size at least $n+1>a$
that contains at least $n>d$ vertices that have no neighbors in
$S$. Similarly, suppose that $S$ realizes an $(a,b,c,d)$-split of
$\Gdiamondfan{n}$ for $n>a+b+c+d$. Again, $S$ is not empty. By the
pigeonhole principle, there is a degree-$n$ vertex $v$ that is not in
$S$ and has no neighbor in $S$. The component of this vertex has size
more than $a$ and the component has more than $d$ vertices (namely,
every neighbor of $v$) whose closed neighborhood is not universal to
$S$.

\subsection{Hardness proofs} \label{section:outline:hardness}
Let us review the concrete hardness results that we prove, which, by
the combinatorial characterizations in Theorems \ref{theorem:ramsey:matchingsplittable}--\ref{theorem:ramsey:separator}, prove the negative parts of
Theorems~\ref{theorem:intro:packing}--\ref{theorem:intro:turingsubgraph}. 
As mentioned above, Kirkpatrick and Hell \cite{KirkpatrickH78} fully characterized the polynomial-time solvable cases of \HPacking.

\begin{theorem}[\cite{KirkpatrickH78}] \label{theorem:fpacking:pvsnp}
\HPacking is polynomial-time solvable if every connected component of $H$ has at most two vertices and NP-complete otherwise.
\end{theorem}

It follows from Theorem~\ref{theorem:fpacking:pvsnp} that
\FSubgraphTest is NP-hard if $\F$ contains $n\cdot K_3$ for
every $n\ge 1$ or if $\F$ contains $n\cdot P_3$ for every $n\ge
1$, as then the problem is more general than $K_3$-\Packing
or $P_3$-\Packing, respectively. Also, \FSubgraphTest is NP-hard
if $\F$ contains every clique~\cite[GT7]{GareyJ79} (it generalizes \textsc{Clique}) or if $\F$ contains every biclique~\cite[GT24]{GareyJ79}.

For kernelization lower bounds, observe first that if~$\F$
contains every clique, then \FPacking and \FSubgraphTest are clearly
W[1]-hard~\cite[Theorem 21.2.4]{DowneyF13} and therefore do not admit a (Turing) kernel of any size, unless FPT~$=$~W[1] and the Exponential Time Hypothesis fails~\cite[Chapter 29]{DowneyF13}. A recent result of Lin~\cite{Lin14} shows that if~$\F$ contains every biclique, then \FPacking and \FSubgraphTest are also W[1]-hard. Since the parameterized \textsc{Clique} and \textsc{Biclique} problems are NP-hard and OR-compositional~\cite{BodlaenderDFH09}, it follows from standard kernelization lower bound machinery that if \FPacking or \FSubgraphTest has a polynomial (many-one) kernel when~$\F$ contains every clique or biclique, then \containment. To complete the proof of the negative parts of Theorems~\ref{theorem:intro:packing} and \ref{theorem:intro:turingsubgraph}, we prove the following two sets of \textup{WK[1]}-hardness results.
\begin{restatable}{\retheorem}{restatepackinglower}
\label{theorem:packing:lowerbounds}
The \FPacking problem is \textup{WK[1]}-hard under polynomial-parameter transformations if $\F$ is a superset of any of the following families:
\begin{enumerate*}
	\item $\Fsubdivstar$,
	\item $\Flongfountain{s}{t}$ for some integer~$t \geq 1$ and some \emph{odd} integer~$s \geq 3$,
	\item $\Fdoublebroom{s}$ for some odd integer $s \geq 1$,
	\item $\Ffountain{s}$ for some odd integer $s \geq 3$, or
	\item $\Foperahouse{s}$ for some odd integer $s \geq 1$.
\end{enumerate*}
\end{restatable}

\begin{restatable}{\retheorem}{restatesubgraphlower}\label{theorem:subgraphtest:lowerbounds}
  The \FSubgraphTest problem is \textup{WK[1]}-hard under polynomial-parameter transformations if $\F$ is a superset of any of the following families:
\begin{enumerate*}
\item $\Fdiamondfan$, or
\item $\Fsubdivtree{s}$ for some integer $s$.
\end{enumerate*}
\end{restatable}
All \textup{WK[1]}-hardness proofs are by reduction from \nRegularExactSetCover, where the parameter equals the size of the universe on which the set system is defined. The uniform variant, in which all sets have the same size, is particularly useful for proving these results. 
We prove the \textup{WK[1]}-hardness of the problem by a two-stage transformation from \nExactSetCover~\cite{HermelinKSWW13}, first introducing a small number of new elements  to ensure that solutions exist that contain a prescribed number of sets, and then using this knowledge to introduce another small number of elements that can be added to the sets to make the system uniform.

\subsection{Many-one kernels}
We do not have a complete characterization of the existence of many-one
kernels for \FSubgraphTest. The authors believe that if such a
characterization is possible, then it has to be significantly more
delicate than the characterization of Turing kernels in
Theorem~\ref{theorem:intro:turingsubgraph} and both the positive and
the negative parts should involve a larger number of specific
cases. We present two lower bounds and two upper bounds to show the difficulties that arise (see also Figure~\ref{fig:karp}). The following two theorems give the lower bounds.

\begin{figure}[t]
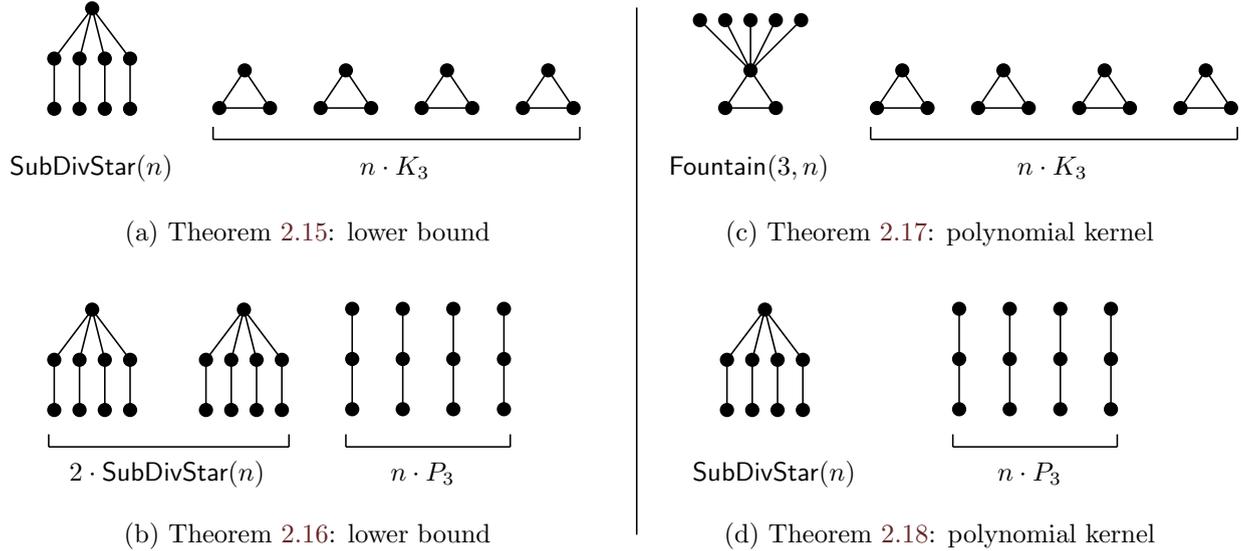

\begin{center}
{\small \svg{\linewidth}{karp}}
\caption{Illustrating the classes of graphs in Theorems~\ref{theorem:karp:lowerbound:onesubstar:manytriangles}--\ref{theorem:karp:subgraphkernel:starsandpaths}.}
\label{fig:karp}
\end{center}
\end{figure}

\begin{restatable}{\retheorem}{restatekarponesubstarmanytriangles}
\label{theorem:karp:lowerbound:onesubstar:manytriangles}
Let~\F be any hereditary graph family containing all graphs of the form~$H' + \ell \cdot K_3$, where~$\ell \geq 1$ and~$H' \in \Fsubdivstar$. Then \kFSubgraphTest does not admit a polynomial many-one kernel unless \containment.
\end{restatable}

\begin{restatable}{\retheorem}{restatekarptwostarmanyps}
\label{theorem:karp:lowerbound:twosubstar:manyps}
Let~\F be any hereditary graph family containing all graphs of the form~$H' + H'' + \ell \cdot P_3$, where~$\ell \geq 1$ and~$H', H'' \in \Fsubdivstar$. Then \kFSubgraphTest does not admit a polynomial many-one kernel unless \containment.
\end{restatable}

Observe that the graph families described by these theorems are~$(3,0,2,2)$-splittable: letting~$S$ contain the (at most two) centers of the subdivided stars, the connected components that remain after removing~$S$ have at most three vertices. Every leg of a subdivided star becomes a component of size two in which one of the vertices is universal to~$S$ and the other is not; hence the closed neighborhoods of the two vertices are not universal to~$S$. The \kFSubgraphTest problem for these families therefore has polynomial Turing kernels by Theorem~\ref{theorem:kernel:subgraph}, highlighting the difference between many-one and Turing kernelization for \kFSubgraphTest. The following two theorems give upper bounds.

\begin{restatable}{\retheorem}{restatekarpfountaintriangles}
\label{karp:subgraphkernel:fountainandtriangles}
Let~\F be the hereditary closure of the family containing all graphs of the form~$H' + \ell \cdot K_3$, where~$\ell \geq 1$ and~$H' \in \Ffountain{3}$. Then \kFSubgraphTest admits a polynomial many-one kernel.
\end{restatable}

\begin{restatable}{\retheorem}{restatekarpstarspaths}
\label{theorem:karp:subgraphkernel:starsandpaths}
Let~\F be the hereditary closure of the family containing all graphs of the form~$H' + \ell \cdot P_3$, where~$\ell \geq 1$ and~$H' \in \Fsubdivstar$. Then \kFSubgraphTest admits a polynomial many-one kernel.
\end{restatable}

Comparing Theorem~\ref{theorem:karp:lowerbound:onesubstar:manytriangles} to Theorem~\ref{karp:subgraphkernel:fountainandtriangles}, we find that changing the type of the single large component from a subdivided star to a fountain crosses the threshold for the existence of a polynomial kernel, even though both types of graphs can be reduced to constant-size components by a single vertex deletion. Comparing Theorem~\ref{theorem:karp:lowerbound:twosubstar:manyps} to Theorem~\ref{theorem:karp:subgraphkernel:starsandpaths} we see that decreasing the number of subdivided star components from two to one makes a polynomial kernel possible. While the definition of splittable graph families that characterizes the existence of polynomial Turing kernels for \kFSubgraphTest is robust under increases by constants, this is clearly not the case for the many-one complexity of \kFSubgraphTest.


\renewcommand{\retheorem}{rtheorem}
\renewcommand{\reabctheorem}{rabctheorem}
\renewcommand{\relemma}{rlemma}

\subsection{Motivation for hereditary classes} \label{section:motivate:hereditary}
In this paper, we
restricted our study to hereditary classes $\F$. There are a number of
reasons motivating this decision. First, considering arbitrary classes
$\F$ can make it very hard to prove lower bounds by polynomial-time
reductions (even if the classes are decidable). For a concrete example, pick $\F$
consisting of every clique of size $2^{2^{2^i}}$ for $i\ge 1$. Then \FSubgraphTest is unlikely to be polynomial-time solvable, but this seems
difficult to prove with a polynomial-time reduction as the smallest clique in~$\F$ of size exceeding~$n$ may have superpolynomial size. A
difficulty of different sorts appears if $\F$ contains cliques such
that the sizes of cliques in $\F$ form a more dense set of integers
than in the previous example, but deciding if the clique of a particular
size~$n$ is in~$\F$ takes time exponential in~$n$. These issues may
be considered artifacts of trying to prove hardness by uniform
polynomial-time reductions that work for every input length~$n$; 
potentially the issues can be avoided by formulating the complexity
framework in a different way. However, there are even more substantial
difficulties that appear when the class $\F$ is not hereditary. For
example, let $\F$ be the set of all paths. Then $\FSubgraphTest$ is
NP-hard and it does not admit a polynomial kernel, unless \containment
\cite{BodlaenderDFH09}. Consider now the class $\F'$ containing, for
every $i\ge 1$, the graph formed by a path of length $i$ together with
$2^i$ isolated vertices. The introduction of the isolated vertices
should not change the complexity of the problem, but, surprisingly, it
does. The problem of finding a path of length $k$ in an $n$-vertex
graph can be solved in time $2^{O(k)}\cdot n^{O(1)}$
\cite{AlonYZ95,DBLP:journals/ipl/Williams09,Bjorklund10}. Therefore,
if $H$ consists of a path of length $k$ and $2^k$ isolated vertices,
then these algorithms give a polynomial-time algorithm for finding $H$
in a graph $G$: the running time $2^{O(k)}\cdot n^{O(1)}$ is
polynomial in the size of $H$ and $G$. Therefore, $\F'$-\SubgraphTest
is polynomial-time solvable, but apparently only because finding a
path of length $k$ is fixed-parameter tractable and has
$2^{O(k)}n^{O(1)}$ time algorithms (note that the $2^{O(k\log
  k)}n^{O(1)}$ time algorithm of Monien \cite{Monien85} would not be
sufficient for this argument).  Therefore, it seems that we need a
very tight understanding of the fixed-parameter tractability of
\kFSubgraphTest to argue about its polynomial-time solvability. There
are examples in the literature where the polynomial-time solvability
of a problem was characterized for every (not necessarily hereditary)
class $\F$
\cite{DBLP:conf/icalp/ChenTW08,DBLP:journals/tcs/DalmauJ04,1206036,380867},
but in all these results, the characterization of polynomial-time was
possible only because it coincided with fixed-parameter
tractability. There is certainly no such coincidence for
\kFSubgraphTest (for example, finding a path of length $k$ is NP-hard,
but FPT) and moreover the fixed-parameter tractability of
\kFSubgraphTest is not well understood, as shown, for example, by the
\textsc{Biclique} problem.

All these problems disappear if we restrict $\F$ to be hereditary
(e.g., adding isolated vertices certainly cannot make the problem
easier and if $\F$ contains arbitrary large cliques, then $\F$
contains every clique). While this restricts the generality of our
results to some extent, we believe that avoiding the difficulties
discussed above more than compensates for this lack of generality.

Let us also comment on the fact that, while the problems we study concern finding and packing (non-induced) subgraphs, we characterize the difficulty of such problems 
for each class~$\F$ of pattern graphs that is closed under induced subgraphs (i.e., for hereditary classes). The discrepancy between induced and non-induced here is entirely 
natural. Note that every class~$\F$ of pattern graphs that is closed under subgraphs, is also closed under induced 
subgraphs, and is therefore covered by our dichotomies. The fact that we can also classify~$\F$ that are merely closed under induced subgraphs, rather than normal subgraphs, gives our results extra strength.

We mention in passing the classical result of Lewis and Yannakakis
\cite{LewisY80} on fully characterizing the complexity of
vertex-deletion problems defined by hereditary properties; these
results also rely crucially on the assumption that the property is
hereditary. Note that our results on \FSubgraphTest are unrelated to
the results of Lewis and Yannakakis \cite{LewisY80}: their problem is
related to finding induced subgraphs and the task is not to find a
specific induced subgraph, but to find a subgraph belonging to the
class and having a specified size.

\section{Preliminaries} \label{section:preliminaries}

For integers~$n$ we denote the set~$\{1, \ldots, n\}$ by~$[n]$. If~$X$ is a finite set and~$n \in \mathbb{N}$ then~$\binom{X}{n}$ is the collection of size-$n$ subsets of~$X$. Similarly, we use~$\binom{X}{\leq n}$ for the collection of all subsets of~$X$ that have size \emph{at most}~$n$, including the empty set.

\subsection{Parameterized complexity and kernelization}
A parameterized problem~$\Q$ is a subset of~$\Sigma^* \times \mathbb{N}$, the second component of a tuple~$(x,k) \in \Sigma^* \times \mathbb{N}$ is called the \emph{parameter}. A parameterized problem is (strongly uniformly) \emph{fixed-parameter tractable} if there exists an algorithm to decide whether $(x,k) \in \Q$ in time~$f(k)|x|^{\Oh(1)}$ where~$f$ is a computable function. A \emph{many-one kernelization algorithm} (or \emph{many-one kernel}) of size~$f \colon \mathbb{N} \to \mathbb{N}$ for a parameterized problem~$\Q \subseteq \Sigma^* \times \mathbb{N}$ is an algorithm that, on input~$(x,k) \in \Sigma^* \times \mathbb{N}$, runs in time polynomial in~$|x| + k$ and outputs an instance~$(x', k')$ with~$|x'|, k' \leq f(k)$ such that~$(x,k) \in \Q \Leftrightarrow (x', k') \in \Q$. It is a \emph{polynomial kernel} if~$f$ is a polynomial (cf.~\cite{Bodlaender09}). 

\begin{definition} \label{definition:turing:kernelization}
Let~$\Q$ be a parameterized problem and let~$f \colon \mathbb{N} \to \mathbb{N}$ be a computable function. A \emph{Turing kernelization for~$\Q$ of size~$f$} is an algorithm that decides whether a given instance~$(x,k) \in \Sigma^* \times \mathbb{N}$ is contained in~$\Q$ in time polynomial in~$|x| + k$, when given access to an oracle that decides membership in~$\Q$ for any instance~$(x',k')$ with~$|x'|, k' \leq f(k)$ in a single step.
\end{definition}

We refer to a textbook~\cite{FlumG06} for more background on parameterized complexity.

\subsection{Graphs} \label{section:preliminaries:graphs}
All graphs we consider are finite, undirected, and simple, unless explicitly stated otherwise. A graph~$G$ consists of a vertex set~$V(G)$ and an edge set~$E(G) \subseteq \binom{V(G)}{2}$. If~$H,G$ are graphs such that~$V(H) \subseteq V(G)$ and~$E(H) \subseteq E(G)$ then~$H$ is a \emph{subgraph} of~$G$, denoted~$H \subseteq G$. For a vertex set~$X \subseteq V(G)$, the subgraph of~$G$ \emph{induced} by~$G$ is the graph with vertex set~$X$ and edge set~$E(G) \cap \binom{X}{2}$. We use~$G - X$ as a shorthand for~$G[V(G) \setminus X]$. The \emph{open neighborhood} of a vertex~$v \in V(G)$ in graph~$G$ is denoted~$N_G(v)$, while the closed neighborhood (which includes~$v$ itself) is~$N_G[v]$. If~$X$ is a vertex set then~$N_G(X) = \bigcup _{v\in X} N_G(v) \setminus X$, while~$N_G[X] = \bigcup _{v \in X} N_G[v]$. We use~$\deg_G(v)$ to denote the degree of vertex~$v$ in graph~$G$. The maximum degree of~$G$ is denoted~$\Delta(G)$. For vertex sets~$X$ and~$Y$ of a graph~$G$ we say that \emph{$X$ is universal to~$Y$} if each vertex of~$X$ is adjacent to all vertices of~$Y$. If~$G$ is a graph and~$X \subseteq V(G)$ is a vertex set, then the operation of \emph{identifying the vertices~$X$ into a single vertex} consists of removing the vertices~$X$ and their incident edges, replacing them by a single new vertex~$v_X$ whose neighborhood becomes~$N_G(X)$. If~$G$ and~$H$ are graphs, then~$G + H$ denotes the disjoint union of the two graphs. For a positive integer~$t$, we denote by~$t \cdot G$ the disjoint union of~$t$ copies of~$G$. A \emph{vertex cover} of a graph~$G$ is a set~$X \subseteq V(G)$ that contains at least one endpoint of every edge. The \emph{vertex cover number} of a graph is the size of a smallest vertex cover.

By~$K_n$ ($K_{n,n}$) we denote the complete (bipartite) graph on~$n$ vertices. A complete bipartite graph is also called a \emph{biclique}. It is \emph{balanced} if its two partite classes have equal sizes. A connected bipartite graph~$G$ is $b$-thin for~$b \in \mathbb{N}$ if it has a partite class of size at most~$b$. The cycle on~$n$ vertices is denoted~$C_n$ and the path on~$n$ vertices is denoted~$P_n$. Observe that graph~$P_n$ is a path of length~$n - 1$. At various points in the paper we have to consider a graph~$G$ and the common neighbors of a vertex set~$D \subseteq V(G)$ in that graph, which is denoted~$\bigcap _{v \in D} N_G(v)$. To avoid some case distinctions, we will also allow the set~$D$ to be empty in such expressions. Since~$\bigcap _{v \in D} N_G(v)$ consists of the elements~$x$ that belong to~$N_G(v)$ for every~$v \in D$, when~$D = \emptyset$ this holds for all elements. Hence for~$D = \emptyset$ the expression~$\bigcap _{v \in D} N_G(v)$ evaluates to all elements in the universe of discourse, which will simply be~$V(G)$ unless explicitly stated otherwise. The \emph{hereditary closure} of a graph family~$\F$ is the hereditary family containing all graphs in~$\F$ and all their induced subgraphs.

Let~$H$ and~$G$ be graphs and let~$P \subseteq V(H)$. A \emph{$P$-partial subgraph model} of~$H$ in~$G$ is an injection~$\phi \colon P \to V(G)$ such that for all edges~$\{u,v\} \in E(H)$ with~$\{u,v\} \subseteq P$ we have~$\{\phi(u), \phi(v)\} \in E(G)$. We say that~$P$ is the \emph{domain} of~$\phi$. For sets~$P' \subseteq P$ we will write~$\phi(P')$ to denote the set~$\{ \phi(v) \mid v \in P'\}$. A $P$-partial $H$-subgraph model is a \emph{full subgraph model} if~$P = V(H)$. If~$\phi$ is a $P$-partial subgraph model in~$G$ and~$\phi'$ is a $P'$-partial subgraph model in~$G' \subseteq G$ then~$\phi'$ is an \emph{extension} of~$\phi$ if~$P \subseteq P'$ and~$\phi(v) = \phi'(v)$ for all~$v \in P$. If~$\phi$ and~$\phi'$ are partial $H$-subgraph models with domains~$P, P' \supseteq X$ then the models \emph{agree on~$X$} if~$\phi(v) = \phi'(v)$ for all~$v \in X$. If~$\phi$ is a $P$-partial $H$-subgraph model in~$G$, then for any set~$X \subseteq V(H)$ we define the \emph{restriction of~$\phi$ to~$X$} as the partial $H$-subgraph model~$\phi|_X \colon P \cap X \to V(G)$ given by~$\phi|_X(v) = \phi(v)$ for all~$v \in P \cap X$. The restriction will sometimes be used as a partial $H[X]$-subgraph model rather than a partial $H$-subgraph model; it will be clear from the context which is meant.

A \emph{separation} of a graph~$G$ is a pair~$(A, B)$ of subsets of~$V(G)$ such that~$A \cup B = V(G)$ and there are no edges between~$A \setminus B$ and~$B \setminus A$. The following observation formalizes that if~$(A,B)$ is a separation of~$H$ into two parts and we have full subgraph models for~$H[A]$ and~$H[B]$ that realize the separator~$A \cap B$ in the same way, then these models can be glued together to form a full subgraph model of~$H$.

\begin{observation} \label{observation:merge:models}
Let~$G$ and~$H$ be graphs, let~$(A,B)$ be a separation of~$H$, and let~$\phi$ be a partial~$H$-subgraph model in~$H$ with domain~$P \supseteq A \cap B$. If~$\phi_A$ is a full $H[A]$-subgraph model in~$G$ that extends~$\phi|_A$ and~$\phi_B$ is a full $H[B]$-subgraph model in~$G$ that extends~$\phi|_B$ such that~$\phi_A(A \setminus B) \cap \phi_B(B \setminus A) = \emptyset$, then the injection~$\phi^* \colon V(H) \to V(G)$ defined as follows:
\begin{equation*}
\phi^*(v) = 
\begin{cases} \phi_A(v) & \text{if~$v \in A \setminus B$,} \\
\phi_B(v) & \text{if~$v \in B \setminus A$,} \\
\phi_A(v) = \phi_B(v) & \text{if~$v \in A \cap B$,}
\end{cases}
\end{equation*}
is a full $H$-subgraph model in~$G$ that extends~$\phi$.
\end{observation}

\subsection{Ramsey theory}
We review the basic results that we need. Ramsey's Theorem states that if
the edges of a sufficiently large clique are colored with a bounded
number of colors, then there has to be a large \emph{monochromatic} clique with every edge
having the same color.
\begin{theorem}[{Cf.~\cite[Chapter 25]{GrahamGL95}}]\label{theorem:ramsey}
For every choice of positive integers~$n$ and~$t$ there exists a number~$R(n,t)$ such that for every $t$-coloring~$f \colon E(K_{R(n,t)}) \to [t]$ of the edges of the complete $R(n,t)$-vertex graph, there exists a monochromatic complete subgraph with~$n$ vertices.
\end{theorem}
There is a variant of Ramsey's Theorem for bipartite graphs, where we
are looking for a monochromatic biclique of a specific size in an edge
coloring of a large biclique.
\begin{theorem}[{\cite{BeinekeS76,CarnielliC1999}}]\label{theorem:bipartite:ramsey}
  For every choice of positive integers~$n$ and~$t$, there exists a
  number~$\bipRamsey(n,t)$ such that for every $t$-coloring~$f \colon
  E(K_{\bipRamsey(n,t),\bipRamsey(n,t)}) \to [t]$ of the edges of the balanced biclique
  with $\bipRamsey(n,t)$ vertices in each class, there exists a monochromatic
  balanced biclique with $n$ vertices in each class.
\end{theorem}
We also use a recent result of Atminas, Lozin, and Razgon \cite{AtminasLR12} showing
that a long path implies the existence of a long induced path or a
large biclique.
\begin{theorem}[{\cite[Theorem 1]{AtminasLR12}}] \label{theorem:path:ramsey}
For every choice of positive integers~$n$ and~$k$, there exists a number~$P_0(n,k)$ such that any graph with a path on~$P_0(n,k)$ vertices either contains an induced path on~$n$ vertices or a (not necessarily induced)~$K_{k,k}$ subgraph.
\end{theorem}
Note that in Theorem~\ref{theorem:path:ramsey}, the biclique $K_{k,k}$ appears
as a subgraph, not as an induced subgraph. However, it is easy to
strengthen Theorem~\ref{theorem:path:ramsey} in a way that we can
assume that~$K_{k,k}$ is induced.

\begin{corollary}\label{theorem:path:ramsey2}
For every choice of positive integers~$n$ and~$k$, there exists a number~$P(n,k)$ such that any graph with a path on~$P(n,k)$ vertices either contains an induced path on~$n$ vertices or a $K_{k}$ subgraph, or an induced $K_{k,k}$ subgraph.
\end{corollary}
\begin{proof}
  Let $k'=R(k,k)$ for the function $R$ in Ramsey's Theorem (Theorem~\ref{theorem:ramsey}) and let
  $P(n,k)=P_0(n,k')$ for the function $P_0$ in
  Theorem~\ref{theorem:path:ramsey}. Then
  Theorem~\ref{theorem:path:ramsey} implies that every graph with at
  least $P(n,k)$ vertices contains either an induced path on $n$
  vertices (in which case we are done) or a (not necessarily induced)
  $K_{k',k'}$ subgraph. Let $X$ and $Y$ be the two partite classes of
  the $K_{k',k'}$ subgraph. By Ramsey's Theorem, $X$ contains either a
  clique on $k$ vertices (in which case we are done) of an independent
  set $X'\subseteq X$ on $k$ vertices. Similarly, we may assume that
  there is an independent set $Y'\subseteq Y$ of size $k$. Now $X\cup
  Y$ induces a $K_{k,k}$ subgraph.
\end{proof}

\subsection{WK[1]-hardness proofs} \label{section:wkhardness}

The complexity class \textup{WK[1]} can be defined~\cite[Section 4]{HermelinKSWW13} as the closure of \nExactSetCover under polynomial-parameter transformations, which are defined as follows.

\begin{definition}[\cite{BodlaenderTY11}] \label{definition:polyParamTransform}
Let~$\Q,\Q'\subseteq\Sigma^*\times \mathbb{N}$ be parameterized problems. A \emph{polynomial-parameter transformation} from~$\Q$ to~$\Q'$ is an algorithm that on input~$(x,k)\in\Sigma^*\times \mathbb{N}$ takes time polynomial in~$|x|+k$ and outputs an instance~$(x',k')\in\Sigma^*\times \mathbb{N}$ such that:
\begin{itemize}
\item The parameter value~$k'$ is polynomially bounded in~$k$.
\item $(x',k') \in \Q'$ if and only if~$(x,k) \in \Q$.
\end{itemize}
\end{definition}

\noindent We will use a \emph{regular} variant of the set cover problem as the starting point for our \textup{WK[1]}-hardness proofs. If~$\S$ is a set system over universe~$U$ then we will say that the pair~$(\S,U)$ has an exact cover if there is a subsystem~$\S' \subseteq \S$ such that each element of~$U$ is contained in exactly one set of~$\S'$. Recall that a set system is $r$-uniform if all sets have size exactly~$r$.

\parproblemdef{\nRegularExactSetCover}
{An integer~$r \geq 3$ and an $r$-uniform set system~$\S$ over a universe~$U$.}
{$n := |U|$.}
{Does~$(\S,U)$ have an exact cover?}

\begin{lemma} \label{lemma:regularexactsetcover:wkhard}
\nRegularExactSetCover is \textup{WK[1]}-hard.
\end{lemma}
\begin{proof}
Hermelin et al.~\cite{HermelinKSWW13} proved that the variant \nExactSetCover, where sets may have different sizes, is \textup{WK[1]}-hard. The \nExactSetCover problem asks for a given set system~$\S$ over a universe~$U$ whether~$(\S,U)$ has an exact cover. The parameter is~$n := |U|$. We may assume that~$n \geq 2$, as otherwise we can solve the problem in polynomial time and output a constant-size instance that gives the same answer. We transform an instance~$(\S, U, n)$ of \nExactSetCover to an equivalent instance of \nRegularExactSetCover in two steps.

First we create a system~$(\S', U')$ such that~$|U'| = 2|U|$ and~$\S$ has an exact set cover if and only if~$\S'$ has an exact set cover with exactly~$n+1$ sets. The system is built by adding~$n$ new elements~$u'_1, \ldots, u'_n$ to the universe and adding all sets~$\{ \{u'_i, \ldots, u'_j\} \mid 1 \leq i \leq j \leq n \}$ to~$\S$ to obtain the system~$\S'$.

\begin{claim}
$(\S,U)$ has an exact set cover if and only if~$(\S',U')$ has an exact cover consisting of~$n+1$ sets.
\end{claim}
\begin{claimproof}
Suppose that~$(\S,U)$ has an exact set cover~$S_1, \ldots, S_k$. Observe that~$1 \leq k \leq n$, otherwise some universe element is contained in two sets. The~$n$ new elements~$u'_1, \ldots, u'_n$ that have been added to the sequence can be exactly covered with the single set~$\{u'_1, \ldots, u'_n\}$, with~$n$ sets~$\{u'_1\}, \{u'_2\}, \ldots, \{u'_n\}$, and with any number of sets between one and~$n$, since all consecutive intervals of these new elements have been added to~$\S'$. Hence we may augment the exact cover for~$(\S,U)$ with an exact cover of the new elements with~$(n+1) - k$ sets. We obtain an exact cover for~$(\S',U')$ with~$n+1$ sets.

In the reverse direction, observe that all sets in~$\S'$ that contain an element from the original set~$U$, also exist in~$\S$. Hence from an exact cover of~$(\S',U')$ we can select the sets containing elements from~$U$ to obtain an exact cover for~$(\S,U)$.
\end{claimproof}

From the system~$(\S', U')$ we then construct another system~$(\S^*, U^*)$, as follows. Form~$U^*$ by adding~$2n^2$ new elements~$\{u^*_1, \ldots, u^*_{2n^2}\}$ to~$U'$. Define~$\S^*$ as~$\{S \cup \{u^*_i, \ldots, u^*_{i + (2n - |S|) - 1}\} \mid S \in \S' \wedge 1 \leq i \leq 2n^2 - (2n - |S|) + 1 \}$, which is~$2n$-uniform.

\begin{claim}
$(\S',U')$ has an exact cover consisting of~$n+1$ sets if and only if~$(\S^*, U^*)$ has an exact cover.
\end{claim}
\begin{claimproof}
Suppose that~$(\S',U')$ has an exact cover consisting of~$n+1$ sets~$S_1, \ldots, S_{n+1}$. Observe that each set in~$\S'$ has size at most~$n$. For~$i \in [n+1]$ define~$t_i := \sum _{j=1}^{i-1} (2n-|S_j|)$, implying~$t_1 = 0$. For each~$i \in [n+1]$, define~$S'_i := S_i \cup \{u^*_{t_{i-1} + 1}, \ldots, u^*_{t_{i-1} + 2n - |S_i|}\}$. Then~$S'_i \in \S^*$ by our definition of~$\S^*$. Hence the sets~$S'_1, \ldots, S'_{n+1}$ are contained in~$\S^*$, each have size~$2n$, and are pairwise disjoint. As they contain~$2n(n+1) = 2n^2 + 2n = |U^*|$ elements in total, they form an exact set cover of~$(\S^*,U^*)$.

For the reverse direction, observe that as each set in~$\S^*$ has size~$2n$ while the universe size is~$2n^2 + 2n$, any exact set cover in~$(\S^*, U^*)$ consists of exactly~$n+1$ sets~$S'_1, \ldots, S'_{n+1}$. The intersection of each set~$S'_i$ with~$U'$ is non-empty and contained in~$\S'$, by definition of~$\S^*$. Hence~$S'_1 \cap U', \ldots, S'_{n+1} \cap U'$ is an exact cover of~$(\S',U')$ consisting of~$n+1$ sets.
\end{claimproof}

Together the two claims show that the system~$(\S^*,U^*)$ has an exact set cover if and only if the input~$(\S,U)$ has one. System~$\S^*$ is $r$-uniform for~$r := 2n \geq 4$. The universe size~$n^* := |U^*|$ equals~$2n + 2n^2$, which is polynomial in the parameter of the input instance. Since the construction can be performed in polynomial time, it forms a valid polynomial-parameter transformation from \nExactSetCover to \nRegularExactSetCover, which concludes the proof by the \textup{WK[1]}-hardness of the former problem.
\end{proof}

\section{Polynomial-time solvable versus NP-complete}
\label{sec:polyn-time-solv}
In this section, we prove Theorem~\ref{theorem:intro:polysubgraph},
characterizing the herditary classes $\F$ for which \FSubgraphTest can
be solved in randomized polynomial time.  In
Section~\ref{sec:poly-upper-bound}, we use the randomized matching
algorithm of Mulmuley, Vazirani, and Vazirani \cite{MulmuleyVV87}
(Theorem~\ref{prop:wmatching}) to solve \FSubgraphTest in randomized
polynomial-time for matching-splittable hereditary families.  In
Section~\ref{sec:poly-lower-bound}, we prove
Theorem~\ref{theorem:ramsey:matchingsplittable} characterizing
matching-splittable herditary classes. In
Section~\ref{subsection:poly:theorem}, we simply put together these
results to complete the proof of
Theorem~\ref{theorem:intro:polysubgraph}.

\subsection{Upper bound}
\label{sec:poly-upper-bound} \label{section:polynomialvsnpcomplete:upperbounds}

Our algorithm builds on the following algebraic matching procedure.

\begin{proposition}[Mulmuley, Vazirani, and Vazirani \cite{MulmuleyVV87}]\label{prop:wmatching}
  There exists a randomized algorithm with false negatives that, given
  a multigraph $G$ with nonnegative integer weights and a target weight $w_0$, checks in time polynomial in $|V(G)|+|E(G)|$
  and $w_0$ whether there exists a perfect matching in $G$ of weight
  exactly $w_0$.
\end{proposition} 

Note that originally the result is not stated for multigraphs, but it is easy to modify the proof accordingly. Alternatively, one can 
get rid of multiple edges by subdividing each edge twice and setting the weights appropriately. Let us prove Theorem~\ref{theorem:coloredmatching} using this result.
\restatecoloredmatching*
\begin{proof}
  We may assume that $C=[c]$ for some integer $i\ge 1$. Let $n=|V(G)|$
  and let $s=\sum_{i=1}^{c}f(i)$ be the size of the matching we are
  looking for. Let us set the weight of each edge of color $i$ to
  $n^i$. Additionally, let us introduce a set $X$ of $n-2s$ vertices
  and let us connect every original vertex to every vertex of $X$ with
  and edge of weight 0. Let $G'$ be the resulting graph. We claim that
  $G$ has a matching with the required number of colors if and only
  $G'$ has a perfect matching of weight exactly
  $w_0=\sum_{i=1}^cf(i)n^i=n^{O(c)}$ (we may assume that $f(i)\le
  n$ otherwise there is no solution). The existence of such a perfect
  matching can be tested with the algorithm of
  Proposition~\ref{prop:wmatching} in time polynomial in $|E(G)|$ and
  $n^{O(c)}$.

  Suppose that $M$ is a matching of $G$ with $f(i)$ edges of color $i$
  for every $i\in C$. Then the total weight of $M$ in $G'$ is exactly
  $w_0$. Let us extend $M$ the following way: if vertex $v\in V(G)$ is
  not covered by the matching $M$, then let us add the edge $uv$ (of weight 0) for
  some $u\in X$ to the matching. As exactly $n-2s=|X|$ vertices of $G$
  are not covered by $M$, we may select a distinct $u\in X$ for each
  such edge, resulting in a perfect matching $M'$ of $G'$ having
  weight exactly $w_0$.

  Conversely, suppose that $M'$ is a perfect matching of $G'$ having
  weight exactly $w_0$. Ignoring the edges of weight 0, we get a
  matching $M$ of $G$. It is easy to see that the only way $M'$ can
  have weight exactly $w_0$, is if $M$ has exactly $f(i)$ edges of color
  $i$ for every $1\le i \le c$. Indeed, interpreting $w_0$ as a number
  in base-$n$ notation, this is the only way the weight of the
  edges in $M'$ add up to $w_0$ (note that no ``overflow'' can occur,
  as clearly there are at most $n/2$ edges of weight $n^i$ in
  $M'$).
\end{proof}

Equipped with Theorem~\ref{theorem:coloredmatching}, we prove
Theorem~\ref{theorem:fsubgraphtest:alg}.  \restatesubgraphtestalg*
\begin{proof}
  Given a graph $H\in \F$ and arbitrary graph $G$, we proceed the
  following way. By assumption, there is a set $S\subseteq V(H)$ of
  size at most $c_\F$ such that every component of $H- S$ has
  at most two vertices. We can find such a set $S$ by brute force in
  time $|V(H)|^{O(c_\F)}\le |V(G)|^{O(c_\F)}$, which is polynomial in the input size.  Let
  $s=2^{|S|}$ and let us fix an arbitrary bijection $\iota:[s]\to
  2^{S}$ defining a numbering of the subsets of $S$.  The
  single-vertex components of $H- S$ can be classified
  according to the neighborhood of the vertex in $S$: let $n^1_{i}$ be
  the number of components of $H- S$ where the neighborhood of
  the vertex is exactly $\iota(i)$. For the two-vertex components, we
  need to take into account the neighborhood of both vertices: for
  $1\le i \le j \le s$, let $n^2_{i,j}$ be the number of two-vertex
  components where the two vertices of the component have
  neighborhoods $\iota(i)$ and $\iota(j)$ in $S$, respectively.
This way, we classify the two-vertex components into $s+\binom{s}{2}$ different types.

  The algorithm finds a subgraph model of $H$ in $G$ by considering
  all $S$-partial subgraph model $\phi_0$; this adds a factor of at
  most $|V(G)|^{|S|}\le |V(G)|^{c_\F}$ to the running time. For a
  fixed $\phi_0$, we need to find images for the components of
  $H- S$. We construct an edge-colored multigraph $G'$ the
  following way.  Let $S^*=\phi_0(S)$. For every $v\in V(G)\setminus
  S^*$, we introduce vertices $v$ and $v'$ into $G'$ and add an
  edge $\{v,v'\}$ of color $\iota^{-1}(X)$ for every $\emptyset \subseteq X
  \subseteq \phi^{-1}_0(N_G(v)\cap S^*))$. If $v_1$ and $v_2$ are adjacent
  vertices in $G- S^*$, then we add parallel edges between $v_1$
  and $v_2$ the following way. For every $\emptyset \subseteq X
  \subseteq \phi^{-1}_0(N_G(v_1)\cap S^*)$ and $\emptyset \subseteq Y
  \subseteq \phi^{-1}_0(N_G(v_2)\cap S^*)$, we let $i=\iota^{-1}(X)$,
  $j=\iota^{-1}(Y)$, and add an edge $\{v_1,v_2\}$ of color $(i,j)$ (if $i\le
  j$) or $(j,i)$ (if $j\le i$). Observe that set $C$ of colors we have
  used on the edges of $G$ has size $s+s+\binom{s}{2}=2^{O(c_{\F})}$.

  We can use the algorithm of Theorem~\ref{theorem:coloredmatching} to
  decide if there is a matching $M$ containing exactly $n^1_{i}$ edges
  of color $i$ and exactly $n^2_{i,j}$ edges of color $(i,j)$. Note
  that the number of colors is a fixed constant depending only on
  $c_{\F}$ and the size of the graph $G'$ is polynomial in $|V(G)|$ and the number of colors (as this bounds the number of parallel edges between two vertices). Therefore, the algorithm of
  Theorem~\ref{theorem:coloredmatching} runs in polynomial time for fixed $c_{\F}$. We
  claim that $\phi_0$ can be extended to a full subgraph model if and
  only if such a matching $M$ exists.

  Suppose that $\phi$ is a full subgraph model extending $\phi_0$. If
  vertex $u$ is a single-vertex component of $H- S$ and $u'$
  is a neighbor of $u$ in $S$, then $\phi_0(u')$ is a neighbor of
  $\phi(u)$, or in other words, $N_H(u)\cap S\subseteq
  \phi^{-1}_0(N_{G}(\phi(u))\cap S^*)$. Therefore, if $v=\phi(u)$,
  then by construction an edge $\{v,v'\}$ of color $\iota^{-1}(N_H(v)\cap
  S)$ exists; let us add it to the matching $M$. If $u_1$ and $u_2$
  form a two-vertex component of $H- S$ and $v_1=\phi(u_1)$,
  $v_2=\phi(u_2)$, then $X=N_H(u_1)\cap S$ is a subset of
  $\phi^{-1}_0(N_{G}(v_1)\cap S^*)$ and $Y=N_H(u_2)\cap S$ is a subset
  of $\phi^{-1}_0(N_{G}(v_2)\cap S^*)$. Therefore, by construction,
  there is an edge of color $(\iota^{-1}(X),\iota^{-1}(Y))$ (or
  $(\iota^{-1}(Y),\iota^{-1}(X))$) between $v_1$ and $v_2$; let us add
  it to the matching $M$. Observe that $M$ is indeed a matching
  and contains exactly $n^1_i$ edges of color $i$ and exactly
  $n^2_{i,j}$ edges of color $(i,j)$.

  For the reverse direction, suppose that $M$ contains the required
  number of edges from each color. In particular, there are $n^1_i$
  edges of color $i$. Each edge of color $i$ is of the form $\{v,v'\}$ for
  some $v\in V(G)\setminus S^*$; let $V_i$ contain every such
  $v$. The fact that the edge $\{v,v'\}$ of color $i$ exists implies that
  $\iota(i)$ is a subset of $\phi^{-1}_0(N_G(v)\cap S^*)$. Therefore,
  if $u$ is a single-vertex component of $H- S$ with
  $N_H(u)=\iota(i)$, then $u$ can be mapped to any vertex $v\in V_i$,
  as $\phi_0(\iota(i))\subseteq N_G(v)\cap S^*$. Let us extend
  $\phi_0$ by mapping the $n^1_i$ such single-vertex components to
  $V_i$. Similarly, suppose that $u_1$ and $u_2$ form a two-vertex
  component of $H- S$ where the neighborhoods of the two
  vertices in $S$ are $\iota(i)$ and $\iota(j)$, respectively. Then
  $\{u_1,u_2\}$ can be mapped to any edge of $M$ with color
  $(i,j)$. This way, we can extend $\phi_0$ to a full subgraph model
  $\phi$: as $M$ is a matching, the images of distinct components of
  $H- S$ are disjoint.
\end{proof}

\subsection{Lower bound}
\label{sec:poly-lower-bound}
We prove Theorem~\ref{theorem:ramsey:matchingsplittable} characterizing
matching-splittable graphs by a relative simple application of
Ramsey's Theorem. (Recall that $P_3$ is the path on 3 vertices.)

\restateramseymatching*
\begin{proof}
Assume for contradiction that none of the four cases holds. Then there is an integer $Q\ge 1$ such that $\F$ does not contain any of $\Gclique{Q}$, $\Gbiclique{Q}$, $Q\cdot K_3$, or $Q\cdot P_3$.

Let~$r$ be the Ramsey number~$R(4Q, 2^9)$. Let~$H$ be a graph in~$\F$ that is not $3r$-matching-splittable; by the assumption that $\F$ is not matching-splittable, such an $H$ exists. Observe that if the maximum number of vertex-disjoint connected three-vertex graphs that can be packed in~$H$ is~$r$, then it is $3r$-matching-splittable: by maximality of the packing, the union of the vertices in the~$r$ graphs gives the required set of size $3r$. Therefore, if $H$ is not $3r$-matching-splittable, then there is a packing of~$r$ vertex-disjoint connected three-vertex subgraphs in~$H$. Let~$H_1, \ldots, H_r$ be such subgraphs and fix an arbitrary ordering~$v_{i,1}, v_{i,2}, v_{i,3}$ of the three vertices in each subgraph~$H_i$. If we consider two subgraphs~$H_i$ and~$H_{i'}$, then the adjacencies between the three vertices of~$H_i$ and the three vertices of~$H_{i'}$ are characterized exactly by the $3 \times 3$ adjacency matrix whose rows correspond to~$V(H_i)$ and whose columns correspond to~$V(H_j)$, with a one in cells corresponding to adjacent vertex pairs and a zero in the remaining cells. Since there are~$2^{3 \cdot 3} = 2^9$ different incidence matrices, the number of distinct ways in which two subgraphs~$H_i, H_j$ can be adjacent (under the chosen vertex ordering) is~$2^9$. 

We use these $3\times 3$ matrices to create an auxiliary graph~$F$ as follows: it is an $r$-vertex complete graph whose vertices are in correspondence with the subgraphs~$H_1, \ldots, H_r$. Number the~$2^9$ possible $3 \times 3$ adjacency matrices arbitrarily from~$1$ to~$2^9$ and give an edge~$\{i,j\} \in E(F)$ the color corresponding to the adjacency between~$H_i$ and~$H_j$. Since~$F$ has~$r = R(4Q, 2^9)$ vertices and its edges have been colored with~$2^9$ distinct colors, there is a monochromatic complete subgraph on~$4Q$ vertices. This subgraph of~$F$ corresponds to~$4Q$ subgraphs in the list~$H_1, \ldots, H_r$ that pairwise all have the same adjacencies to each other. Since there are only two different connected three-vertex graphs ($K_3$ and~$P_3$) this implies there are at least~$2Q$ pairwise vertex-disjoint, isomorphic subgraphs of~$H$ that all have the same adjacency to each other; denote the indices of these subgraphs by~$\I$. Suppose that the adjacency-type corresponding to the color of the monochromatic subgraph has at least one edge, i.e., that there are indices~$a,b \in [3]$ such that~$v_{i,a}$ is adjacent to~$v_{i',b}$ for all~$i \neq i' \in \I$. If this occurs for~$a = b$, then the set~$\{ v_{i,a} \mid a \in \I \}$ is a clique of size~$|\I| \geq 2Q$ in~$H$; but then, since~$\F$ is hereditary, family~\F contains the $\Gclique{2Q}$, which contradicts our choice of~$Q$. Now suppose that $a\neq b$ holds. Let $\I_1$ and $\I_2$ be two disjoint subsets of $\I$, each of size $Q$, and consider the vertex set~$\{v_{i,a} \mid i \in \I_1\} \cup \{v_{i,b} \mid i \in \I_2\}$. Since there are no edges between~$v_{i,a}$ and~$v_{i', a}$ for~$i,i' \in \I$ (by assumption that the previous case~$a=b$ does not apply), and similarly there are no edges between~$v_{i,b}$ and~$v_{i',b}$ for~$i,i' \in \I$, while all edges between~$v_{i,a}$ and~$v_{i',b}$ are present for~$i \in \I_1$ and~$i' \in \I_2$, the defined vertex set of size~$2Q$ induces a balanced biclique in~$H$. But then~\F contains $\Gbiclique{Q}$ (since~\F is hereditary), again contradicting our choice of~$Q$.

We may therefore conclude that vertices in different subgraphs~$H_i, H_{i'}$ with~$i \neq i' \in \I$ are not adjacent to each other. But then the vertices in the graphs~$H_i$ for~$i \in \I$ induce~$2Q$ disjoint copies of the same three-vertex graph, that is, either $Q\cdot K_3$ or $Q\cdot P_3$ is an induced subgraph of $H$, contradicting the choice of~$Q$.
\end{proof}

\subsection{Proof of the dichotomy for polynomial-time solvability of subgraph
  problems} \label{subsection:poly:theorem} By combining the algorithm of
Section~\ref{sec:poly-upper-bound}  and the
characterization proved in Section~\ref{sec:poly-lower-bound}, the proof of Theorem~\ref{theorem:intro:polysubgraph} follows.

\restateintropolysubgraph*
\begin{proof}
  Let $\F$ be a hereditary class of graphs. If $\F$ is
  matching-splittable, then Theorem~\ref{theorem:fsubgraphtest:alg}
  shows that \FSubgraphTest can be solved in randomized polynomial
  time.  If $\F$ is not matching-splittable, then $\F$ is the superset
  of one of the four classes listed in
  Theorem~\ref{theorem:ramsey:matchingsplittable}. If $\F$ is a superset
  of $\Fclique$, then the NP-hard \textsc{Clique} problem can be
  reduced to \FSubgraphTest and hence \FSubgraphTest is also
  NP-hard. Similarly, if $\F$ is a superset of $\Fbiclique$, then
  \textsc{Biclique} can be reduced to \FSubgraphTest. Suppose now that
  $\F$ contains $n\cdot K_3$ for every $n\ge 1$. Then
  $K_3$-\Packing, which is NP-hard by
  Theorem~\ref{theorem:fpacking:pvsnp}, can be reduced to
  \FSubgraphTest: an instance $(G,K_3,t)$ of
  $K_3$-\Packing can be expressed as an instance $(G,t\cdot
  K_3)$ of \FSubgraphTest. The situation is similar if $\F$
  contains $n\cdot P_3$ for every $n\ge 1$; note that
  $P_3$-\Packing is also NP-hard by
  Theorem~\ref{theorem:fpacking:pvsnp}. Therefore, we have shown that
  if $\F$ is not matching-splittable, then \FSubgraphTest is NP-hard,
  completing the proof Theorem~\ref{theorem:intro:polysubgraph}.
\end{proof}

\section{Computing representative sets for subgraph detection}
\label{sec:comp-repr-sets}

In this section, we present the main technology behind the
kernelization results of the paper: a marking algorithm that can be used
to find a representative set of small and thin bipartite subgraphs
that is sufficient for the solution. This marking algorithm is used in
the positive side of both Theorem~\ref{theorem:intro:packing} for
small/thin classes (Section~\ref{subsection:packing:upperbounds}) and
in the positive side of Theorem~\ref{theorem:intro:turingsubgraph} for
splittable classes
(Section~\ref{subsection:subgraph:turing:upperbounds}). The
application for  Theorem~\ref{theorem:intro:turingsubgraph} is more
general, as it involves the more general splittable proprety.  We
present the results in a way suitable for this more general
application. While this is more general than what is need for
Theorem~\ref{theorem:intro:packing}, it makes no sense to present two
versions of essentially the same algorithm.

\subsection{Finding thin bipartite subgraphs}
The definition of small/thin and splittable graphs involves components
of bounded size and thin bipartite graphs. Therefore, the very least,
we should be able to find such components efficiently. If a component
has at most $a$ vertices, then we can try all $n^a$ possible images in
$G$ by brute force. If a component is $b$-thin bipartite, then we can
try all $n^b$ possible images for the partite class containing at
most $b$ vertices and then we have to solve a bipartite matching
problem to find the location of the vertices in the larger class. The
follow lemma presents this reduction to matching in the slightly more
general context when the images of a set of vertices, including every
vertex in the smaller class, are already fixed.

\begin{lemma} \label{lemma:compute:extension:partial:specifiedbipartite:subgraph}
Let~$G$ be a graph, let~$H$ be a bipartite graph with partite sets~$A = \{a_1, \ldots, a_\alpha\}$ and~$B = \{b_1, \ldots, b_\beta\}$, and let~$\phi$ be a $P$-partial subgraph model of~$H$ with~$A \subseteq P \subseteq V(H)$. One can compute a full subgraph model of~$H$ in~$G$ that extends~$\phi$, or determine that no such model exists, in polynomial time by computing a maximum matching in a bipartite graph of order at most~$2|V(G)|$.
\end{lemma}
\begin{proof}
If~$|V(H)| > |V(G)|$ then obviously there is no full model of~$H$ in~$G$ and we output \no. Otherwise we proceed as follows. Let~$\N := \{N_H(b) \mid b \in B \setminus P\}$ be the set of different neighborhoods in~$A$ that have to be realized for the vertices that are not yet specified in the partial model. Since~$H$ is bipartite we have~$A' \subseteq A$ for all~$A' \in \N$. Number the sets in~$\N$ as~$\N = \{A'_1, \ldots, A'_t\}$. For each~$i \in [t]$ let~$n(i)$ be the number of vertices in~$B \setminus P$ whose neighborhood in~$H$ is exactly~$A'_i$. As~$A \subseteq P$, it follows that there is an extension of~$\phi$ to a full model if and only if we can assign to each set~$A'_i$ a set~$B'_i \subseteq \bigcap _{a \in A'_i} N_G(\phi(a)) \setminus \phi(P)$ of size~$n(i)$ such that the sets~$B'_i$ are pairwise disjoint. The crucial insight is that this condition can be checked by computing a matching in a related bipartite graph~$\hat{G}$ that is defined as follows.

The partite set~$\hat{B}$ of~$\hat{G}$ consists of the vertices~$V(G) \setminus \phi(P)$. The other partite set~$\hat{A}$ contains, for each set~$A'_i$ with~$i \in [t]$, exactly~$n(i)$ vertices~$w_{i,1}, \ldots, w_{i, n(i)}$ that are false twins in~$\hat{G}$. Each vertex~$w_{i,j}$ for~$j \in [n(i)]$ is adjacent in~$\hat{G}$ to~$\bigcap _{a \in A'_i} N_G(\phi(a)) \setminus \phi(P)$. Since~$|V(H)| \leq |V(G)|$, the order of~$\hat{G}$ is at most~$2|V(G)|$.

\begin{claim}
If there is a maximum matching in~$\hat{G}$ that saturates~$\hat{A}$, then~$\phi$ can be extended to a full $H$-subgraph model in~$G$: for each~$i \in [t]$, let~$v_{i,1}, \ldots, v_{i, n(i)} \subseteq V(H) \setminus P$ be the vertices whose $H$-neighborhood is~$A'_i$ and set~$\phi(v_{i,j})$ to the matching partner of~$w_{i,j}$. If a maximum matching in~$\hat{G}$ has size less than~$|\hat{A}|$ then there is no extension of~$\phi$ to a full subgraph model of~$H$ in~$G$.
\end{claim}
\begin{claimproof}
Suppose that~$M$ is a matching in~$\hat{G}$ saturating~$\hat{A}$. Consider the extension of~$\phi$ suggested in the claim. Then a vertex~$v_{i,j}$ is mapped to the matching partner of~$w_{i,j}$, which is a neighbor of~$w_{i,j}$ in~$\hat{G}$. As the neighbors of~$w_{i,j}$ in~$\hat{G}$ are exactly the vertices in~$\bigcap _{a \in A'_i} N_G(\phi(a)) \setminus \phi(P)$, this maps~$v_{i,j}$ to a vertex of~$G$ not yet used in the model that is adjacent to the $\phi$-model of all of~$v_{i,j}$'s neighbors in~$H$. As the matching ensures that all images assigned in this way are distinct, we obtain a valid $H$-subgraph model in~$G$.

We prove the second statement by contraposition: if there is an extension of~$\phi$ to a full subgraph model of~$H$, then there is a matching in~$\hat{G}$ saturating~$\hat{A}$. Suppose that~$\phi^*$ is a full $H$-subgraph model that extends~$\phi$. For each~$i \in [t]$ let~$v_{i,1}, \ldots, v_{i,n(i)}$ be as in the statement of the claim. Match~$w_{i,j} \in \hat{A}$ to~$\phi(v_{i,j}) \in \hat{B}$ for all~$i$ and~$j$ to obtain a matching saturating~$\hat{A}$, which has size~$|\hat{A}|$. This proves the claim.
\end{claimproof}

The claim shows how to construct a $H$-model that extends~$\phi$, if one exists. Since the construction of~$\hat{G}$ can be done in time polynomial in~$|V(G)|$ and bipartite matching is polynomial-time, for example using the Hopcroft-Karp algorithm~\cite{HopcroftK73}, this concludes the proof of Lemma~\ref{lemma:compute:extension:partial:specifiedbipartite:subgraph}.
\end{proof}

We now present the algorithm for finding a $b$-thin bipartite
graph. The following lemma formulates this in a more general way
suitable for use in $(a,b,c,d)$-splittable graphs: there is a set $D$
such that $H-D$ is a thin bipartite graph, and all but a bounded number
of vertices in $H-D$ are universal to $D$.
\begin{lemma} \label{lemma:compute:extension:partial:separated:subgraph}
Let~$G$ be a graph, let~$H$ be a graph with a (possibly empty) vertex set~$D \subseteq V(G)$ such that~$H' := H - D$ is a bipartite graph with partite sets~$A = \{a_1, \ldots, a_\alpha\}$ and~$B = \{b_1, \ldots, b_\beta\}$, and let~$\phi_0$ be a $P_0$-partial subgraph model of~$H$ in~$G$ with~$D \subseteq P_0$. Let~$B_N \subseteq B$ contain the vertices whose closed neighborhood in~$H$ is not universal to~$D$. One can compute a full subgraph model of~$H$ in~$G$ that extends~$\phi_0$, or determine that no such model exists, in time~$|V(G)|^{\Oh(1 + |(A \cup B_N) \setminus P_0|)}$.
\end{lemma}
\begin{proof}
We consider all~$\Oh(|V(G)|^{|A \cup B_N \setminus P_0|})$ possible ways to map the unspecified vertices~$(A \cup B_N) \setminus P_0$ to distinct vertices of~$V(G) \setminus P_0$. For each resulting $P'$-partial $H$-subgraph model~$\phi'$ we test whether it is valid, i.e., whether for all edges~$\{a,b\} \in E(H)$ with~$a, b \in P'$ we have~$\{\phi'(a), \phi'(b)\} \in E(G)$. If this is the case, then we want to invoke Lemma~\ref{lemma:compute:extension:partial:specifiedbipartite:subgraph} to determine whether~$\phi'$ can be extended to a full model of~$H$ that extends~$\phi'$ (and therefore~$\phi$). The crucial observation is that such an extension of~$\phi'$ exists if and only if the partial $H'$-subgraph model~$\phi'|_{V(H')}$ can be extended to a full model of~$H'$ in the graph~$G[\phi'(V(H')) \cup \bigcap _{v \in D} N_G(\phi(v))]$. This follows from the fact that all vertices of~$H'$ that are not assigned an image by~$\phi'$ are contained in~$B \setminus B_N$, are therefore universal in~$H$ to~$D$, and must therefore be mapped to members of~$\bigcap _{v \in D} N_G(\phi(v))$ by any extension. Hence by invoking Lemma~\ref{lemma:compute:extension:partial:specifiedbipartite:subgraph} we can test whether a particular choice of~$\phi'$ can be extended to a full model of~$H$ in~$G$ that extends~$\phi_0$. If any valid extension~$\phi'$ of~$\phi_0$ results in a full model then the first such model is given as the output. If Lemma~\ref{lemma:compute:extension:partial:specifiedbipartite:subgraph} never returns an extension then, since we try all possibilities for mapping~$(A \cup B_N) \setminus P_0$ to the free vertices~$V(G) \setminus \phi_0(P_0)$ in~$G$, no full model extending~$\phi_0$ exists and we output \no. The time bound follows from the fact that we try~$\Oh(|V(G)|^{|(A \cup B_N) \setminus P_0|})$ possibilities that can each be tested for feasibility in polynomial time.
\end{proof}

\subsection{Representative sets for small separators into unbalanced bipartite graphs}
\label{sec:repr-sets-thin}

This section contains the main technical part of the kernelization
algorithms: the marking algorithm for thin bipartite graphs. We
present it in a way suitable for $(a,b,c,d)$-splittable graphs: there
is a set $D$ such that $H-D$ is a thin bipartite graph, and all but a
bounded number of vertices in $H-D$ are universal to $D$.
\begin{lemma} \label{lemma:compute:representative:set:subgraph}
There is an algorithm with the following specifications. 
The input is a graph~$G$, a graph~$H$ with a (possibly empty) vertex set~$D \subseteq V(H)$ such that~$H' := H - D$ is a connected bipartite graph with partite sets~$A = \{a_1, \ldots, a_\alpha\} \neq \emptyset$ and~$B = \{b_1, \ldots, b_\beta\}$, an integer~$\ell$, and a partial subgraph model~$\phi_0$ of~$H$ with domain~$P_0 \supseteq D$. Define~$B_U := \{b \in B \mid \bigcap_{v \in N_{H'}[b]} N_H(v) \supseteq D\}$ and let~$B_N := B \setminus B_U$. Let~$h := |V(H)|$. The output is a set~$X \subseteq V(G)$ with the following properties.
\begin{enumerate}[(P1)]
	\item $|X| \leq (h^2 \ell + h^3)^{\mu+1}(1 + 2^{\alpha}(\ell + h))$, where~$\mu := |(A \cup B_N) \setminus P_0| + \sum _{a \in A} \max (0, \alpha - |N_H(a) \cap B_U \cap P_0|)$. \label{rep:set:subgr:size}
	\item For any vertex set~$Z \subseteq V(G)$ of size at most~$\ell$, if~$G - Z$ contains a full subgraph model~$\phi_1$ of~$H$ that extends~$\phi_0$, then~$G[X] - Z$ contains a full subgraph model~$\phi'_1$ of~$H$ with~$\phi'_1(v) = \phi_0(v)$ for all~$v \in D$.\label{rep:set:subgr:preservesmodel}
\end{enumerate}
The running time of the algorithm is polynomial in~$2^\alpha + |V(G)|^\mu$.
\end{lemma}
\begin{proof}
The algorithm is recursive and branches in a bounded number of directions to explore different ways in which the partial subgraph model~$\phi$ can be extended to a full model. The depth of the recursion is bounded by the measure~$\mu$, which will decrease in each recursive call. Throughout the proof we use the convention (see Section~\ref{section:preliminaries:graphs}) that if~$D = \emptyset$, then the common neighborhood of~$D$ in~$G$ equals~$V(G)$. We consider several cases. These correspond to the cases in the proof of Lemma~\ref{lemma:bicliquemarking:specialcase:general}, except that the first case in the lemma is not present here because it collapses into Case 2.b.

\textbf{Case 1: There is a vertex~$a \in A \setminus P_0$ with~$|N_H(a) \cap P_0 \cap B_U| \geq \alpha$.} In this case there is a vertex~$a$ whose model is not yet specified by~$\phi_0$, but for which many neighbors in the $B$-side of the bipartite graph~$H'$ already have been assigned an image. If there are many options for valid images for~$a$ (Case 1.a) then this implies the existence of a large biclique, which we can mark to preserve models of the bipartite graph~$H'$. If there are few options for valid images for~$a$ (Case 1.b) then we can branch into all possible options, recursively calculate a representative set, and output the union of these sets. Formally, we make another distinction. Let~$R := N_H(a) \cap P_0 \cap B_U$.

\textbf{Case 1.a: The set~$\phi(R \cup D)$ has at least~$\ell + h$ common neighbors in~$G$.} If the number of common neighbors~$\bigcap _{v \in R \cup D} N_G(\phi(v))$ is at least~$\ell + h$, then there is a large biclique in the common neighborhood of~$\phi_0(D)$ in~$G$: one side of the biclique is~$\phi(R)$ and the other side is formed by the common neighbors of~$\phi(R \cup D)$. 

Let~$X$ contain~$\phi(P_0)$ together with~$\ell + h$ common $G$-neighbors of~$\phi(R \cup D)$. It is easy to see that~\ref{rep:set:subgr:size} is satisfied. To see that~\ref{rep:set:subgr:preservesmodel} holds, observe the following. If~$\phi_1$ is a full $H$-model in~$G - Z$ that extends~$\phi_0$, then~$\phi(R \cup D)$ is disjoint from~$Z$. Build a full $H$-model~$\phi'_1$ as follows. Let~$\phi'_1(d) = \phi_0(d)$ for all~$d \in D$. Map~$A$ to~$\alpha$ arbitrary vertices in~$R$. Map~$B$ to vertices in the set~$\phi(R \cup D) \setminus Z$, which has size at least~$h - \alpha \geq \beta$. It is easy to verify that~$\phi'_1$ is a full $H$-model that coincides with~$\phi_0$ on the vertices of~$D$. In particular, all vertices~$\phi_1(V(H'))$ lie in the common neighborhood of~$\phi_1(D) = \phi_0(D)$, while~$G[\phi_1(H')]$ is a biclique with partite sets of sizes~$\alpha$ and~$\beta$.

\textbf{Case 1.b: The set~$\phi(R \cup D)$ has fewer than~$\ell + h$ common neighbors in~$G$.} Observe that, to extend~$\phi_0$ to a full subgraph model~$\phi_1$ of~$H$, the image of vertex~$a$ has to be adjacent in~$G$ to the images of all vertices in~$N_H(a) \cap P_0$. Hence~$\phi_1(a) \in \bigcap _{v \in R} N_G(\phi_0(v))$ for all valid extensions~$\phi_1$ of~$\phi_0$. Since~$|N_H(a) \cap P_0 \cap B_U| = |R| \geq \alpha \geq 1$ by the assumption of Case 1, the vertex~$a$ has a neighbor in~$H$ in the set~$B_U$. Since~$B_U$ contains the vertices whose \emph{closed neighborhood} in~$H$ is universal to~$D$, this implies that~$a$ is universal to~$D$ in~$H$. Hence in any valid extension~$\phi_1$ of~$\phi_0$ to a full $H$-subgraph model in~$G$, the image of~$a$ has to be contained in~$\bigcap _{v \in D} N_G(\phi_0(v))$. Hence~$\phi_1(a) \in T := \bigcap _{v \in R \cup D} N_G(\phi_0(v))$. By the assumption of Case 1.b we have~$|T| < \ell + h$ and hence there are only few options for valid images of~$a$. We now do the following.

Initialize~$X$ as an empty set. For each~$a' \in T$ we extend the partial subgraph model~$\phi_0$ to~$\phi_0^{a \mapsto a'}$ by setting~$\phi_0^{a \mapsto a'}(a) := a'$. We then recursively invoke the algorithm for the model~$\phi_0^{a \mapsto a'}$ and add the resulting set~$X^{a \mapsto a'}$ to~$X$. Since each such extension assigns strictly more $A$-vertices an image than~$\phi_0$, the measure~$\mu^{a \mapsto a'}$ of each recursive call is strictly smaller than the measure~$\mu$ for the current invocation. We can therefore bound the size of~$X$ as follows:
\begin{align*}
|X| &\leq \sum _{a' \in T} |X^{a \mapsto a'}| \leq \sum _{a' \in T} (h^2 \ell + h^3)^{(\mu-1)+1}(1 + 2^{\alpha}(\ell + h)) \leq |T| (h^2 \ell + h^3)^{(\mu-1)+1}(1 + 2^{\alpha}(\ell + h)) \\
&\leq (\ell + h) (h^2 \ell + h^3)^{(\mu - 1)+1}(1 + 2^{\alpha}(\ell + h)) \leq (h^2 \ell + h^3)^{\mu+1}(1 + 2^{\alpha}(\ell + h)),
\end{align*}
which satisfies~\ref{rep:set:subgr:size}. To see that~\ref{rep:set:subgr:preservesmodel} is satisfied, consider an arbitrary set~$Z \subseteq V(G)$ of size at most~$\ell$. If~$\phi_1$ is a full $H$-model in~$G - Z$ that extends~$\phi_0$, then by the observation above we have~$a' := \phi_1(a) \in T$. Hence~$\phi_1$ is a full $H$-model in~$G - Z$ that extends~$\phi_0^{a \mapsto a'}$. By the guarantee of the recursive call this implies that~$G[X^{a \mapsto a'}] - Z$ contains a full $H$-subgraph model~$\phi'_1$ with~$\phi'_1(v) = \phi_0(v)$ for all~$v \in D$. As~$X^{a \mapsto a'} \subseteq X$ this model also exists in~$G[X] - Z$, proving~\ref{rep:set:subgr:preservesmodel}.

\textbf{Case 2: All~$a \in A \setminus P_0$ satisfy~$|N_H(a) \cap P_0 \cap B_U| < \alpha$.} Let~$\hat{B}_U := \{b \in B_U \mid N_H(b) \subseteq P_0\}$. Intuitively, we can easily identify a small set of vertices to add to~$X$ in order to preserve images for the vertices of~$\hat{B}_U$, because the images of all $H$-neighbors of vertices in~$\hat{B}_U$ are already specified. On the other hand, one can verify that specifying the image of a vertex in~$V(H) \setminus (P_0 \cup \hat{B}_U)$ decreases the measure~$\mu$, which allows us to deal with those vertices by branching. We proceed as follows.

Let~$\hat{H} := H - \hat{B}_U$ and let~$\hat{\phi}_0$ be the partial model~$\phi_0$ restricted to the vertices of~$\hat{H}$; we interpret it as a partial~$\hat{H}$-subgraph model. By repeated application of Lemma~\ref{lemma:compute:extension:partial:separated:subgraph} we can compute a maximal packing~$\Phi$ of partial~$H - \hat{B}_U$-subgraph models in~$G$ that extend~$\hat{\phi}_0$, such that for distinct~$\phi, \phi' \in \Phi$ we have~$\phi(V(\hat{H}) \setminus P_0) \cap \phi'(V(\hat{H}) \setminus P_0) = \emptyset$. Informally, this says that the models in~$\Phi$ extend~$\hat{\phi}_0$ onto disjoint sets of vertices. (If~$V(\hat{H}) \subseteq P_0$ then there is nothing left to extend, and we define~$\Phi$ as a set containing~$\ell + h$ copies of~$\hat{\phi}_0$ to trigger Case 2.b.) The size of the computed packing~$\Phi$ of full~$\hat{H}$-subgraph models is the basis for the last case distinction.

\textbf{Case 2.a: $|\Phi| < \ell + h$.} If the packing~$\Phi$ contains less than~$\ell + h$ models of~$\hat{H}$ in~$G$ that extend~$\hat{\phi}_0$ to disjoint sets of vertices, then this is a \emph{maximal} packing and we can use it to guide a branching step into a bounded number of relevant extensions of~$\phi_0$. Define the set~$S := \bigcup _{\phi \in \Phi} \phi(V(\hat{H}) \setminus P_0)$. Since~$|\Phi| < \ell + h$, the size of~$S$ is bounded by~$(\ell + h)\cdot h$. Every $H$-subgraph model~$\phi_1$ that extends~$\phi_0$ contains a $\hat{H}$-subgraph model~$\hat{\phi}_1$ that extends~$\hat{\phi}_0$. By the maximality of~$\Phi$, such a model~$\hat{\phi}_1$ maps at least one vertex to~$S$. Since the extension~$\hat{\phi}_1$ maps all vertices of~$P_0$ in the same way as~$\hat{\phi}_0$, we find that all~$H$-subgraph models~$\phi_1$ that extend~$\phi_0$ map at least one vertex of~$V(\hat{H}) \setminus P_0$ to a vertex in~$S$; we will branch on all possibilities, as follows. For each vertex~$v \in V(\hat{H}) \setminus P_0$, for each vertex~$v' \in S$, we extend~$\phi_0$ to~$\phi_0^{v \mapsto v'}$ by setting~$\phi_0^{v \mapsto v'}(v) := v'$. If this results in a valid partial $H$-subgraph model (i.e., if~$v'$ is adjacent in~$G$ to all vertices in~$\phi_0(N_H(v) \cap P_0)$) then we recursively invoke the algorithm to compute a set~$X^{v \mapsto v'}$ that preserves $H$-subgraph models that extend~$\phi_0^{v \mapsto v'}$. We let~$X$ be the union of the resulting sets.

To bound the size of~$X$ we have to consider the measure associated with the recursive calls. Consider a recursive call for~$\phi_0^{v \mapsto v'}$, let~$P_0^{v \mapsto v'}$ be the domain of the partial model, and let~$\mu^{v \mapsto v'}$ be the measure of the recursive call. If~$v \in (A \cup B_N) \setminus P_0$ then~$|(A \cup B_N) \setminus P_0| > |(A \cup B_N) \setminus P_0^{v \mapsto v'}|$ since~$v \in ((A \cup B_N) \cap P_0^{v \mapsto v'}) \setminus P_0$. It then easily follows from the definition of measure that~$\mu > \mu^{v \mapsto v'}$. Let us now argue that in the other cases, the measure of the recursive calls is also strictly smaller than~$\mu$. 

If~$v \in (V(\hat{H}) \setminus P_0) \setminus (A \cup B_N)$ then, as~$\hat{H}$ does not contain the vertices~$\hat{B}$ and~$D \subseteq P_0$, it follows that~$v \in B_U \setminus \hat{B}$. By the definition of~$\hat{B}$ this implies that~$N_H(v)$ is not fully contained in~$P_0$. Since~$D \subseteq P_0$ by the precondition to the lemma, while~$v$ is contained in the~$B$-side of the bipartite graph~$H' = H - D$, this implies that there is a vertex~$a \in (A \cap N_H(v)) \setminus P_0$, i.e., that at least one $H$-neighbor of~$v$ is contained in the $A$-side of the bipartite graph and no image is specified for it by~$\phi_0$. By the precondition to Case 2 this implies that~$|N_H(a) \cap P_0 \cap B_U| < \alpha$. Let~$P_0^{v \mapsto v'}$ be the domain of~$\phi_0^{v \mapsto v'}$. Then the previous argumentation shows that~$\max (0, \alpha - |N_H(a)\cap B_U \cap P_0|) > \max (0, \alpha - |N_H(a) \cap B_U \cap P^{v \mapsto v'}_0|)$, which shows that the term for~$a$ in the measure of the recursive call to~$\phi^{v \mapsto v'}_0$ is strictly smaller than for the current call. It is easy to verify that the other terms in the measure do not increase. Hence the measure associated to all recursive calls is less than~$\mu$, guaranteeing that each recursive call outputs a set~$X^{v \mapsto v'}$ of size at most~$(h^2 \ell + h^3)^{(\mu-1)+1}(1 + 2^{\alpha}(\ell + h))$. This allows us to bound the size of~$X$ as follows:
\begin{align*}
|X| &\leq \sum _{v \in V(\hat{H}) \setminus P_0} \sum _{v' \in S} |X^{v \mapsto v'}| \leq \sum _{v \in V(\hat{H}) \setminus P_0} \sum _{v' \in S} (h^2 \ell + h^3)^{(\mu-1)+1}(1 + 2^{\alpha}(\ell + h)) \\
&\leq |V(\hat{H}) \setminus P_0| \cdot |S| \cdot (h^2 \ell + h^3)^{\mu}(1 + 2^{\alpha}(\ell + h)) \\
&\leq [h] \cdot [(\ell + h)h] \cdot (h^2 \ell + h^3)^{\mu}(1 + 2^{\alpha}(\ell + h)) \\
&\leq (h^2 \ell + h^3) (h^2 \ell + h^3)^{\mu}(1 + 2^{\alpha}(\ell + h)) \leq (h^2 \ell + h^3)^{\mu+1}(1 + 2^{\alpha}(\ell + h)).
\end{align*}
%
%
Hence~\ref{rep:set:subgr:size} is satisfied. To see that~\ref{rep:set:subgr:preservesmodel} is also satisfied, consider some set~$Z \subseteq V(G)$ and assume that~$G - Z$ contains a full $H$-subgraph model~$\phi_1$ that extends~$\phi_0$. As argued above,~$\phi_1$ maps at least one vertex of~$v \in V(\hat{H}) \setminus P_0$ to a vertex~$v' \in S$. Hence these is a choice of~$v$ and~$v'$ such that~$\phi_1$ is a full $H$-subgraph model in~$G - Z$ that extends a model~$\phi_0^{v \mapsto v'}$ on which we recursed. By the correctness guarantee for the recursive call, this implies that there is a full $H$-subgraph model~$\phi'_1$ in~$G[X^{v \mapsto v'}] - Z$ such that~$\phi'_1(u) = \phi_0(u)$ for all~$u \in D$. As~$X^{v \mapsto v'} \subseteq X$ this model also exists in~$G[X] - Z$, which shows that~\ref{rep:set:subgr:preservesmodel} is satisfied.

\textbf{Case 2.b: $|\Phi| \geq \ell + h$.} In the last case the set~$\Phi$ contains at least~$\ell + h$ partial $\hat{H}$-subgraph models in~$G$ that extend~$\phi_0$ onto pairwise disjoint sets. Recall that these partial models are trivial (they coincide with~$\phi_0$) in the case that~$V(\hat{H}) \subseteq P_0$. Let~$\Phi' \subseteq \Phi$ contain exactly~$\ell + h$ disjoint extensions of~$\hat{\phi}_0$ to full models of~$\hat{H}$. We initialize~$X$ as~$\bigcup _{\phi \in \Phi} \phi(V(\hat{H}))$. For each subset~$A' \subseteq A \cap P_0$, we add~$\ell + h$ vertices from the set~$\bigcap _{v \in A' \cup D} N_G(\phi_0(v))$ to~$X$ (or all such vertices, if the number of common neighbors is less than~$\ell + h$). The result is output as the set~$X$. It is easy to see that~$|X| \leq (\ell + h)h + 2^{|A \cap P_0|}(\ell + h) \leq (h \ell + h^2) + 2^{\alpha}(\ell + h)$, which satisfies~\ref{rep:set:subgr:size} since~$\mu \geq 0$. 

It remains to argue that~\ref{rep:set:subgr:preservesmodel} holds. Let~$Z \subseteq V(G)$ be a set of size at most~$\ell$ and assume that~$G - Z$ contains a full $H$-subgraph model~$\phi_1$ that extends~$\phi_0$. We build a full $H$-subgraph model~$\phi'_1$ in~$G[X] - Z$, as follows. 

\begin{itemize}
	\item For all vertices~$v \in (A \cup D) \cap P_0$ we set~$\phi'_1(v) := \phi_0(v) = \phi_1(v)$; the last equality holds since~$\phi_1$ extends~$\phi_0$. As~$D \subseteq P_0$ by the precondition to the lemma, it remains to define an image for the vertices in~$A \setminus P_0$ and for the vertices in~$B$. As~$\phi_0$ is a partial subgraph model of~$H$, whenever there is an edge in~$H$ between vertices of~$(A \cup D) \cap P_0$, these assignments ensure there is an edge between the images of the vertices as well. Since the image of~$\phi_1$ is disjoint from~$Z$, no vertex is mapped to a member of~$Z$ in this step. As~$\Phi$ is not empty and for any~$\phi \in \Phi$ the set~$\phi((A \cup D) \cap P_0) = \phi_1((A \cup D) \cap P_0)$ is contained in~$X$, the image of the partial model~$\phi'_1$ constructed so far lies in~$G[X] - Z$.
	\item For all vertices~$b \in \hat{B}$, do the following. If~$\phi_1(b) \in X$ then let~$\phi'_1(b) := \phi_1(b)$. Now consider what happens if~$\phi_1(b) \not \in X$. By definition of~$\hat{B}$ we have~$N_H(b) \subseteq P_0$. As~$\hat{B} \subseteq B_U$ we know by definition of the latter set that~$D \subseteq N_H(b)$. Since~$H' = H - D$ is bipartite and~$b$ is contained in its $B$-side, there is a set~$A' \subseteq A$ such that~$N_H(b) = D \cup A'$. As~$\phi_1$ forms a valid $H$-subgraph model that extends~$\phi_0$ we have~$\phi_1(b) \in \bigcap _{v \in N_H(b)} \phi_1(v) = \bigcap _{v \in D \cup A'} \phi_0(v)$. While constructing the set~$X$ we have considered the set~$A' \subseteq A \cap P_0$ and have added common neighbors of~$\phi_0(A' \cup D)$ to~$X$. Since~$\phi_1(b) \not \in X$, it follows that the number of common $G$-neighbors of~$\phi_0(A' \cup D)$ exceeded~$\ell + h$ and~$\phi_1(b)$ was not chosen to be added to~$X$. This implies that~$(X \cap \bigcap _{v \in A' \cup D} N_G(\phi_0(v))) \setminus Z$ has size at least~$h$, as~$|Z| \leq \ell$. Let~$b'$ be a vertex of~$(X \cap \bigcap _{v \in A' \cup D} N_G(\phi_0(v))) \setminus Z$ that is not in the image of the partial model~$\phi'_1$ that is under construction. As~$\phi'_1$ does not yet specify an image for~$b$, its image contains at most~$h - 1$ vertices, which implies that such a vertex~$b'$ exists. Then define~$\phi'_1(b) := b'$. By the choice of~$b$ this image is contained in~$X \setminus Z$ and is adjacent to the images of all of~$b$'s neighbors in~$H$.
	\item After having defined~$\phi'_1$ on all vertices~$(A \cup D) \cap P_0$ and~$\hat{B}$ in the previous two steps, it remains to define~$\phi'_1$ on~$V(\hat{H}) \setminus P_0$. If~$V(\hat{H}) \setminus P_0 = \emptyset$ then the model is already completed and we are done. In the remaining case, the size of the current image of~$\phi'_1$ under construction is strictly less than~$h$. Observe that the set~$\Phi'$ contained exactly~$\ell + h$ partial $\hat{H}$-subgraph models such that for all distinct~$\phi, \phi' \in \Phi'$ we have~$\phi(V(\hat{H}) \setminus P_0) \cap \phi'(V(\hat{H}) \setminus P_0) = \emptyset$. Since~$Z$ has size at most~$\ell$, while the current image of~$\phi'_1$ has size less than~$h$, there is a full $\hat{H}$-subgraph model~$\phi^* \in \Phi$ such that~$\phi^*(V(\hat{H}) \setminus P_0)$ contains no vertex of~$Z$ and contains no vertex of the current image of~$\phi'_1$. We use~$\phi^*$ to define~$\phi'_1$ on the vertices in~$V(\hat{H}) \setminus P_0$: for each vertex~$v \in V(\hat{H}) \setminus P_0$ set~$\phi'_1(v) := \phi^*(v)$. Since~$X$ contains~$\phi(V(\hat{H}))$, all images assigned in this step are contained in~$X$. As~$\phi^*$ is a full $\hat{H}$-subgraph model that extends~$\phi_0$, its images realize all the edges of~$\hat{H}$. The only vertices of~$V(H) \setminus V(\hat{H})$ are those in~$\hat{B}$. As the vertices in~$\hat{B}$ have all their $H$-neighbors in~$P_0$, which~$\phi^*$ assigns the same images as~$\phi_0$, all edges of~$H$ are realized by the partial model~$\phi'_1$, which is contained in~$G[X] - Z$.
\end{itemize}

\noindent The construction of~$\phi'_1$ proves that~\ref{rep:set:subgr:preservesmodel} holds, which concludes the argumentation for Case 2.b. 

It is easy to see that the running time of the algorithm is polynomial in~$2^\alpha + |V(G)|^\mu$. There are two nontrivial computational steps: marking common neighbors for all subsets of~$A \cap P_0$ in Case 2.2 (which contributes the~$2^{\alpha}$ factor), and the invocation of Lemma~\ref{lemma:compute:extension:partial:separated:subgraph} for constructing the set~$\Phi$ in Case 2. As the time bound for Lemma~\ref{lemma:compute:extension:partial:separated:subgraph} is dominated by~$|V(G)|^{\mu}$, the algorithm obeys the claimed size bound. As the case distinction is exhaustive, this finishes the proof of Lemma~\ref{lemma:compute:representative:set:subgraph}.
\end{proof}

\subsection{Representative sets for small separators into small components}
The proof of the following lemma can be considered as a straight-forward application of the sunflower lemma. We provide its proof for completeness, and to illustrate how it provides an analogue of Lemma~\ref{lemma:compute:representative:set:subgraph} for constant-size graphs~$H$.

\begin{lemma} \label{lemma:compute:representative:set:constanth}
There is an algorithm with the following specifications. The input is a graph~$G$, a graph~$H$ with a (possibly empty) vertex set~$D \subseteq V(H)$, an integer~$\ell$, and a partial subgraph model~$\phi_0$ of~$H$ with domain~$P_0 \supseteq D$. Let~$\alpha := |V(H) \setminus P_0|$ and let~$h := |V(H)|$. The output is a set~$X \subseteq V(G)$ with the following properties.
\begin{enumerate}[(P1)] 
	\item $|X| \leq h + (\alpha+1)! \cdot (\ell+1)^{\alpha}$.\label{rep:set:constanth:size}
	\item For any vertex set~$Z \subseteq V(G)$ of size at most~$\ell$, if~$G - Z$ contains a full subgraph model~$\phi_1$ of~$H$ that extends~$\phi_0$, then~$G[X] - Z$ contains a full subgraph model~$\phi'_1$ of~$H$ that extends~$\phi_0$.\label{rep:set:constanth:preservesmodel}
\end{enumerate}
The running time of the algorithm is~$|V(G)|^{\Oh(\alpha)}$.
\end{lemma}
\begin{proof}
Define a system of sets~$\S$ as follows. For each full $H$-subgraph model~$\phi_1$ in~$G$ that extends~$\phi_0$, add a set~$S_{\phi_1} := \phi_1(V(H) \setminus P_0)$ to~$\S$ if it was not contained in~$\S$ already. The time bound allows us to try all images of~$V(H) \setminus P_0$ by brute force; we thus obtain the family~$\S$ in time~$|V(G)|^{\Oh(\alpha)}$.

We will compute a set~$\S' \subseteq 2^{V(G)}$ with the following \emph{preservation property}: for every set~$Z \subseteq V(G)$ of size at most~$\ell$, if~$G - Z$ contains a full $H$-subgraph model~$\phi_1$ that extends~$\phi_0$, then~$G-Z$ contains a full $H$-subgraph model~$\phi'_1$ that extends~$\phi_0$ and satisfies~$\phi'_1(V(H) \setminus P_0) \in \S'$. By our choice of~$\S$ it is clear that~$\S$ has the preservation property, so we initialize~$\S'$ as a copy of~$\S$. 

\begin{claim}
If~$\S' \subseteq \S$ has the preservation property and~$|\S'| \geq \alpha!(\ell+1)^{\alpha}$, we can identify a set~$S^* \in \S'$ in time~$|V(G)|^{\Oh(\alpha)}$ such that~$\S' \setminus \{S^*\}$ also has the preservation property. 
\end{claim}
\begin{claimproof}
Suppose that~$|\S'| \geq \alpha!(\ell+1)^{\alpha}$. By Lemma~\ref{lemma:sunflowers} there is a sunflower in~$\S'$ consisting of at least~$\ell+2$ sets~$S_1, \ldots, S_{\ell+1}, S_{\ell+2}$ and this can be found in time polynomial in the size of the set family and the universe. Let~$C := \bigcap _{i=1}^{\ell+2} S_i$ be the core of the sunflower. We show that~$\S' \setminus \{S_1\}$ has the preservation property.

Let~$Z \subseteq V(G)$ have size at most~$\ell$ and suppose that~$\phi_1$ is a full $H$-subgraph model in~$G - Z$ that extends~$\phi_0$. As~$\S'$ has the preservation property, there is a full $H$-subgraph model~$\phi'_1$ in~$G-Z$ that extends~$\phi_0$ such that~$S_{\phi'_1} := \phi'_1(V(H) \setminus P_0)$ is contained in~$\S'$. If~$S_{\phi'_1} \neq S_1$ then the set~$S_{\phi'_1}$ is contained in~$\S' \setminus \{S_1\}$ which establishes the preservation property. If~$S_1 = S_{\phi'_1}$, then we proceed as follows. Since~$\phi'_1$ is an $H$-subgraph model in~$G - Z$, we have~$\phi'_1(V(H) \setminus P_0) \cap Z = S_1 \cap Z = \emptyset$ which shows in particular that none of the vertices in the core~$C$ of the sunflower are contained in~$Z$. Since the petals~$S_1 \setminus C, \ldots, S_{\ell + 2} \setminus C$ of the sunflower are pairwise disjoint, the set~$Z$ of size at most~$\ell$ can intersect at most~$\ell$ sets among~$S_2 \setminus C, \ldots, S_{\ell + 2} \setminus C$, implying that there is an~$i > 1$ such that~$(S_i \setminus C) \cap Z = \emptyset$ which implies~$S_i \cap Z = \emptyset$. Since~$i > 1$ this set~$S_i$ is contained in~$\S' \setminus \{S_1\}$. If~$\phi_1$ is a model in~$G-Z$ that extends~$\phi_0$ then~$\phi_0(P_0) \cap Z = \emptyset$. Let~$\phi_i$ be the $H$-subgraph model extending~$\phi_0$ that caused the set~$S_i$ to be added to~$\S$. Then~$\phi_i(V(H)) = \phi_i(V(H) \setminus P_0) \cup \phi_i(P_0) = S_i \cup \phi_i(P_0)$. As both sets are disjoint from~$Z$, we conclude that~$\phi_i$ is a full $H$-subgraph model in~$G-Z$ that extends~$\phi_0$ and satisfies~$\phi_i(V(H) \setminus P_0) \in \S' \setminus \{S_1\}$. Hence~$\S' \setminus \{S_1\}$ has the preservation property.
\end{claimproof}

By iterating the argument above, we arrive at a set system~$\S' \subseteq \S$ with the preservation property that contains at most~$\alpha!(\ell+1)^\alpha$ sets. The preservation property directly implies that picking~$X := \phi_0(P_0) \cup \bigcup _{S \in \S'} S$ satisfies~\ref{rep:set:constanth:preservesmodel}. Since each set in~$\S$ has size at most~$\alpha$ we find that~$|X| \leq |P_0| + \alpha \cdot \alpha!(\ell+1)^{\alpha}$, which satisfies~\ref{rep:set:constanth:size}. Since each iteration can be done in time polynomial in the size of~$\S$ and~$G$, while the number of iterations is bounded by~$|\S|$ which is~$\Oh(|V(G)|^\alpha)$, this concludes the proof.
\end{proof}

\subsection{General representative sets for subgraph detection}
In this section, we put together
Lemma~\ref{lemma:compute:representative:set:subgraph} and
\ref{lemma:compute:representative:set:constanth} to obtain a marking
procedure for $(a,b,c,d)$-splittable graphs. We show that if $D$
realizes the $(a,b,c,d)$-split and we invoke the marking procedure
separately for each component of $H-D$, then the union of the marked
sets contains a solution, if exists.

\begin{lemma} \label{lemma:kernel:generic}
For every triple of constants~$a,b,d \in \mathbb{N}$ there is an algorithm with the following specifications. The input is a graph~$G$, a graph~$H$ with a (possibly empty) vertex set~$D \subseteq V(H)$ that realizes an~$(a,b,|D|,d)$ split of~$H$, and a partial subgraph model~$\phi_0$ of~$H$ with domain~$P_0 = D$.
Let~$k := |V(H)|$. The output is a set~$X \subseteq V(G)$ of size~$\Oh(k^{\Oh(a + b^2 + d)})$ such that~$G$ contains a full $H$-subgraph model that extends~$\phi_0$ if and only if~$G[X]$ contains a full $H$-subgraph model that extends~$\phi_0$. The running time of the algorithm is polynomial in~$|V(G)|$ for fixed~$a,b$, and~$d$.
\end{lemma}
\begin{proof}
Fix~$a,b,d \in \mathbb{N}$ and let~$(G,H,D,\phi_0, P_0)$ be an input satisfying the requirements. Let~$C_1, \linebreak[1] \ldots, C_r$ be the connected components of~$H - D$ and define~$C'_i$ as the subgraph of~$H$ induced by~$N_H[V(C_i)]$ (i.e.,~$C'_i$ is component~$C_i$ together with its neighbors in~$D$). We build a set~$X$ as follows. Initialize~$X$ as~$\phi_0(D)$ and do the following for each~$i \in [r]$.

\begin{itemize}
	\item If~$|V(C_i)| \leq a$ then we invoke the algorithm of Lemma~\ref{lemma:compute:representative:set:constanth} with parameters~$\hat{G},\hat{H},\hat{D},\hat{\ell},\hat{\phi}_0$ chosen as follows. The source graph~$\hat{G}$ equals~$G$, the query graph~$\hat{H}$ is set to~$C'_i$, the set~$\hat{D}$ is~$V(C'_i) \cap D$, the value of~$\hat{\ell}$ is~$k - |V(C'_i)|$, and~$\hat{\phi}_0$ is the restriction~$\phi_0|_{V(C'_i)}$. Under these definitions, the value of~$\hat{\alpha}$ defined in the lemma is~$|V(C'_i) \setminus (V(C'_i) \cap D)| = |V(C_i)| \leq a$. The algorithm outputs a set~$X_{C'_i}$ of size at most~$|V(C'_i)| + (a + 1)! \cdot (k - |V(C'_i)| + 1)^a$ which is~$\Oh(k^a)$ since~$a$ is a constant. We add~$X_{C'_i}$ to~$X^*$.
	\item If~$|V(C_i)| > a$, then by the precondition~$C_i$ is a $b$-thin bipartite graph. Let~$A$ and~$B$ the partite classes of~$H$ and choose~$A$ to be the smallest class, which has size at most~$b$. The precondition ensures that the number of vertices in~$C'_i$ that are not universal to~$V(C'_i) \cap D$ in the graph~$C'_i$ is at most~$d$. We invoke the algorithm of Lemma~\ref{lemma:compute:representative:set:subgraph} with parameters~$\hat{G} := G, \hat{H} := C'_i, \hat{D} := V(C'_i) \cap D, \hat{A} := A, \hat{B} := B, \hat{\ell} := k - |V(C'_i)|$, and we let~$\hat{\phi}_0$ be the restriction~$\phi_0|_{V(C'_i)}$. Therefore the value of~$\hat{\mu}$ defined in the lemma is bounded by~$|\hat{A}| + d + \sum _{v \in \hat{A}} |\hat{A}| \leq b + d + b^2$. By~\ref{rep:set:subgr:size} the size of the set~$X_{C'_i}$ computed by the described algorithm is bounded by~$(|V(C'_i)|^2 (k - |V(C'_i)|) + |V(C'_i)|^3)^{\hat{\mu} + 1} \cdot (1 + 2^b (k - |V(C'_i)| + |V(C'_i)|))$. Since~$b$ is a constant and~$|V(C_i)| \leq k$ we find that~$|X'_{C_i}| \in \Oh(k^{\Oh(\hat{\mu})}) \in \Oh(k^{\Oh(b^2 + d)})$. We add~$X'_{C_i}$ to~$X^*$.
\end{itemize}

Observe that the number of connected components of~$H$ does not exceed its order~$k$. Hence~$X^*$ is the union of~$\phi_0(D)$ with at most~$k$ sets that each have size~$\Oh(k^a + k^{\Oh(b^2 + d)})$, which proves that~$|X^*|$ is bounded by a fixed polynomial in the parameter~$k$. It is easy to verify that the running time of the procedure is polynomial in~$|V(G)|$ for fixed~$a,b$, and~$d$. To prove the correctness of the procedure, it remains to prove that~$G$ contains an $H$-subgraph model that extends~$\phi_0$ if and only if~$G[X]$ does. As~$G[X]$ is an induced subgraph of~$G$, the implication from right to left is trivial. We therefore consider the case that~$G[X]$ contains a full~$H$-subgraph model~$\phi^*$ that extends~$\phi_0$ and proceed to prove that~$G[X]$ contains such an extension as well. We construct an $H$-subgraph model~$\phi'$ in~$G[X]$ that extends~$\phi_0$ by repeatedly moving the $\phi^*$-model of one connected component of~$H - D$ into the set~$G[X]$. This is formalized by the following claim.

\begin{claim}
If~$\phi^*$ is a full $H$-subgraph model in~$G$ that extends~$\phi_0$, then for each~$i$ with~$0 \leq i \leq r$ the graph~$G$ contains a full $H$-subgraph model~$\phi_1$ that extends~$\phi_0$ and satisfies~$\phi_1(V(C'_j)) \subseteq X$ for all~$1 \leq j \leq i$.
\end{claim}
\begin{claimproof}
We use induction on~$i$. For~$i=0$ the model~$\phi^*$ satisfies the stated condition. Consider some~$i > 0$ and assume by the induction hypothesis that~$\phi'_1$ is a full $H$-subgraph model that extends~$\phi_0$ and satisfies~$\phi'_1(V(C'_j)) \subseteq X$ for all~$1 \leq j \leq i - 1$. During the construction of~$X$ we considered the component~$C'_i$ and computed a set~$X_{C'_i}$ by invoking Lemma~\ref{lemma:compute:representative:set:constanth} or Lemma~\ref{lemma:compute:representative:set:subgraph}. Both lemmata guarantee that the set~$X_{C'_i}$ they output satisfies the following: if~$Z \subseteq V(G)$ has size at most~$\hat{\ell} = k - |V(C'_i)|$ and~$G - Z$ contains a full $C'_i$-subgraph model that extends~$\phi_0|_{V(C'_i)}$, then~$G[X_{C'_i}] - Z$ contains a full~$C'_i$-subgraph model~$\phi_{C'_i}$ that agrees with~$\phi_0|_{P_0 \cap V(C'_i)}$ on~$V(C'_i) \cap D$. 

Now choose~$Z := \phi'_1(V(H) \setminus V(C'_i))$, which has size~$k - |V(C'_i)|$, and observe that~$\phi'_1|_{V(C'_i)}$ is a full~$C'_i$ model in~$G - Z$ that extends~$\phi_0|_{V(C'_i)}$. Hence the left hand side of the implication mentioned above is satisfied and~$G[X_{C'_i}] - Z \subseteq G[X] - Z$ contains a full~$C'_i$-subgraph model~$\phi_{C'_i}$ that agrees with~$\phi_0|_{P_0 \cap V(C'_i)}$ on~$V(C'_i) \cap D$. Now build a model~$\phi_1$ as described in the claim, as follows. 

Initialize~$\phi_1$ as~$\phi'_1|_{V(H) \setminus V(C'_i)}$; by the induction hypothesis this maps the vertices of the first~$i - 1$ components into~$X$. Then augment the model~$\phi_1$ by mapping the vertices of~$V(C'_i)$ in the same way as~$\phi_{C'_i}$. Model~$\phi_{C'_i}$ maps all vertices into~$X - Z$, which implies that the vertices of~$C'_i$ are not mapped to vertices that are used in the model~$\phi'_1|_{V(H) \setminus V(C'_i)}$ as they are contained in~$Z$. Hence this extension results in an~$H$ subgraph model~$\phi_1$ in~$G[X]$ that maps the first~$i$ components into~$X$. It is a valid $H$-subgraph model since we are combining two partial models~$\phi_{C'_i}$ and~$\phi'_1|_{V(H) \setminus V(C_i)}$ for subgraphs~$C'_i$,~$V(H) \setminus V(C_i)$ of~$H$ that agree on the mapping of the separator~$V(C'_i) \cap D$, which may be verified using Observation~\ref{observation:merge:models}. To see that~$\phi_1$ extends~$\phi_0$, observe that~$\phi_1$ agrees with~$\phi_0$ on~$P_0 \setminus V(C'_i)$ as we copied the behavior on these vertices from~$\phi'_1$, which extends~$\phi_0$; it agrees with~$\phi_0$ on~$P_0 \cap V(C'_i)$ since we copied the behavior on these vertices from~$\phi_{C'_i}$ which agrees with~$\phi_0|_{P_0 \cap V(C'_i)}$ on~$V(C'_i) \cap D$. Hence it agrees with~$\phi_0$ on~$P_0$ and therefore extends~$\phi_0$.
\end{claimproof}

The case~$i = r$ of the claim proves that~$X$ preserves the existence of $H$-subgraph models extending~$\phi_0$. Note that a model as constructed in the claim for~$i = r$ maps all vertices of~$D$ into~$X$, as the model extends~$\phi_0$ while~$\phi_0(D) \subseteq X$ by construction. This concludes the proof of Lemma~\ref{lemma:kernel:generic}.
\end{proof}


\section{Kernelization complexity of packing problems}
In this section, we prove Theorem~\ref{theorem:intro:packing}
characterizing the hereditary classes $\F$ for which \FPacking admits a
polynomial many-one or Turing kernel; the outcomes coincide for \FPacking. In
Section~\ref{subsection:packing:upperbounds}, we invoke the marking
algorithm developed in Section~\ref{sec:comp-repr-sets} to give a
polynomial kernel for small/thin classes. In
Section~\ref{subsection:packing:lowerbounds}, we establish a series of
\textup{WK[1]}-hardness results for packing various basic classes of graphs. In
Section~\ref{sec:packing-ramsey}, we prove Theorem~\ref{theorem:ramsey:1}
characterizing small/thin hereditary classes. In
Section~\ref{subsection:packing:theorem}, we put together all these
results to complete the proof of Theorem~\ref{theorem:intro:packing}.

\subsection{Upper bounds} \label{subsection:packing:upperbounds}
 To prove the
positive part of Theorem~\ref{theorem:intro:packing} for
small/thin hereditary classes, we obtain a kernel by invoking the marking algorithm of 
Lemma~\ref{lemma:kernel:generic} to obtain a bounded-size instance.

\restatepackingkernel*
\begin{proof}
Let~$\F$ be a hereditary family for which such constants~$a$ and~$b$ exist. We show how to derive a polynomial many-one kernel for \kFPacking. An instance of \kFPacking consists of a tuple~$(G, H \in \F, t \in \mathbb{N})$ that asks whether~$G$ contains~$t$ vertex-disjoint subgraphs isomorphic to~$H$. Recall that the parameter is~$k := t \cdot |V(H)|$. Presented with an instance~$(G,H,t)$, the kernelization algorithm invokes Lemma~\ref{lemma:kernel:generic} to the graph~$G$, using~$t \cdot H$ as the query graph, an empty set~$D$, and an empty model~$\phi_0$. For fixed~$a$ and~$b$ the computation outputs a set~$X^* \subseteq V(G)$ in polynomial time. For the kernelization we output the instance~$(G[X^*], H, t)$. The lemma guarantees that~$|X^*| \in \Oh(k^{\Oh(a + b^2)})$ and that~$G[X^*]$ contains a $t \cdot H$-subgraph if and only if~$G$ does. Hence this procedure forms a correct kernelization.
\end{proof}

\subsection{Lower bounds} \label{subsection:packing:lowerbounds}

In this section we present the polynomial-parameter transformations that establish Theorem~\ref{theorem:packing:lowerbounds}, which we repeat here for the reader's convenience.

\restatepackinglower*

The proofs are ordered in increasing difficulty. For each of the graph families mentioned in the theorem, we give a polynomial-parameter transformation from the \nRegularExactSetCover problem, whose \textup{WK[1]}-hardness we established in Lemma~\ref{lemma:regularexactsetcover:wkhard}.

\begin{lemma} \label{lemma:fountainpacking:wkhard}
\textsc{\Ffountain{s}-Packing} is \textup{WK[1]}-hard for any odd integer~$s \geq 3$. Similarly, \textsc{\Flongfountain{s}{t}-Packing} is \textup{WK[1]}-hard for any odd integer~$s \geq 3$ and any integer~$t \geq 1$.
\end{lemma}
\begin{proof}
Let~$(r, \S, U, n)$ be an instance of \nRegularExactSetCover as described in Section~\ref{section:wkhardness}. We show how to construct an equivalent instance of a packing problem as described in the lemma statement in polynomial time, with a parameter that is polynomial in~$n$. An instance of the packing problem consists of a host graph~$G$, a pattern graph~$H$ that is contained in the relevant family~$\F$, and an integer~$t$. The construction we present works simultaneously for the family \Ffountain{s} for any odd~$s \geq 3$ and for the family \Flongfountain{s}{t} for any odd~$s \geq 3$ and~$t \geq 1$. To make the construction generic, let~$H$ be the (long) fountain in the class~$\F$ that we are targeting where the unique high-degree vertex of the (long) fountain has exactly~$r$ pendant vertices attached to it. The instance resulting from the transformation will ask for a packing of~$t := n / r$ copies of~$H$. Let~$c$ be the unique high-degree vertex in~$H$ to which the pendant vertices are attached. Let~$H'$ be the graph~$H$ without the~$n$ pendant vertices. The host graph~$G$ is defined as follows.

\begin{itemize}
	\item Graph~$G$ contains the vertex set~$U$ as an independent set. Observe that we identify the elements of the universe~$U$ with a subset of the vertices of~$G$ by this definition.
	\item For every set~$S \subseteq U \in \S$, which has size exactly~$r$, we add a new copy~$H'_S$ of the graph~$H'$ to~$G$. Let~$c_S \in V(H'_S)$ be the copy of the vertex~$c$. We add edges between~$c_S$ and the vertices representing~$S$. This concludes the construction of~$G$.
\end{itemize}

It is easy to see that the construction can be carried out in polynomial time. The parameter~$k$ of the problem is defined as~$k := t \cdot |V(H)|$. Since~$|V(H)| \in \Oh(r)$ (the exact value depends on the constant~$s$ and possibly on~$t$) and~$t = n/r$ we have~$k \in \Oh(n)$, which is polynomially bounded in the parameter of the input instance. To complete the polynomial-parameter transformation it remains to prove that the instance~$(r, \S, U, n)$ of \nRegularExactSetCover is equivalent to the constructed instance~$(G,H,t)$ of \kFSubgraphTest.

\begin{claim}
If~$(\S,U)$ has an exact cover, then~$t \cdot H \subseteq G$.
\end{claim}
\begin{claimproof}
Assume that~$(\S,U)$ has an exact cover, which must consist of~$t=n / r$ distinct and disjoint sets~$S_1, \ldots, S_t \in \S$. For each~$S_i$ with~$i \in [t]$ the copy~$H'_{S_i}$ of~$H'$ in~$G$ together with the~$r$ vertices~$S_i$ in~$G$ form a~$H$-subgraph in~$G$. Since the sets~$S_i$ are pairwise disjoint, there are~$t$ vertex-disjoint $H$-subgraphs in~$G$. Hence~$t \cdot H$ is a subgraph of~$G$.
\end{claimproof}

Before proving the reverse direction, we establish a structural claim.

\begin{claim} \label{claim:fountain:hardness:goodmodels}
Any~$H$ subgraph in~$G$ consists of a subgraph~$H'_S$ for some~$S \in \S$ together with the~$r$ vertices in~$N_G(c_S) \cap U$.
\end{claim}
\begin{claimproof}
Recall that~$s \geq 3$ is odd. The only simple odd cycles in~$G$ are the copies of the unique length-$s$ cycle in~$H'$. To see this, pick an arbitrary edge~$e$ on the length-$s$ cycle in~$H'$ and let~$Y = \{e_1, \ldots, e_{|\S|}\}$ be the copies of this edge in the graphs~$H'_S$ for~$S \in \S$. Deleting the edges of~$Y$ from~$G$ results in a bipartite graph, which may be verified by noting that the graph~$G - Y$ can be reduced to the graph~$G[U \cup \{c_S \mid S \in \S\}]$ by repeatedly removing degree-one vertices (which does not remove any odd cycles), while~$G[U \cup \{c_S \mid S \in \S\}]$ is clearly bipartite as both~$G$ and~$\{c_S \mid S \in \S\}$ form independent sets in~$G$. So all odd cycles in~$G$ use an edge in~$Y$, and since each edge~$e_i \in Y$ is contained in a unique simple odd cycle (the cycle cannot be made longer without repeating a vertex, see Figure~\ref{fig:basic}) this shows that the only simple odd cycles in~$G$ are the length-$s$ cycles in the copies of~$H'$.

Hence the only way an $H$-subgraph in~$G$ can realize a the length-$s$ cycle is to map it to a length-$s$ cycle in one of the copies~$H'$. Since the copies of~$H'$ are not connected to each other, while they only connect to~$U$ through their center vertex~$c_S$, an $H$-subgraph in~$G$ can only appear as described in the claim.
\end{claimproof}

\begin{claim}
If~$t \cdot H \subseteq G$, then~$(\S,U)$ has an exact cover.
\end{claim}
\begin{claimproof}
Assume that~$t \cdot H \subseteq G$. By Claim~\ref{claim:fountain:hardness:goodmodels}, every $H$-subgraph in~$G$ consists of a subgraph~$H'_S$ for~$S \in \S$ together with the~$r$ vertices in~$N_G(c_S) \cap U$. If there are~$t$ vertex-disjoint copies of~$H$ in~$G$, this implies that there are~$t$ vertices~$c_{S_1}, \ldots, c_{S_t}$ such that the sets~$N_G(c_{S_i})$ and~$N_G(c_{S_j})$ are disjoint for~$i \neq j$. As these sets contain~$t \cdot r = n$ vertices in total, the sets~$N_G(c_{S_1}) = S_1, \ldots, N_G(c_{S_t}) = S_t$ cover all of~$U$. The corresponding sets~$S_1, \ldots, S_t$ therefore form an exact cover of~$U$.
\end{claimproof}

As the claims establish the equivalence of the input and output instance, this concludes the proof of Lemma~\ref{lemma:fountainpacking:wkhard}.
\end{proof}

\begin{lemma} \label{lemma:operahousepacking:wkhard}
\textsc{\Foperahouse{s}-Packing} is \textup{WK[1]}-hard for any odd integer~$s \geq 1$.
\end{lemma}
\begin{proof}
Let~$s\geq 1$ be an odd integer. We transform an instance~$(r, \S, U, n)$ of \nRegularExactSetCover in polynomial time into an equivalent instance~$(G, H := \Goperahouse{s}{r},t := n/r)$ of \textsc{\Foperahouse{s}-Packing} with~$k := t \cdot |V(H)| \in \Oh(n)$ (see Figure~\ref{fig:basic}). The host graph~$G$ is defined as follows.

\begin{itemize}
	\item Graph~$G$ contains the vertex set~$U$ as an independent set.
	\item For every size-$r$ set~$S \in \S$ we add a new length-$s$ path~$P_S$ to~$G$. Let~$x_S$ and~$y_S$ be the endpoints of~$P_S$. We make the vertices~$x_S$ and~$y_S$ adjacent to all vertices representing members of~$S$.
\end{itemize}

As the construction can be performed in polynomial time and produces an instance of the target problem whose parameter is suitably bounded, it remains to prove that the input and output instances are equivalent.

\begin{claim}
If~$(\S,U)$ has an exact cover, then~$t \cdot H \subseteq G$.
\end{claim}
\begin{claimproof}
Assume that~$(\S,U)$ has an exact cover~$S_1, \ldots, S_t \in \S$. For each~$S_i$ with~$i \in [t]$ the path~$P_{S_i}$ together with the~$r$ vertices~$S_i$ form a~$H$-subgraph in~$G$. Since the sets~$S_i$ are pairwise disjoint, there are~$t$ vertex-disjoint $H$-subgraphs in~$G$. Hence~$t \cdot H$ is a subgraph of~$G$.
\end{claimproof}

\begin{claim} \label{claim:operahouse:hardness:goodmodels}
Any~$H$ subgraph in~$G$ consists of a path~$P_S$ for some~$S \in \S$ together with the~$r$ vertices in~$N_G(x_S) \cap N_G(y_S) = S$.
\end{claim}
\begin{claimproof}
Recall that~$s \geq 1$ is odd. The only simple odd cycles in~$G$ of length~$s+2$ are those consisting of a path~$P_S$ for~$S \in \S$ together with a vertex in~$N_G(x_S) \cap N_G(y_S) \subseteq U$. To see this, observe that removing the first edge of every path~$P_S$ ($S \in \S$) results in a bipartite graph: after removing these edges and iterately removing degree-one vertices, we are left with the bipartite graph where~$U$ forms one partite class and~$\{x_S, y_S \mid S \in \S\}$ forms the other partite class. Hence every odd cycle contains an edge on a path~$P_S$, and use of such an edge forces the entire path~$P_S$ to be used. As it has length~$s$, to complete a length-$(s+2)$ path requires a vertex in the common neighborhood of the endpoints of the path; as~$s$ is odd, the only such vertices are those in~$N_G(x_S) \cap N_G(y_S) = U$.

Now consider a $H$-subgraph in~$G$. As~$H = \Goperahouse{s}{r}$ contains a length~$(s+2)$ cycle, the above argument shows that the model of this cycle is formed by a path~$P_S$ for~$S \in \S$ together with a common $U$-neighbor of~$\{x_S, y_S\}$, resulting in a length-$(s+2)$ cycle~$C'$ in~$G$. Since the vertices~$x_S$ and~$y_S$ form the unique pair of vertices on~$C'$ that have degree at least~$r + 1$ in~$G$ and are connected by a path of length~$s$ through~$C'$, the two high-degree vertices of~$\Goperahouse{s}{r}$ must be realized by~$x_S$ and~$y_S$. Since~$N_G(x_S) \cap N_G(y_S) = S \subseteq U$, the subgraph must use the~$r$ vertices in~$S$ to realize the common neighbors of the endpoint of the path.
\end{claimproof}

\begin{claim}
If~$t \cdot H \subseteq G$, then~$(\S,U)$ has an exact cover.
\end{claim}
\begin{claimproof}
Assume that~$t \cdot H \subseteq G$. By Claim~\ref{claim:operahouse:hardness:goodmodels}, every $H$-subgraph in~$G$ consists of a path~$P_S$ for~$S \in \S$ together with the vertices representing~$S$. If there are~$t$ vertex-disjoint copies of~$H$ in~$G$, then there are~$t$ disjoint sets~$S_1, \ldots, S_r \in \S$. As these sets contain~$t \cdot r = n$ vertices in total, they cover all of~$U$ and form an exact cover.
\end{claimproof}

This concludes the proof of Lemma~\ref{lemma:operahousepacking:wkhard}.
\end{proof}


\begin{lemma} \label{lemma:substarpacking:wkhard}
\textsc{\Fsubdivstar-Packing} is \textup{WK[1]}-hard.
\end{lemma}
\begin{proof}
We transform an instance~$(r, \S, U, n)$ of \nRegularExactSetCover in polynomial time into an equivalent instance~$(G, H := \Gsubdivstar{r},t := n/r)$ of \textsc{\Fsubdivstar-Packing} with~$k := t \cdot |V(H)| \in \Oh(n)$. By the definition of \nRegularExactSetCover we have~$r \geq 3$. Let~$c \in V(H)$ be the center of the star whose subdivision yields~$H$, implying that~$\deg_H(c) \geq 3$. The host graph~$G$ is defined as follows.

\begin{itemize}
	\item Graph~$G$ contains the vertex set~$U$ as an independent set. For every vertex~$u \in U$ we add a pendant degree-one neighbor~$u'$ adjacent to~$u$.
	\item For every size-$r$ set~$S \in \S$ we add a new vertex~$x_S$ to~$G$ that is adjacent to the~$r$ vertices in~$U$ representing~$S$.
\end{itemize}

As the construction can be performed in polynomial time and produces an instance of the target problem whose parameter is suitably bounded, it remains to prove that the input and output instances are equivalent.

\begin{claim}
If~$(\S,U)$ has an exact cover, then~$t \cdot H \subseteq G$.
\end{claim}
\begin{claimproof}
Assume that~$(\S,U)$ has an exact cover~$S_1, \ldots, S_t \in \S$. For each~$S_i$ with~$i \in [t]$ the vertices~$\{u, u' \mid u \in S\}$ form a subdivided star with~$r$ leaves and~$x_{S_i}$ as its center. Since the sets~$S_i$ are pairwise disjoint, there are~$t$ vertex-disjoint $H$-subgraphs in~$G$.
\end{claimproof}

\begin{claim} \label{claim:subdivstar:hardness:goodmodels}
If~$\phi$ is a full $H$-subgraph model in~$G$, then the following holds.
\begin{itemize}
	\item If~$\phi(c) \in \{x_S \mid S \in \S\}$ then~$|\phi(V(H)) \cap U| \geq r$.
	\item If~$\phi(c) \not \in \{x_S \mid S \in \S\}$ then~$|\phi(V(H)) \cap U| \geq r + 1$.
\end{itemize}
\end{claim}
\begin{claimproof}
Consider a full $H$-subgraph model~$\phi$ in~$G$. Since~$\deg_H(c) \geq 3$ while~$\deg_G(u') = 1$ for all~$u \in U$, the image of~$c$ can only be a vertex in~$U$ or a vertex~$x_S$ for~$S \in \S$. If~$\phi(c) = x_S$ for some~$S \in \S$, then~$N_G(x_S) = S \subseteq \phi(V(H))$, as~$c$ has~$r$ neighbors in~$H$ that must be realized by the unique set of~$r$ neighbors of~$x_S$ in~$G$. This proves the first part of the claim.

To prove the second part, by the observation above it suffices to consider the case that~$\phi(c) = u \in U$. Observe that the vertex~$u'$ has degree one in~$G$ and therefore cannot form the image of a subdivider vertex of the star. Hence all~$r$ subdivider vertices are realized by other neighbors of~$u$, which all belong to~$\{x_S \mid S \in \S\}$. Consequently, the images of the degree-one endpoints of the subdivided star are neighbors of~$\{x_S \mid S \in \S\}$ distinct from~$u$. As all neighbors of~$\{x_S \mid S \in \S\}$ belong to~$U$, this implies that~$|\phi(V(H)) \cap U| \geq r+1$.
\end{claimproof}

\begin{claim}
If~$t \cdot H \subseteq G$, then~$(\S,U)$ has an exact cover.
\end{claim}
\begin{claimproof}
Assume that~$t \cdot H \subseteq G$. By Claim~\ref{claim:subdivstar:hardness:goodmodels}, every $H$-subgraph model in~$G$ uses at least~$r$ vertices of~$U$. As there are only~$n$ vertices in~$U$ in total, to realize~$t = n/r$ vertex-disjoint $H$-subgraphs, each of the~$t$ subgraphs has to use exactly~$r$ vertices of~$U$. Such subgraphs therefore map the center of the subdivided star to a vertex~$x_S$ for~$S \in \S$. Hence there are~$t$ distinct vertices~$x_{S_1}, \ldots, x_{S_t}$ forming the centers of $H$-subgraph models in~$G$. For each~$i \in [t]$, the vertices~$N_G(x_{S_i}) = S_i$ must be used in the model containing~$x_{S_i}$, since they are the unique~$r$ neighbors of~$x_{S_i}$ in~$G$. Thus~$S_1, \ldots, S_t$ are pairwise disjoint sets in~$\S$ covering~$n$ elements in total, which means they form an exact cover of~$U$.
\end{claimproof}

This concludes the proof of Lemma~\ref{lemma:substarpacking:wkhard}.
\end{proof}


\begin{lemma} \label{lemma:broompacking:wkhard}
\textsc{\Fdoublebroom{s}-Packing} is \textup{WK[1]}-hard for any odd integer~$s \geq 1$.
\end{lemma}
\begin{proof}
We transform an instance~$(r \geq 3, \S, U, n)$ of \nRegularExactSetCover in polynomial time into an equivalent instance~$(G, H := \Gdoublebroom{s}{r},t := n/r)$ of \textsc{\Fdoublebroom{s}-Packing} with~$k := t \cdot |V(H)| \in \Oh(n)$. The host graph~$G$ is defined as follows.

\begin{itemize}
	\item Graph~$G$ contains the vertex set~$U$ as an independent set.
	\item For every size-$r$ set~$S \in \S$ we add a new length-$s$ path~$P_S$ to~$G$. Let~$x_S$ and~$y_S$ be the endpoints of~$P_S$. We make~$x_S$ adjacent to the~$r$ vertices in~$U$ representing~$S$. We add~$r$ new vertices and make them adjacent to~$y_S$.
\end{itemize}

As the construction can be performed in polynomial time and produces an instance of the target problem whose parameter is suitably bounded, it remains to prove that the input and output instances are equivalent.

\begin{claim}
If~$(\S,U)$ has an exact cover, then~$t \cdot H \subseteq G$.
\end{claim}
\begin{claimproof}
Assume that~$(\S,U)$ has an exact cover~$S_1, \ldots, S_t \in \S$. For each~$S_i$ with~$i \in [t]$ the vertices representing~$S$ combine with the path~$P_S$ and the~$r$ degree-one neighbors of~$y_S$ to form a \Gdoublebroom{s}{r} subgraph. Since the sets~$S_i$ are pairwise disjoint, there are~$t$ vertex-disjoint $H$-subgraphs in~$G$.
\end{claimproof}

Let~$x$ and~$y$ be the unique two vertices of degree~$r + 1$ in \Gdoublebroom{s}{r} (recall that~$s \geq 1$ and~$r \geq 3$).

\begin{claim} \label{claim:doublebroom:hardness:goodmodels}
If~$\phi$ is a full $H$-subgraph model in~$G$, then~$\phi(V(H)) \cap U \supseteq S$ for some~$S \in \S$.
\end{claim}
\begin{claimproof}
Consider a full $H$-subgraph model~$\phi$ in~$G$. If~$\phi(\{x,y\}) = \{x_S, y_S\}$ for~$S \in \S$, then as~$P_S$ is the unique length-$s$ path in~$G$ between~$x_S$ and~$y_S$, this path is used as the image of the length-$s$ path in~$H$. Hence the unique~$r$ vertices~$N_G(x_S) \setminus V(P_S) = S$ must be used as images of the degree-one neighbors of~$x$ in~$H$. Hence~$\phi(V(H)) \supseteq S$.

If~$\phi(\{x,y\}) \cap \{y_S \mid S \in \S\} \neq \emptyset$, then we must have~$\phi(\{x,y\}) = \{x_S, y_S\}$ for some~$S \in \S$, since the only vertex in~$G$ at distance exactly~$s$ from~$y_S$ that has degree more than one is~$x_S$. Hence if $\phi(\{x,y\}) \cap \{y_S \mid S \in \S\} \neq \emptyset$ then the previous argument shows that the claim holds.

It remains to consider the case that the images of~$x$ and~$y$ do not belong to~$\{y_S \mid S \in \S\}$. As~$\deg_G(\phi(x)) \geq \deg_H(x) \geq 4$ (recall~$r \geq 3$) and~$\deg_G(\phi(y)) \geq 4$, the only remaining options for~$\phi(x)$ and~$\phi(y)$ are the vertices in~$U \cup \{x_S \mid S \in \S\}$. As the constructed graph~$G$ is bipartite and all vertices of~$U$ are contained in the same partite class, while the path connecting~$x$ and~$y$ in~$H$ has odd length~$s$, it follows that the model of one of~$\{x,y\}$ belongs to~$U$ while the other does not. By the previous argument, this implies that the model of the other vertex belongs to~$\{x_S \mid S \in \S\}$. Assume without loss of generality that~$\phi(y) \in U$ and~$\phi(x) = x_S$ for some~$S \in \S$. Since~$\deg_H(x) = \deg_G(x_S) = r + 1$, all vertices of~$N_G(x_S)$ are used in the model to realize the images of~$N_H(x)$. Since~$S \subseteq N_G(x_S)$, the claim follows.
\end{claimproof}

\begin{claim}
If~$t \cdot H \subseteq G$, then~$(\S,U)$ has an exact cover.
\end{claim}
\begin{claimproof}
Assume that~$t \cdot H \subseteq G$. By Claim~\ref{claim:doublebroom:hardness:goodmodels}, the image of every $H$-subgraph model in~$G$ is a superset of some~$S \in \S$. If there is a subgraph model that uses a strict superset, then this leaves less than~$n - r$ vertices to realize the remaining~$t-1$ $H$-subgraph models, which is not possible since each requires a superset of a size-$r$ set. Hence the intersection of each $H$-subgraph in the packing with~$U$ is a set~$S \in \S$. Since the packing consists of vertex-disjoint subgraphs, there are~$t$ pairwise disjoint sets in~$\S$. As they each have size~$r$, these~$n/r$ sets cover the entire universe~$U$. Hence~$(\S,U)$ has an exact cover.
\end{claimproof}

This concludes the proof of Lemma~\ref{lemma:broompacking:wkhard}.
\end{proof}

As we have given \textup{WK[1]}-hardness proofs for all graph families listed in Theorem~\ref{theorem:packing:lowerbounds}, the theorem follows.

\subsection{Combinatorial characterizations}
\label{sec:packing-ramsey}
In this section, our goal is to prove Theorem~\ref{theorem:ramsey:1} characterizing hereditary classes that are not small/thin.
If a class $\F$ is not small/thin, then $\F$ contains graphs either with large nonbipartite components or with large non-thin bipartite components. We investigate both cases separately and show that one the basic classes listed in Theorem~\ref{theorem:ramsey:1} is contained in $\F$.

Let us treat first the case of large nonbipartite graphs.

\begin{lemma} \label{lemma:ramsey:nonbipartite0} For every $Q\ge 1$,
  there is an $a=a(Q)\ge 1$ such that if $H$ is a connected nonbipartite
  graph with at least $a$ vertices, then $H$ contains at least one of
  the following graphs as induced subgraph:
\begin{enumerate}
\item $\Gpath{Q}$,
\item $\Gclique{Q}$,
\item $\Gfountain{s}{Q}$ for some odd integer $3\le s\le Q+2$,
\item $\Glongfountain{s}{t}{Q}$ some odd integer $3\le s\le Q+2$ and integer $1\le t\le Q$, or
\item $\Goperahouse{s}{Q}$ for some odd integer $1\le s\le Q$.
\end{enumerate}
\end{lemma}
\begin{proof}

We define the following constants:
\begin{align*}
d_1=&R(3Q,2)+(Q+1),\\
d_2=&R(3Q,2)+1+(Q+2)(d_1+1),\\
a=&\sum_{i=0}^{Q-1}d_2^i+1,
\end{align*}
where $R(3Q,2)$ is the Ramsey number guaranteeing a clique or
independent set of size at least $3Q$ (Theorem~\ref{theorem:ramsey}). Observe that the constant $a$ depends only on $Q$.

Let $C$ be an odd cycle of minimum length in $H$.  Observe that $C$
has no chord, otherwise the chord would split the odd cycle $C$ into a
strictly shorter odd cycle and a strictly shorter even cycle, a
contradiction. Thus $C$ is an induced cycle and hence has length at
most $Q+2$, otherwise $\Gpath{Q}$ would be an induced subgraph of $H$.

Assume that $H$ has at least $a$ vertices, but does not contain any of
the graphs listed in the lemma as induced subgraph. Select an
arbitrary vertex of $H$. If $H$ does not contain $\Gpath{Q}$ as an
induced subgraph, then every vertex of $H$ is at distance at most
$Q-1$ from the selected vertex. If $H$ has maximum degree $d_2$, then
this would imply that $H$ has at most $\sum_{i=0}^{Q-1}d_2^i<a$
vertices, a contradiction. Thus we may assume that $H$ has a vertex
$v$ of degree at least $d_2$.  We consider two cases.

\textbf{Case 1: Every vertex on the cycle $C$ has degree at most
  $d_1$.} As $d_1<d_2$, vertex $v$ is not on $C$. Let $p_0=v$, $p_1$,
$\dots$, $p_{\ell}=w$ be a shortest path $P$ from $v$ to a vertex $w$
of the cycle $C$.  Vertex $v$ has at most $|C|(d_1+1)$ neighbors in the
closed neighborhood of $C$ and (as $P$ is a shortest path), only one
neighbor on $P$, namely $p_1$. Thus we can select from the at least $d_2$ neighbors of $v$ a set $X_1$ of
at least $d_2-1-|C|(d_1+1)\ge d_2-1-(Q+2)(d_1+1)=R(3Q,2)$ vertices that are not
on $P$, not on $C$, and not in the neighborhood of $C$. As $H$ has no
clique of size $3Q$, Ramsey's Theorem (Theorem~\ref{theorem:ramsey}) implies that there is an
independent set $X_2\subseteq X_1$ of size at least $3Q$. As $P$ is a
shortest path, a vertex of $X_2$ cannot be adjacent to $p_i$ for
$i>2$. Let us partition $X_2$ into three classes: let $X_{2,i}$
contain a vertex of $X_2$ if it is adjacent to $p_i$, but not to
$p_{i'}$ for any $i'>i$ (as $p_0=v$, every vertex of $X_2$ is in exactly one of 
$X_{2,0}$, $X_{2,1}$, or $X_{2,2}$). As $|X_2|\ge 3Q$, there is an
$0\le i^* \le 2$ such that $|X_{2,i^*}|\ge Q$.  Note that $i^*<\ell$:
since $p_\ell$ is on $C$, no vertex of $X_2$ is adjacent to it. As $P$
is a shortest path, only $p_{\ell-1}$ can have neighbors on the cycle
$C$. We consider two subcases (see Figure~\ref{fig:oddcycle}).
\begin{figure}[t]
\begin{center}
{\small \svg{\linewidth}{oddcycle}}
\caption{Proof of Lemma~\ref{lemma:ramsey:nonbipartite0}.}\label{fig:oddcycle}
\end{center}
\end{figure}

\textbf{Case 1.a: $p_\ell$ is the only neighbor of $p_{\ell-1}$ on
  $C$.} Then the set $X_{2,i^*}$, the subpath $P'$ of $P$ from
$p_{i^*}$ to $p_\ell=w$, and the cycle $C$ form a
$\Glongfountain{|C|}{|P'|}{|X_{2,i^*}|}$, contradicting our
assumptions (note that $|C|\le Q+2$, $|P'|\le Q$, and $|X_{2,i^*}|\ge
Q$).

\textbf{Case 1.b: $p_\ell$ has at least two neighbors on $C$.}  Then
$C$ has a subpath $R$ of odd length such that $p_{\ell-1}$ is adjacent
to the endpoints of $R$, but not to the internal vertices of $R$ (this
is because the neighbors of $p_{\ell-1}$ split the cycle $C$ into
subpaths, and at least one such subpath has odd length). We may
observe that, by the minimality of the cycle $C$, the length of $R$
is exactly $|C|-2$, otherwise $R$ and $p_{\ell-1}$ would form an odd
cycle strictly shorter than $C$. Now $X_{2,i^*}$, the subpath $P'$
of $P$ from $p_{i^*}$ to $p_{\ell-1}$, and path $R$ form a
$\Glongfountain{|R|+2}{|P'|}{|X_{2,i^*}|}$. As $|R|+2= |C|\le Q+2$,
$|P'|\le Q$, and $|X_{2,i^*}|\ge Q$, this contradicts assumptions on
$Q$.

\textbf{Case 2: $C$ has a vertex $w$ of degree at least $d_1$.}
Vertex $w$ has a set $X_1$ of at least $d_1-(|C|-1)\ge
d_1-(Q+2-1)=R(3Q,2)$ neighbors that are not on $C$. As $H$ has no
clique of size $3Q$, Ramsey's Theorem (Theorem~\ref{theorem:ramsey}) implies that there is an
independent set $X_2\subseteq X_1$ of size at least $3Q$.  We need to
treat the case when $C$ has length 3 separately, thus we consider two
subcases (see Figure~\ref{fig:oddcycle}).

  \textbf{Case 2.a: $|C|=3$.}  Let $w$, $z_1$, and $z_2$ be the
  vertices of $C$. If $X_2$ has a subset $X_3$ of $Q$ vertices adjacent to $z_1$, then
  the set $X_3$ and the path $wz_1$ form an
  $\Goperahouse{1}{Q}$. We can get an $\Goperahouse{1}{|Q|}$ in a
  similar way if $X_2$ has a subset $X_3$ of $Q$ vertices adjacent to $z_2$. Otherwise, 
  there is a set $X_3\subseteq X_2$ of at least $Q$ vertices that are
  adjacent to neither $z_1$ nor $z_2$, and hence $X_3$ and the cycle
  $C$ form a $\Gfountain{3}{Q}$.

  \textbf{Case 2.b: $|C|\ge 5$.}  Let $q_1$, $q_2$, $w$, $q_3$, $q_4$
  be consecutive vertices appearing on $C$ in this order; as $|C|\ge
  5$, these 5 vertices are distinct. It is easy to see that no vertex
  $u$ of $X_2$ can be adjacent to any vertex $q'$ of $C\setminus
  \{q_1,q_2,w,q_3,q_4\}$: then the path $wuq'$ and the subpath of $C$
  between $w$ and $q'$ that has odd length would form an odd cycle
  strictly shorter than $C$, a contradiction. If $X_2$ has a vertex
  adjacent to either $q_2$ or $q_3$, then there is a triangle,
  contradicting the minimality of $C$.  If $X_2$ has a vertex $u$
  adjacent to both $q_1$ and $q_4$, then the path $q_1uq_4$ and the
  subpath of $C$ between $q_1$ and $q_4$ avoiding $w$ would form an
  odd cycle strictly shorter than $C$. Thus a vertex in $X_2$ can have at
  most one neighbor on $C$ besides $w$, and this neighbor can only be
  $q_1$ or $q_4$.  If there is a set $X_3\subseteq X_2$ of $Q$
  vertices that are adjacent to $q_1$, then the subpath of $C$ between
  $q_1$ and $w$ avoiding $q_2$ form an $\Goperahouse{|C|-2}{Q}$ (note
  that $|C|-2\le Q$). We get an $\Goperahouse{|C|-2}{Q}$ in a similar
  way if there is a set $X_3\subseteq X_2$ of $Q$ vertices that are
  adjacent to $q_4$. Otherwise, there is a set $X_3\subseteq X_2$ of
  $Q$ vertices that are adjacent to only $w$ on $C$; then $X_3$ and
  $C$ form a $\Gfountain{|C|}{Q}$.
\end{proof}

Lemma~\ref{lemma:ramsey:nonbipartite0} is a statement about a specific
graph. We turn it into a statement on graph classes.
\begin{lemma} \label{lemma:ramsey:nonbipartite} Let~$\F$ be a
  hereditary graph family. Then at least one of the following is true.
\begin{enumerate}
\item There is an integer $a\ge 1$ such that every connected nonbipartite graph in $\F$ has at most~$a$ vertices.
\item $\F$ is a superset of \Fpath.
\item $\F$ is a superset of \Fclique.
\item $\F$ is a superset of \Ffountain{s} for some odd integer $s\ge 3$.
\item $\F$ is a superset of \Flongfountain{s,t} for some odd integer $s\ge 3$ and integer $t\ge 1$.
\item $\F$ is a superset of \Foperahouse{s} for some odd integer $s\ge 1$.
\end{enumerate}
\end{lemma}
\begin{proof}
Assuming that $\F$ is not a superset of $\Fpath$, there is an integer $\ell\ge 1$ such that $\Gpath{\ell}\not\in \F$. Then, 
assuming that $\F$ is not a superset of any of the other families described in the lemma, there is a $Q\ge \ell \ge 1$ such that $\F$ does not contain any of the following graphs:
\begin{itemize}
 \item $\Gclique{Q}$,
 \item $\Gfountain{s}{Q}$ for any odd integer $3\le s < \ell+2$,
 \item $\Glongfountain{s}{t}{Q}$ for any odd integer $3 \le s < \ell+2$ and integer $1\le t < \ell$,
 \item $\Goperahouse{s}{Q}$ for any odd integer $1\le s < \ell$.
 \end{itemize}
 The reason why such a $Q$ can be defined is that we have to consider
 a finite number of infinite sequences of graphs (such as
 $\Gfountain{s}{i}$ for a fixed $3\le s \le \ell+2$ and for
 $i=1,2,\dots$), we know by assumption that these sequences are not
 contained in $\F$, hence for each sequence there is an $i$ such that
 the $i$-th element of the sequence is not contained in $\F$. In fact,
 as the $i$-th element of the sequence is an induced subgraph of the
 $i'$-th element for every $i'>i$, we know that no element of the
 sequence after the $i$-th element can appear in $\F$ either. Then we
 can define $Q$ to be the maximum of all these $i$'s corresponding to
 the finitely many forbidden sequences. As $\F$ is hereditary, we also
 know that none of these graphs appear as induced subgraphs in the
 members of $\F$.

 Observe that an induced $\Gpath{\ell}$ is contained in every
 $\Gfountain{s}{Q}$ with $s\ge \ell+2$, in every
 $\Glongfountain{s}{t}{Q}$ with $s\ge \ell+2$ or $t\ge \ell$, and in
 every $\Goperahouse{s}{Q}$ with $s\ge \ell$. Therefore, none of the
 graphs listed in Lemma~\ref{lemma:ramsey:nonbipartite0} can appear in
 $\F$, hence every connected nonbipartite graph in $\F$ has size at
 most $a(Q)$.
\end{proof}

Next we consider large bipartite graphs that are not thin.

\begin{lemma} \label{lemma:ramsey:balancedbipartite0} For every $Q\ge
  1$, there is an integer $b=b(Q)\ge 1$ such that every connected
  bipartite graph is either $b$-thin or contains one the following
  graphs as induced subgraphs:
\begin{enumerate}
\item $\Gpath{Q}$,
\item $\Gclique{Q}$,
\item $\Gbiclique{Q}$,
\item $\Gsubdivstar{Q}$, or
\item $\Gdoublebroom{s}{Q}$ some odd integer $1\le s\le Q$.
\end{enumerate}
\end{lemma}
\begin{proof}
We define the following constants:
\begin{align*}
h&=P(Q,Q), && \textup{($P$ is from Corollary~\ref{theorem:path:ramsey2})}\\
\kappa&=\sum_{i=0}^{h}(Q-1)^i,\\
d&=1+R(\bipRamsey(2Q,2),2),&&\textup{($R$ is from Ramsey's Theorem; $\bipRamsey$ is from Theorem~\ref{theorem:bipartite:ramsey})}\\
b&=\kappa d.
\end{align*}
Observe that the constant $b$ depends only on $Q$.
Assume that $H\in \F$ is a connected $b$-thin bipartite graph not
containing any of the listed graphs as an induced subgraph. First we
bound the vertex cover number of $H$.

\begin{claim}\label{claim:vcbound}
  $H$ has a vertex cover of size at most $\kappa$.
\end{claim}
\begin{proof}
Construct a DFS tree of $H$ starting at an
arbitrary root. Let $L\subseteq V(H)$ be the set of leaves in the DFS
tree and let $Z=V(H)\setminus L$.  Observe that $Z$ is a vertex cover
of $H$: there are no edges between the leaves.

We claim that if a vertex $v\in Z$ has $Q$ children $x_1,\dots,x_Q\in
Z$ in the DFS tree, then $H$ contains $\Gsubdivstar{Q}$ as an induced
subgraph, a contradiction. To show this, let us select an arbitrary child $y_i$ of each 
$x_i$ (note that $x_i\in Z$ is not a leaf). Then it is clear that the
set $\{v,x_1,\dots, x_s,y_1,\dots,y_s\}$ forms a $\Gsubdivstar{Q}$
subgraph. To see that it is an induced subgraph, observe that there is
no edge between $\{x_i,y_i\}$ and $\{x_j,y_j\}$ for any $i\neq j$ by
the properties of the DFS tree and there is no edge $\{v,y_i\}$, as 
this would create a triangle $\{v,x_i,y_i\}$ and the graph is
bipartite by assumption. Therefore, every $v\in Z$ has at most $Q-1$
children in $Z$, otherwise the graph would contain $\Gsubdivstar{Q}$ as an
induced subgraph.

The height of the DFS tree is at most $h=P(Q,Q)$, otherwise
Corollary~\ref{theorem:path:ramsey2} would imply that $H$ contains
$\Gpath{Q}$, $\Gclique{Q}$, $\Gbiclique{Q}$ as an induced
subgraph. This means that $|Z|\le \sum_{i=0}^{h}(Q-1)^i=\kappa$. As we
have observed, $Z$ is a vertex cover of $H$, thus $H$ has vertex cover
number at most $\kappa$.  \cqed
\end{proof}

Assume first that one partite class of $H$ has maximum degree $d$.  Let $A$ and $B$ be the two partite classes and assume that the
vertices in $A$ have degree at most $d$. 
 As $H$ is not $b$-thin, both partite classes of $H$ have size more than
$b$. Pick one edge incident to
each vertex of $B$, we get $|B|\ge b+1=\kappa d+1$ distinct edges. Now
each vertex can cover at most $d$ of these edges: a vertex in $A$ can
cover at most $d$ (as it has degree at most $d$) and a vertex in $B$
covers exactly one such edge. This contradicts Claim~\ref{claim:vcbound}.

Therefore, we can assume that $H$ has two vertices $x$ and $y$ having
degree at least $d$ in each partite class of $H$.  Let $p_0=x$,
$p_1$, $\dots$, $p_{\ell-1}$, $p_\ell=y$ be a shortest path $P$
between $x$ and $y$ (see Figure~\ref{fig:2broom}).
\begin{figure}[t]
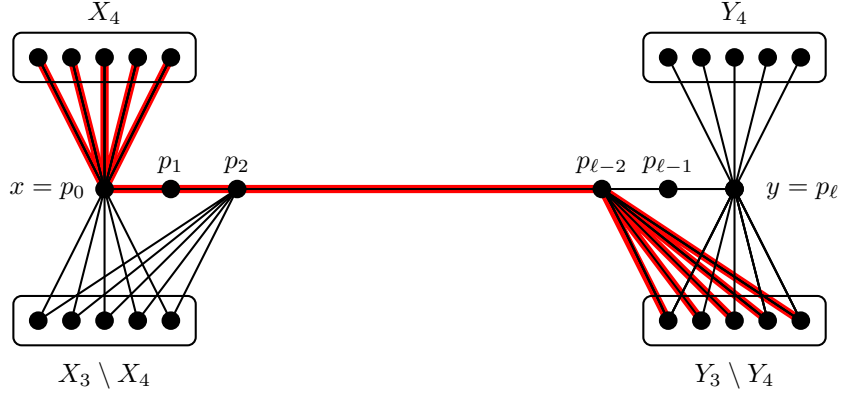

\begin{center}
{\small \svg{0.7\linewidth}{2broom}}
\end{center}
\caption{Proof of Lemma~\ref{lemma:ramsey:balancedbipartite0}: the case when $|X_4|\ge Q$ and $|X_3\setminus X_4|\ge Q$.}\label{fig:2broom}
\end{figure}
As $x$ and $y$ are in different partite classes, we have that
$\ell\ge 1$ is odd and $|P|<Q$, otherwise it would form an induced
$\Gpath{Q}$. Let $X_1$ be a set of $d-1$ neighbors of $x$ different
from $p_1$ and let $Y_1$ be a set of $d-1$ neighbors of $y$ different
from $p_{\ell-1}$. Note that $X_1$ and $Y_1$ are in different partite
classes, hence disjoint. Moreover, no vertex of $X_1$ or $Y_1$ is on
the path $P$, as $P$ is a shortest path between $x$ and $y$. The
definition of $d$ and Ramsey's Theorem (Theorem~\ref{theorem:ramsey}) implies that $X_1$ has a clique
or independent set of size $\bipRamsey(2Q,2)$. As $H$ does not contain
$\Gclique{Q}$ and certainly $Q$ is less than $\bipRamsey(2Q,2)$, the only
possibility is that there is an independent set $X_2\subseteq X_1$ of
size $\bipRamsey(2Q,2)$. Similarly, there is an independent set $Y_2\subseteq
Y_1$ of size $\bipRamsey(2Q,2)$. As $H$ contains no $\Gbiclique{2Q}$ as
induced subgraph, Theorem~\ref{theorem:bipartite:ramsey} implies that
there are sets $X_3\subseteq X_2$ and $Y_3\subseteq Y_2$ of size $2Q$
such that there is no edge between $X_3$ and $Y_3$. It is clear that
the sets $X_3$, $Y_3$, and the path $P$ form a $\Gdoublebroom{|P|}{Q}$
subgraph, but it is not necessarily an induced subgraph.  As $P$ is a
shortest $x-y$ path, there is no edge between $X_3$ and $p_i$ for
$i\ge 3$, and there is no edge between $X_3$ and $p_1$, as it would
form a triangle with $x=p_0$. Therefore, the only possible extra edges
are between $X_3$ and $p_2$ and between $Y_3$ and $p_{\ell-2}$ (and
this is only possible if $\ell\ge 3$). Let $X_4\subseteq X_3$ be those
vertices of $X_3$ that are not adjacent to $p_2$ and let $Y_4\subseteq
Y_3$ be those vertices of $Y_3$ that are not adjacent to
$p_{\ell-2}$. If $|X_4|\ge Q$, then let $X'$ be a $Q$ element subset
of $X_4$ and let $x'=p_0$; if $|X_4|< Q$, then let $X'$ be a $Q$
element subset of $X_3\setminus X_4$ and let $x'=p_2$.  Similarly, if
$|Y_4|\ge Q$, then let $Y'$ be a $Q$ element subset of $Y_4$ and let
$y'=p_\ell$; if $|Y_4|< Q$, then let $Y'$ be a $Q$ element subset of
$Y_3\setminus Y_4$ and let $y'=p_{\ell-2}$. Let $P'$ be the subpath of
$P$ between $x'$ and $y'$. Note that this path has odd length (if
$\ell=3$, it is possible that $x'=p_2$ and $y'=p_1$, but the length is
still an odd number, namely one). Then $X'$, $Y'$, and $P'$ form an
induced $\Gdoublebroom{|P'|}{Q}$ subgraph, a contradiction.
\end{proof}

As in Lemma~\ref{lemma:ramsey:nonbipartite}, we turn Lemma~\ref{lemma:ramsey:balancedbipartite0} into a statement on graph classes.
\begin{lemma} \label{lemma:ramsey:balancedbipartite}
Let~$\F$ be a hereditary family of bipartite graphs. Then at least one of the following is true.
\begin{enumerate}
\item There is an integer $b\ge 1$ such that every connected graph in $\F$ is $b$-thin.
\item $\F$ is a superset of \Fpath.
\item $\F$ is a superset of \Fclique.
\item $\F$ is a superset of \Fbiclique.
\item $\F$ is a superset of \Fsubdivstar.
\item $\F$ is a superset of \Fdoublebroom{s} for some odd integer $s\ge 1$.
\end{enumerate}
\end{lemma}
\begin{proof}
Assuming that $\F$ is not a superset of $\Fpath$, there is an integer $\ell\ge 1$ such that $\Gpath{\ell}\not\in \F$. Then, 
as in the proof of Lemma~\ref{lemma:ramsey:nonbipartite}, we can assume that there is a $Q\ge 1$ such that $\F$ does not contain any of the following graphs:
\begin{itemize}
\item $\Gpath{Q}$
 \item $\Gclique{Q}$,
 \item $\Gbiclique{Q}$,
 \item $\Gdoublebroom{s}{Q}$ for any odd integer $1\le s < \ell$,
 \item $\Gsubdivstar{Q}$.
 \end{itemize}
 For every $s\ge \ell$, $\Gdoublebroom{s}{Q}$ contains $\Gpath{\ell}$
 as induced subgraph, hence none of the graphs listed in
 Lemma~\ref{lemma:ramsey:balancedbipartite0} is contained in $\F$. It
 follows that every connected bipartite graph in $\F$ is $b(Q)$-thin.
\end{proof}

Combining Lemmas~\ref{lemma:ramsey:nonbipartite} and
\ref{lemma:ramsey:balancedbipartite}, the proof of
Theorem~\ref{theorem:ramsey:1} follows.  \restateramseysmallthin*
\begin{proof}
  Let us apply first Lemma~\ref{lemma:ramsey:nonbipartite} on $\F$. If
  $\F$ is a superset of any of the families listed in items 2--6 of
  Lemma~\ref{lemma:ramsey:nonbipartite}, then we are done. Assume
  therefore that there is an integer $a\ge 1$ such that every
  connected nonbipartite graph in $\F$ has size at most $a$.  Let
  $\F'$ contain every bipartite graph in $\F$ and let us apply
  Lemma~\ref{lemma:ramsey:balancedbipartite} on $\F'$. Again, if $\F$
  is a superset of any of the families listed in items 2--6 of
  Lemma~\ref{lemma:ramsey:balancedbipartite}, then we are
  done. Therefore, we may assume that there is an integer $b\ge 1$
  such that every connected graph in $\F'$ is $b$-thin. Observe now
  that if $C$ is a bipartite component of some $H\in \F$, then $C$
  itself is in $\F$, as $\F$ is hereditary. This implies that $C$ is $b$-thin. Therefore, we have shown
  that $\F$ is $a$-small/$b$-thin: every nonbipartite component has
  at most $a$ vertices and every bipartite component is $b$-thin.
\end{proof}

\subsection{Proof of the dichotomy for packing
  problems} \label{subsection:packing:theorem} Using the algorithm of
Section~\ref{subsection:subgraph:turing:upperbounds}, the hardness
results for the basic families proved in
Section~\ref{subsection:subgraph:turing:lowerbounds}, and the
characterization proved in Section~\ref{sec:subgraph-ramsey}, we can
prove Theorem~\ref{theorem:intro:packing}.

\restateintropacking*
\begin{proof}
  Let $\F$ be a hereditary class of graphs. If $\F$ is small/thin,
  then Theorem~\ref{theorem:kernel:packing} shows that \FSubgraphTest
  admits a polynomial many-one kernel. 
	
  If $\F$ is not small/thin, then it is a superset of one of the
  families listed in Theorem~\ref{theorem:ramsey:1}.  If $\F$ contains
  $\Ffountain{s}$ for some odd integer $s\ge 3$, then the \textup{WK[1]}-hard
  problem $\Ffountain{s}$-\Packing
  (Lemma~\ref{lemma:fountainpacking:wkhard}) can be reduced to
  \FPacking, hence \FPacking is also \textup{WK[1]}-hard. The situation is
  similar if $\F$ is a superset of any of the families listed in
  items 5--8 of Theorem~\ref{theorem:ramsey:1}. As explained in Section~\ref{section:outline:hardness}, if~$\F$ is a superset of \Fclique or \Fbiclique then \FPacking is \textup{W[1]}-hard and does not admit a polynomial many-one kernel unless \containment. Therefore, we have shown
  that if $\F$ is not small/thin, then \FPacking is \textup{WK[1]}-hard, \textup{W[1]}-hard, or \kPath-hard in all
  cases, completing the proof Theorem~\ref{theorem:intro:packing}.
\end{proof}


\section{Turing kernelization complexity of subgraph testing}

In this section, we prove Theorem~\ref{theorem:intro:turingsubgraph}
characterizing the hereditary classes $\F$ for which \FSubgraphTest
admits a polynomial Turing kernel. In
Section~\ref{subsection:subgraph:turing:upperbounds}, we invoke the
marking algorithm developed in Section~\ref{sec:comp-repr-sets} to
give a polynomial Turing kernel for splittable classes. In
Section~\ref{subsection:subgraph:turing:lowerbounds}, we establish the
basic hardness results that $\Fdiamondfan$-\SubgraphTest and
$\Fsubdivtree{s}$-\SubgraphTest for any $s\ge 1$ are \textup{WK[1]}-hard. In
Section~\ref{sec:subgraph-ramsey}, we prove
Theorem~\ref{theorem:ramsey:separator} characterizing splittable hereditary
classes. In Section~\ref{sec:proof-dich-subgr}, we put together all
these results to complete the proof of Theorem~\ref{theorem:intro:turingsubgraph}.

\subsection{Upper
  bounds} \label{subsection:subgraph:turing:upperbounds} To prove the
positive part of Theorem~\ref{theorem:intro:turingsubgraph} for
splittable classes, we try every possible image for the set $D$ that
realizes the $(a,b,c,d)$-split and then invoke
Lemma~\ref{lemma:kernel:generic} to obtain a bounded-size instance for
each possible way of fixing the image. 

 \restatesubgraphturing*
\begin{proof}
Let~$\F$ be a hereditary family that is $(a,b,c,d)$-splittable and let~$(G,H)$ be an instance of the \kFSubgraphTest problem. Recall that the parameter is~$k := |V(H)|$. We will construct a list~$\L$ of~$|V(G)|^{\Oh(1)}$ instances~$(G[X_1], H), \ldots, (G[X_t], H)$ of the \kFSubgraphTest problem, each of size polynomial in~$k$, such that~$H \subseteq G$ if and only if~$H \subseteq G[X_i]$ for some~$i \in [t]$. This easily implies the existence of a polynomial-size Turing kernel following Definition~\ref{definition:turing:kernelization}: we can query the oracle for the answer to each subinstance~$(G[X_i], H)$ of size polynomial in~$k$ and output the logical OR of the answers, which is the correct answer to the instance~$(G, H)$. The description in terms of a list of small instances highlights the fact that the Turing kernelization is non-adaptive and therefore amenable to parallelization.

The kernelization algorithm starts by searching for a set~$D$ that realizes an~$(a,b,c,d)$-split of~$H$. As such a set~$D$ has size at most~$c$, which is a constant, can we try all possible sets~$\binom{V(G)}{\leq c}$ to determine whether there is one that realizes the split. Note that given a candidate set, it is easy to determine in polynomial time whether it realizes an~$(a,b,c,d)$-split of~$H$ or not. If the split cannot be realized then~$H \not \in \F$ and the instance does not satisfy the input requirements; we output \no. In the remainder we can work with a set~$D \subseteq V(G)$ of size at most~$c$ that realizes the split. We then proceed as follows.

For each partial $H$-subgraph model~$\phi_i$ in~$G$ with domain~$D$, we invoke the algorithm of Lemma~\ref{lemma:kernel:generic} to the source graph~$G$, the query graph~$H$ with the separator~$D$, and the partial $H$-subgraph model~$\phi_0$. As~$D$ realizes an~$(a,b,c,d)$-split of~$H$, each connected component~$C$ of~$H - D$ that has size more than~$a$ is a $b$-thin bipartite graph in which the number of vertices whose closed neighborhood is not universal to~$N_H(C) \cap D$ is bounded by~$d$. As the domain of~$\phi_0$ is exactly~$D$, the lemma outputs a set~$X_i$ of size~$\Oh(k^{\Oh(a + b^2 + d)})$ such that~$G$ contains a full $H$-subgraph model that extends~$\phi_i$ if and only if~$G[X_i]$ contains a full $H$-subgraph model that extends~$\phi_i$. We add~$(G[X_i], H)$ to the list~$\L$. Since~$|D| \leq c$, which is a constant, the number of distinct partial $H$-subgraph models in~$G$ with domain~$D$ is~$\Oh(|V(G)|^c)$, which is polynomial in~$|V(G)|$ as~$c$ is a constant. As the computation of Lemma~\ref{lemma:kernel:generic} takes polynomial time for constant~$(a,b,c,d)$, the entire algorithm runs in polynomial time and the size of~$\L$ is bounded by a polynomial. 

To complete the proof it therefore suffices to show that if and only if~$H \subseteq G$ then~$H \subseteq G[X_i]$ for some~$i$. The reverse direction is trivial as~$G[X_i]$ is an induced subgraph of~$G$. For the forward direction, assume that~$\phi^*$ is a full~$H$-subgraph model in~$G$ and consider its restriction~$\phi^*|_D$ to the vertices of~$D$. Then~$\phi^*|_D$ is a partial $H$-subgraph model with domain~$D$, hence it occurred as a model~$\phi_i$ in our enumeration. As~$\phi^*$ is a full $H$-subgraph model in~$G$ that extends~$\phi_i = \phi^*|_D$, the guarantee of Lemma~\ref{lemma:kernel:generic} for set~$X_i$ ensures that~$G[X_i]$ contains a full $H$-subgraph model. Hence there is an index~$i$ such that~$H \subseteq G[X_i]$. By the argumentation given earlier, this concludes the proof.
\end{proof}

\subsection{Lower bounds} \label{subsection:subgraph:turing:lowerbounds}
In this section we present the polynomial-parameter transformations that establish Theorem~\ref{theorem:subgraphtest:lowerbounds}, which we repeat here for the reader's convenience.

\restatesubgraphlower*

The proof of Theorem~\ref{theorem:subgraphtest:lowerbounds} follows
from the following to lemmas, which show the \textup{WK[1]}-hardness of
\Fdiamondfan-\SubgraphTest and \Fsubdivtree{s}-\SubgraphTest by
polynomial-para\-meter transformations from \nRegularExactSetCover.
\begin{lemma}\label{lemma:subgraph:diamondfan:wkhard}
\Fdiamondfan-\SubgraphTest is \textup{WK[1]}-hard.
\end{lemma}
\begin{proof}
 We transform an instance~$(r, \S, U, n)$ of \nRegularExactSetCover in polynomial time into an equivalent instance~$(G, H)$ of \Fdiamondfan-\Packing with $H:= \Gdiamondfan{Q}$ and $Q:=n+1$. Note that the parameter is $k := |V(H)| =Q^2+Q+1\in \Oh(n^2)$ (see Figure~\ref{fig:basic}). Let $t:=n/r$, the number of sets in a solution. The host graph~$G$ is defined as follows.

\begin{itemize}
	\item Graph~$G$ contains the vertex set~$U$ as an independent set.
          \item We introduce $Q-t$ copies of $K_{2,Q}$ and identify into a single vertex $z$ one degree-$Q$ vertex from each a copy.
	\item For every size-$r$ set~$S \in \S$, we introduce a vertex $v_S$ adjacent to all vertices representing members of~$S$.
        \item For every size-$r$ set $S\in S$, we introduce a set
          $X_S$ of $Q-r\ge Q-n\ge 0$ vertices that are adjacent to both $v_S$ and
          $z$.
\end{itemize}
Note that the degree of every $v_S$ is exactly $Q$.

As the construction can be performed in polynomial time and produces an instance of the target problem whose parameter is suitably bounded, it remains to prove that the input and output instances are equivalent.

\begin{claim}
If~$(\S,U)$ has an exact cover, then~$H \subseteq G$.
\end{claim}
\begin{claimproof}
  Assume that~$(\S,U)$ has an exact cover~$S_1, \ldots, S_t \in
  \S$. For each~$S_i$ with~$i \in [t]$, the vertex $z$, vertex $v_S$,
  and the $Q$ neighbors of $v_S$ form a copy of $K_{2,Q}$. As the sets
  $S_i$'s are disjoint, these copies of $K_{2,Q}$ intersect only in
  $z$. Therefore, these copies together with the $Q-t$ copies
  introduced in the construction of $G$ form a $\Gdiamondfan{Q}$
  subgraph centered at $z$ in $G$.
\end{claimproof}

Let $c$ be the unique vertex of $H$ having degree $Q^2$. Graph $H$ contains $Q$ degree-$Q$ vertices at distance two 
   from $c$, let $W$ be the set of all these vertices.
\begin{claim} \label{claim:diamondfan:hardness:goodmodels}
If $\phi$ is a full subgraph model of $H$ in $G$, then $\phi(c)=z$.
\end{claim}
\begin{claimproof}
  Every vertex $v_S$ has degree $Q$ and every vertex in every $X_S$
  has degree 2. The vertices of the copies of $K_{2,Q}$ introduced in
  the construction of $G$ (other than $z$) have degree two or
  $Q$. Therefore, $\phi(c)\neq z$ is only possible if $\phi(c)\in
  U$. The vertices of $W$ have to be mapped to vertices of degree
  at least $Q$ at distance exactly two from $\phi(c)$ in $G$. As
  $\phi(c)\in U$, the vertices of $G$ at distance two from $\phi(c)$ are
  the vertices of $U$, the degree-2 vertices in some $X_S$, and the
  degree-2 vertices of some $K_{2,Q}$ adjacent to $z$. Therefore, only
  at most $|U|=n<Q$ of them can have degree at least $Q$, a
  contradiction.
\end{claimproof}

\begin{claim}
If~$H \subseteq G$, then~$(\S,U)$ has an exact cover.
\end{claim}
\begin{claimproof}
  Assume that $\phi$ is a full subgraph model of $H$ in $G$.  By
  Claim~\ref{claim:diamondfan:hardness:goodmodels}, we have
  $\phi(c)=z$.  The vertices of $W$ have to be mapped to vertices of
  $G$ at distance two from $z$. The candidates for these vertices are
  the vertices $v_S$ and the $Q-t$ degree-$Q$ vertices introduced in
  the $Q-t$ copies of $K_{2,Q}$. Therefore, at least $t$ vertices of
  $W$ are mapped to the $v_S$'s. Note that each vertex $v_S$ has
  degree exactly $Q$, hence the neighborhood of a vertex $w\in W$ is
  mapped bijectively to the neighborhood of $\phi(w)$. The
  neighborhoods of the vertices in $W$ are disjoint, hence the
  neighborhoods of the images should be disjoint as well. Therefore,
  the sets corresponding to the images of $W$ are disjoint, implying
  that there exists at least $t$ disjoint sets in $\S$. As these sets
  contain~at least $t \cdot r = n$ vertices in total, this is only
  possible if there are exactly $t$ of these sets and they cover all
  of~$U$, forming an exact cover.
\end{claimproof}

This concludes the proof of Lemma~\ref{lemma:subgraph:diamondfan:wkhard}.

\end{proof}

\begin{lemma}\label{lemma:subgraph:subdivtree:wkhard}
For every $s\ge 1$, \Fsubdivtree{s}-\SubgraphTest is \textup{WK[1]}-hard.
\end{lemma}
\begin{proof}
  We transform an instance~$(r, \S, U, n)$ of \nRegularExactSetCover
  in polynomial time into an equivalent instance~$(G, H)$ of \Fsubdivtree{s}-\Packing with~$H:=\Gsubdivtree{s}{Q}$ and $Q:=n+2$. Note
  that the parameter is $k := |V(H)| =Q^2+Qs+1\in \Oh(n^2)$ (see
  Figure~\ref{fig:basic}). Let $t:=n/r$, the number of sets in a
  solution. The host graph~$G$ is defined as follows.

\begin{itemize}
	\item Graph~$G$ contains the vertex set~$U$ as an independent set.
          \item We introduce a distinguished vertex $z$.
\item We introduce $Q-t\ge Q-n\ge 0$ vertices $x_1$, $\dots$, $x_{Q-t}$, connect each $x_i$ to $z$ with a path of length $s$, and attach $Q$ pendant vertices to each $x_i$.
\item For every size-$r$ set~$S \in \S$, we introduce a vertex $v_S$
  adjacent to all vertices representing members of~$S$, connect $v_S$
  and $z$ with a path of length $s$, and attach $Q-r\ge Q-n\ge 0$
  pendant vertices to $v_S$.
\end{itemize}

As the construction can be performed in polynomial time and produces
an instance of the target problem whose parameter is suitably bounded,
it remains to prove that the input and output instances are
equivalent.

\begin{claim}
If~$(\S,U)$ has an exact cover, then~$H \subseteq G$.
\end{claim}
\begin{claimproof}
  Assume that~$(\S,U)$ has an exact cover~$S_1, \ldots, S_t \in
  \S$. Consider the set $Z$ of vertices containing $v_{S_i}$ for $i
  \in [t]$ and $x_j$ for $j\in [Q-t]$. Each of these $Q$ vertices are
  connected to $z$ with a path of length $s$, and these paths intersect only in
  $z$. Each vertex of $Z$ has degree $Q+1$, that is, has $Q$ neighbors
  in addition to their neighbor on the path. Distinct vertices in $Z$
  have disjoint neighborhoods: these neighbors are either degree-one 
  (thus disjointness is trivial) or appear in $U$, where disjointness
  follows from the fact that $S_1$, $\dots$, $S_t$ are disjoint.
  Therefore, we have found a copy of $\Gsubdivtree{s}{Q}$ centered at
  $z$.
\end{claimproof}

Let $c$ be the unique vertex of $H$ having degree $Q$.
\begin{claim} \label{claim:subdivtree:hardness:goodmodels}
If $\phi$ is a full subgraph model of $H$ in $G$, then $\phi(c)=z$.
\end{claim}
\begin{claimproof}
  Clearly, $\phi(c)$ has degree at least $Q$. Vertices $x_i$ have
  degree $Q+1$, but they have only one neighbor with degree more than
  one, while every neighbor of $c$ has degree at least two. Vertices $v_S$
  have degree $Q+1$, but they have at most $r+1<Q$ neighbors with
  degree more than one.  Therefore, $\phi(c)=z$ is the only
  possibility.\end{claimproof}

\begin{claim}
If~$H \subseteq G$, then~$(\S,U)$ has an exact cover.
\end{claim}
\begin{claimproof}
  Assume that $\phi$ is a full subgraph model of $H$ in $G$.  By
  Claim~\ref{claim:subdivtree:hardness:goodmodels}, we have
  $\phi(c)=z$. Graph $H$ contains $Q$ degree-$(Q+1)$ vertices at
  distance $s$ from $c$, let $W$ be the set of all these vertices. The
  vertices of $W$ have to be  mapped to vertices of $G$ at distance exactly $s$ from
  $z$. The candidates for these vertices are the vertices $v_S$ and
  the $Q-t$ vertices $x_i$. Therefore, at least $t$ vertices of $W$
  are mapped to the $v_S$'s. Note that each vertex $v_S$ has degree
  exactly $Q+1$, hence the neighborhood of a vertex $w\in W$ is mapped
  bijectively to the neighborhood of $\phi(w)$. The neighborhoods of
  the vertices in $W$ are disjoint, hence the neighborhood of the
  images should be disjoint as well. Therefore, the sets corresponding
  to the images of $W$ are disjoint, implying that there exists at
  least $t$ disjoint sets in $\S$. As these sets contain at least $t \cdot r =
  n$ vertices in total, this is only possible if there are exactly $t$ of these sets and they cover all of~$U$, forming an exact cover.
\end{claimproof}

This concludes the proof of Lemma~\ref{lemma:subgraph:diamondfan:wkhard}.

\end{proof}

\subsection{Combinatorial characterizations}
\label{sec:subgraph-ramsey}
To prove Theorem~\ref{theorem:ramsey:separator} characterizing hereditary classes that are not splittable, we need first the following auxiliary result.

\begin{lemma} \label{lemma:ramsey:separator0}
Let~$\F$ be a hereditary graph family. Then at least one of the following holds:
\begin{enumerate}
\item There is an integer $M\ge 1$ such that every $H\in \F$ has a set
  $S_0\subseteq V(H)$ of at most $M$ vertices such that every
  component of $H- S_0$ has vertex cover number at most $M$.
\item $\F$ is a superset of at least one of
\begin{itemize}
\item $\Fpath$,
\item $\Fclique$,
\item $\Fbiclique$,
\item $Q\cdot \Fsubdivstar$, or
\item $Q\cdot \Ffountain{3}$.
\end{itemize}
\end{enumerate}
\end{lemma}
\begin{proof}
Assuming that $\F$ does not contain any of the forbidden families, there is a $Q\ge 1$ such that $\F$ does not contain any of
\begin{itemize}
\item $\Gclique{Q}$,
\item $\Gbiclique{Q}$,
\item $Q\cdot \Gsubdivstar{Q}$, and
\item $Q\cdot \Gfountain{3}{Q}$.
\end{itemize}
We set the following constants:
\begin{align*}
h&=P(Q,Q), \ \ \textup{(for the function $P$ in Corollary~\ref{theorem:path:ramsey2})}\\
M_1&=2Q(h+1)^2,\\
M_2&=\sum_{i=0}^{h}Q^i,\\
M&=\max\{M_1,M_2\}.
\end{align*}
Pick an arbitrary $H\in \F$ and compute a DFS tree for each component
of $H$. The height of the DFS forest obtain this way is at most $h=P(Q,Q)$ by
Corollary~\ref{theorem:path:ramsey2}, otherwise $H$ would contain
$\Gpath{Q}$, $\Gclique{Q}$, or $\Gbiclique{Q}$ as induced
subgraph. Let $L$ be the set of leaves of this forest and let
$Z=V(H)\setminus L$.

Let $Z^*$ be the those vertices of $Z$ that have at least $Q+1$
children in the DFS forest that belong to~$Z$. We claim that if there are $v_1$, $\dots$,
$v_{2Q}$ vertices of $Z^*$ on the same level of the DFS forest, then
$H$ contains $Q\cdot \Gfountain{s}{Q}$ or $Q\cdot \Gsubdivstar{Q}$ as
an induced subgraph. As $\F$ is hereditary, it would follow that
these graphs are in $H$, a contradiction. Let $x_{i,1}$, $\dots$,
$x_{i,Q+1}$ be children of $v_i$ in $Z$. As they are not leaves,
$x_{i,j}$ has a child $y_{i,j}$. By the properties of the DFS forest,
there is no edge between $\{x_{i,j},y_{i,j}\}$ and
$\{x_{i,j'},y_{i,j'}\}$ for any $j\neq j'$. If there is an edge
between $v_i$ and $y_{i,j'}$, then the triangle
$\{v_i,x_{i,j'},y_{i,j'}\}$ and the vertices $x_{i,j}$, $1 \le j \le
Q+1$, $j\neq j'$ form an induced $\Gfountain{3}{Q}$. If there is no
edge between $v_i$ and $y_{i,j}$ for any $1\le j \le Q$, then
$\{v_i,x_{i,1},\dots,x_{i,Q},y_{i,1},\dots, y_{i,Q}\}$ induces a
$\Gsubdivstar{Q}$. Thus for every $1\le i \le 2Q$, we get
an induced $\Gfountain{3}{Q}$ or $\Gsubdivstar{Q}$ on $v_i$, its children, and
its grandchildren. As all the $v_i$'s are on the same level of the DFS
forest, there is no edge between these $2Q$ graphs. Thus we get either
$Q$ independent copies of $\Gfountain{3}{Q}$ or $Q$ independent copies
of $\Gsubdivstar{Q}$, a contradiction.

We have proved that each level of the DFS forest contains less than
$2Q$ vertices of $Z^*$, implying that $|Z^*|\le 2Q(h+1)$. Let $S_0$
contain every vertex of $Z^*$ and every ancestor of every vertex in
$Z^*$: as the DFS forest has height at most $h$ (that is, at most
$h+1$ levels), we have $|S_0|\le (h+1)|Z^*|=2Q(h+1)^2=M_1$.  Observe
that if $u$ and $v$ are adjacent vertices in $H- S_0$ and $u$
is an ancestor of $v$ in the DFS forest, then the unique $u-v$ path of
the DFS tree is disjoint from $S_0$: if any vertex of this path is in
$S_0$, then $u$ itself is in $S_0$. Therefore, if $C$ is the set of
vertices of a component of $H- S_0$, then $C$ induces a
connected subtree of the DFS forest. As every vertex of $C\cap Z$ has
at most $Q$ children in $Z$ (otherwise it would be in $S_0$) and the
tree has height at most $h$, we have that $|C\cap Z|\le
\sum_{i=0}^{h}Q^i=M_2$. Observe furthermore that $C\cap Z$ is a vertex
cover of $H[C]$: if there is an edge between $u,v\in C\setminus Z$,
then it is an edge between two leaves of the DFS forest. Therefore,
$S_0$ is a set of at most $M_1\le M$ vertices such that every
component of $H- S_0$ has vertex cover number at most $M_2\le
M$, what we had to show.
\end{proof}

Now we are ready to prove Theorem~\ref{theorem:ramsey:separator}.
\restateramseysplittable*
\begin{proof}
  By Lemma~\ref{lemma:ramsey:separator0}, the assumption that $\F$ is
  not the superset of the first 5 families listed in the lemma implies
  that there is an $M\ge 1$ such that every $H\in \F$ has a set
  $S_0\subseteq V(H)$ of at most~$M$ vertices such that every component of
  $H- S_0$ has vertex cover number at most $M$.

  Assuming that $\F$ does not contain any of the forbidden
  families, there is a $Q\ge 1$ such that $\F$ does not contain
\begin{itemize}
\item $\Gpath{Q}$,
\item $\Gclique{Q}$,
\item $\Gbiclique{Q}$,
\item $Q\cdot \Gsubdivstar{Q}$,
\item $Q\cdot \Gfountain{s}{Q}$ for  any integer $3\le s \le Q+2$,
\item $Q\cdot \Goperahouse{s}$ for any odd integer $1 \le s \le Q$,
\item $Q\cdot \Gdoublebroom{s}$ for any odd integer $1 \le s \le Q$,
\item $Q\cdot \Glongfountain{s}{t}$ for any odd integer $3\le s \le Q+2$ and integer $1\le t \le Q$,
\item $\Gsubdivtree{s}{Q}$ for any $1\le s \le Q$,
\item $\Gdiamondfan{Q}$.
\end{itemize}
(The argument why there is such a finite $Q$ is the same as in the
proof of Lemma~\ref{lemma:ramsey:nonbipartite}.) 
We set
the following constants:
\begin{align*}
a&=a(Q),\ \ \textup{(for the function $a(Q)$ in Lemma~\ref{lemma:ramsey:nonbipartite0})}\\
b&=b(Q),\ \ \textup{(for the function $b(Q)$ in Lemma~\ref{lemma:ramsey:balancedbipartite0})}\\
k&=Q(Q^2+2Q)+Q(Q+1)+QM+Q^2M,\\
c&=(k+1)M,\\
d&=Q\cdot 2^{b+M}+b.
\end{align*}

Select an arbitrary $H\in F$ and let $S_0$ be given by
Lemma~\ref{lemma:ramsey:separator0}.  We say that component is $C$ of
$H- S_0$ is {\em good} if it has at most $a$ vertices or it is a 
$b$-thin bipartite graph with at most $d$ vertices whose neighborhood in~$H[C]$
is not universal to $N(C)\cap S_0$. Otherwise, we say that the
component is {\em bad,} which means that at least one of the following
holds:
\begin{enumerate}
\item $C$ is nonbipartite and has more than $a$ vertices.
\item $C$ is bipartite, but not $b$-thin.
\item $C$ is bipartite, is $b$-thin, but has
  more than $d$ vertices whose closed neighborhood in~$H[C]$ is not universal to
  $N(C)\cap S_0$.
\end{enumerate}
The following claim bounds the number of bad components. Then we show how to extend $S_0$ to destroy the bad components.
\begin{claim}\label{claim:boundbad}
 $H - S_0$ has at most $k$ bad components.
\end{claim}
\begin{proof}
Let us bound first the number 
  nonbipartite components that have at least
  $a$ vertices. As $H$ does not contain $\Gpath{Q}$ or $\Gclique{Q}$ as an induced subgraph, Lemma~\ref{lemma:ramsey:nonbipartite0} implies that each
  such nonbipartite component contains either 
 $\Gfountain{s}{Q}$ for some odd integer $3\le s\le Q+2$, 
 $\Glongfountain{s}{t}{Q}$ some odd integer $3\le s\le Q+2$ and integer $1\le t\le Q$, or
$\Goperahouse{s}{Q}$ for some odd integer $1\le s\le Q$ as induced subgraph.
That is, each nonbipartite component of size at least $a$ contains one of these $Q^2+2Q$ graphs, hence if there are
  $Q(Q^2+2Q)$ such components, one of these graphs appears at least
  $Q$ times, contradicting our assumptions.

  Next we bound the number of bipartite components that are not
  $b$-thin. As $\Gpath{Q}$, $\Gclique{Q}$, and $\Gbiclique{Q}$ do not
  appear in $H$ as induced subgraph,
  Lemma~\ref{lemma:ramsey:balancedbipartite0} implies that each such
  bipartite component contains $\Gsubdivstar{Q}$ or $\Gdoublebroom{s}{Q}$
  for some $1\le s \le Q$ as induced subgraph.  Therefore, if there
  are $Q(Q+1)$ bipartite non-thin components, one of these graphs appears at least $Q$
  times, contradicting our assumptions.

  Finally, we bound the number of components that are bipartite,
  $b$-thin, but have more than $d$ vertices whose closed neighborhood
  is not universal to $N(C)\cap S_0$. In particular, there is a set
  $X_0$ of at least $d-b$ such vertices in the larger bipartite
  class. Each vertex in the larger class of $C$ has degree at most
  $b+|S_0|\le b+M$ in~$H$, hence we can partition $X_0$ into at most
  $2^{b+M}$ classes according to their neighborhood. Therefore, there
  is a set $X\subseteq X_0$ of at least $|X_0|/2^{b+M}\ge
  (d-b)/2^{b+M}=Q$ such vertices in the larger side of $C$ whose
  neighborhoods are the same. There are two possibilities why the
  closed neighborhoods of the vertices in $X$ are not universal to
  $N(C)\cap S_0$: either some neighbor of $X$ is not universal to
  $N(C)\cap S_0$, or the vertices of $X$ are not universal to
  $N(C)\cap S_0$.
\begin{figure}[t]
\begin{center}
{\small \svg{\linewidth}{diamondfan2}}
\caption{Proof of Theorem~\ref{theorem:ramsey:separator}, Claim~\ref{claim:boundbad}. Finding (a) a \Gdiamondfan{4} or (b) a \Gsubdivtree{5}{4} in four bad components.}\label{fig:dfan}
\end{center}
\end{figure}

  The first case is when every vertex of $X$ is universal to
  $N(C)\cap S_0$, but they have a neighbor $x\in C$ that is not
  adjacent to some $z\in N(C)\cap S_0$. If there are at least $Q|S_0|$
  components where this case happens, then there are at least $Q$
  of them for which this case happens with the same $z\in S_0$. Then
  the corresponding vertices $x$ and sets $X$ form a $\Gdiamondfan{Q}$ (see Figure~\ref{fig:dfan}(a)).

  The second case is that the vertices in $X$ are not adjacent to some
  vertex $z\in N(C)\cap S_0$.  The fact that $z$ is in $N(C)$ implies
  that $z$ has a neighbor $y\in C$. Let us find a shortest path in
  $H[C]$ from each vertex of $X$ to $y$. As every vertex in $X$ has
  the same set of neighbors, we may assume that the second vertex is
  the same for each of these $|X|$ paths. In other words, there is a
  $v\in C$ and a $v-y$ path $P$ such that $v$ is adjacent to every
  vertex in $X$ and, for every $x\in X$, the path $xP$ is a shortest
  $x-y$ path. This implies that $P$ is an induced path and $x$ is not
  adjacent to any vertex of $P$ except $v$. As $H$ does not contain
  $\Gpath{Q}$ as induced subgraph, we also get that $|P|<Q$. If there are at least
  $Q^2|S_0|$ components where this subcase happens, then there are $Q$
  of them for which this subcase happens with the same $z\in S_0$ and
  the length of $P$ is the same integer $1\le s \le Q$. Then the
  corresponding paths $P$ and sets $X$ form a $\Gsubdivtree{s}{Q}$ (see Figure~\ref{fig:dfan}(b)).

  Summing up all cases, we get that $H- S_0$ has at most
  $k=Q(Q^2+2Q)+Q(Q+1)+QM+Q^2M$ bad components.  \cqed\end{proof}

We have shown that at most $k$ of the components of $H- S_0$
are bad components. Then we destroy these bad components by extending
$S_0$ by a minimum vertex cover of each bad component $C$; let $S$ be
the resulting set of vertices. We have seen that
(Lemma~\ref{lemma:ramsey:separator0}) that each component $C$ has a
vertex cover of size $M$, hence $|S|\le |S_0|+kM\le (k+1)M=c$. Observe
that if $C$ is a bad component of $H- S_0$, then every vertex
of $C\setminus S$ is an isolated vertex of $H- S$. Moreover,
the good components of $H- S_0$ are unaffected by extending
$S_0$ to $S$ and it is also true that $N(C)\cap S_0=N(C)\cap S$ for
every good component $C$ of $H- S_0$. Therefore, every
component $C$ of $H- S$ either has size at most $a$ or it is a
$b$-thin bipartite graph having at most $d$ vertices whose closed
neighborhood is not universal to $N(C)\cap S$. As this is true for
every $H\in \F$, we have shown that $\F$ is splittable.
\end{proof}

\subsection{Proof of the dichotomy for subgraph testing}
\label{sec:proof-dich-subgr}

Using the algorithm of
Section~\ref{subsection:subgraph:turing:upperbounds}, the hardness
results for the basic families proved in
Section~\ref{subsection:subgraph:turing:lowerbounds}, the
characterization proved in Section~\ref{sec:subgraph-ramsey}, and the
hardness results for \FPacking obtained in
Section~\ref{subsection:packing:lowerbounds}, we can prove
Theorem~\ref{theorem:intro:turingsubgraph}.

  \restatemainsubgraph*
\begin{proof}
  Let $\F$ be a hereditary class of graphs. If $\F$ is splittable,
  then Theorem~\ref{theorem:kernel:subgraph} shows that \FSubgraphTest
  admits a polynomial Turing kernel. Otherwise,
  Lemma~\ref{theorem:fpacking:pvsnp} gives a list of classes such that
  one of these classes is fully contained in $\F$. If $\F$ is a
  superset of $\Fpath$, then \FSubgraphTest is clearly \textsc{Long
    Path}-hard. If $\F$ is superset of $\Fclique$ or~$\Fbiclique$, then \FSubgraphTest
  W[1]-hard~\cite{DowneyF13,Lin14}.

  If $\F$ contains $n\cdot \Gsubdivstar{n}$ (that is, the disjoint
  union of $n$ copies of $\Gsubdivstar{n}$, then, as $\F$ is
  hereditary, it also contains $t\cdot \Gsubdivstar{n}$ for every
  $t,n\ge 1$. This means that an instance $(G,H,t)$ of \Packing with
  $H=\Gsubdivstar{n}$ can be expressed as an instance $(G,H')$ of
  $\FSubgraphTest$ with $H'=t\cdot H$. Therfore,
  $\Fsubdivstar$-\Packing, which was shown to be \textup{WK[1]}-hard in
  Theorem~\ref{claim:subdivstar:hardness:goodmodels}, can be reduced
  to $\FSubgraphTest$, implying that the latter problem is \textup{WK[1]}-hard
  as well. The situation is similar for items 5-8 in
  Theorem~\ref{theorem:ramsey:separator}: a \textup{WK[1]}-hard packing problem can
  be reduced to \FSubgraphTest.

  If $\F$ contains $\Fsubdivtree{s}$ for some $s\ge 1$, then
  $\FSubgraphTest$ is \textup{WK[1]}-hard by
  Theorem~\ref{lemma:subgraph:subdivtree:wkhard}. Similarly, if $\F$
  contains $\Fdiamondfan$, then $\FSubgraphTest$ is \textup{WK[1]}-hard by
  Theorem~\ref{lemma:subgraph:diamondfan:wkhard}. Therefore, we have
  shown that if $\F$ is not splittable, then it is \textsc{Long Path}-
  or \textup{WK[1]}-hard in all cases, completing the proof of
  Theorem~\ref{theorem:intro:turingsubgraph}.
\end{proof}

\section{Many-one kernelization complexity of subgraph testing}

For the case of many-one kernelization, we cannot determine the kernelization complexity for all hereditary graph families \F. The reason is that the complexity landscape is very diverse in this case, which we show by presenting several surprising kernelization upper and lower bounds. Obviously, since many-one kernelization is more restrictive than Turing kernelization, the negative results from Turing kernelization for \kFSubgraphTest carry over. However, we will see that for some of the families~\F where polynomial-size Turing kernels exist there is no polynomial-size many-one kernel unless \containment.

\subsection{Lower bounds} \label{subsection:subgraph:karp:lowerbounds}

\subsubsection{Kernelization lower bounds by OR-cross-composition}
The kernelization lower bounds presented until this point employed polynomial-parameter transformations, as they serve simultaneously as \textup{WK[1]}-hardness proofs (obtaining Turing kernelization lower bounds under the assumption that no \textup{WK[1]}-hard problem admits a polynomial Turing kernel) and transformations from incompressible problems (obtaining many-one kernelization lower bounds under the assumption that \ncontainment). The lower bounds in Section~\ref{subsection:subgraph:karp:lowerbounds} have to employ different machinery as they apply only to many-one kernelization (Theorem~\ref{theorem:kernel:subgraph} provides polynomial-size Turing kernels for these problems). We use the technique of OR-cross-composition~\cite{BodlaenderJK14}, which builds on earlier results by Bodlaender et al.~\cite{BodlaenderDFH09} and Fortnow and Santhanam~\cite{FortnowS11}.

\begin{definition} \label{definition:poly:eqv:relation}
An equivalence relation~\eqvr on $\Sigma^*$ is called a \emph{polynomial equivalence relation} if the following two conditions hold:
\begin{enumerate}
	\item There is an algorithm that given two strings~$x,y \in \Sigma^*$ decides whether~$x$ and~$y$ belong to the same equivalence class in~$(|x| + |y|)^{\Oh(1)}$ time.
	\item For any finite set~$S \subseteq \Sigma^*$ the equivalence relation~$\eqvr$ partitions the elements of~$S$ into at most~$(\max _{x \in S} |x|)^{\Oh(1)}$ classes.
\end{enumerate}
\end{definition}
\begin{definition} \label{definition:cross:composition}
Let~$L \subseteq \Sigma^*$ be a set and let~$\Q \subseteq \Sigma^* \times \mathbb{N}$ be a parameterized problem. We say that~$L$ \emph{OR-cross-composes} into~$\Q$ if there is a polynomial equivalence relation~$\eqvr$ and an algorithm that, given~$r$ strings~$x_1, x_2, \ldots, x_r \in \Sigma^*$ belonging to the same equivalence class of~$\eqvr$, computes an instance~$(x^*,k^*) \in \Sigma^* \times \mathbb{N}$ in time polynomial in~$\sum _{i \in [r]} |x_i|$ such that:
\begin{enumerate}
	\item~$(x^*, k^*) \in \Q \Leftrightarrow x_i \in L$ for some~$i \in [r]$,
	\item~$k^*$ is bounded by a polynomial in~$\max _{i \in [r]} |x_i|+\log r$.
\end{enumerate}
\end{definition}
\begin{theorem}[\cite{BodlaenderJK14}] \label{theorem:cross:composition:no:kernel}
If a set~$L \subseteq \Sigma^*$ is NP-hard under many-one reductions and~$L$ OR-cross-composes into the parameterized problem~$\Q$, then there is no polynomial many-one kernel for~$\Q$ unless \containment.
\end{theorem}

\subsubsection{Canonical template graphs for packing problems}

Before presenting the two OR-cross-compositions that prove the two main lower bounds for subgraph testing problems, we discuss the common general idea behind the constructions and give some preliminary lemmas. We will give a superpolynomial many-one kernelization lower bound for detecting a subdivided star together many vertex-disjoint triangles as a subgraph (Theorem~\ref{theorem:karp:lowerbound:onesubstar:manytriangles}), and for testing two subdivided stars together with many vertex-disjoint $P_3$'s as a subgraph (Theorem~\ref{theorem:karp:lowerbound:twosubstar:manyps}; recall that~$P_3$ is the length-two path on three vertices). Hence in both cases the task is to detect a constant number (one or two) large, constant-radius subgraphs together with many constant-size subgraphs. For an OR-cross-composition, we have to embed the logical OR of many instances of an NP-hard problem into a single instance of the target problem with a small value of the parameter. To obtain this OR behavior we make use of the fact that both for~$K_3$-packing and for~$P_3$-packing there are \emph{canonical template graphs} containing NP-complete instances: there are polynomial-time constructable graph families~$\G^{K_3}$ and~$\G^{P_3}$ such that all size-$n$ instances of the NP-complete \XTC can be reduced to~$K_3$ packing (respectively~$P_3$ packing) instances on induced subgraphs of the $n$-th member of the family~$\F^{K_3}$ (resp.~$\F^{P_3}$).

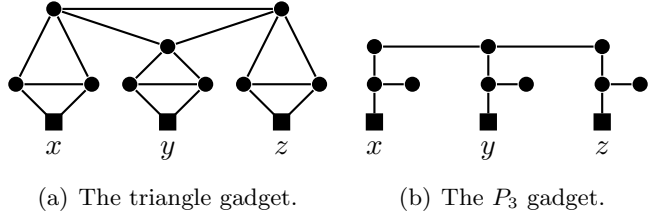
\begin{figure}[t]
\begin{center}
\subfigure[The triangle gadget.]{\label{fig:trianglegadget}
\begin{tikzpicture}[thick,>=stealth,x=0.5cm,y=0.5cm]
\foreach \i \j \k in {a/0/0,b/3/0,c/6/0,al/-1/1,ar/1/1,bl/2/1,br/4/1,cl/5/1,cr/7/1,at/0/3,bt/3/2,ct/6/3}
	\node[fill=black,circle,inner sep=0,minimum size=0.2cm] (\i) at (\j, \k) {};
\foreach \i \j in {a/al,a/ar,at/al,at/ar,b/bl,b/br,bt/bl,bt/br,c/cl,c/cr,ct/cl,ct/cr,at/bt,bt/ct,ct/at,al/ar,bl/br,cl/cr}
	\draw (\i) -- (\j);
\node[fill=black,rectangle,minimum size=0.2cm, inner sep=0,label=below:$x$,draw] (adesc) at ($(a)$) {};
\node[fill=black,rectangle,minimum size=0.2cm, inner sep=0,label=below:$y$,draw] (bdesc) at ($(b)$) {};
\node[fill=black,rectangle,minimum size=0.2cm, inner sep=0,label=below:$z$,draw] (cdesc) at ($(c)$) {};
\end{tikzpicture}
}
\subfigure[The~$P_3$ gadget.]{\label{fig:pathgadget}
\begin{tikzpicture}[thick,>=stealth,x=0.5cm,y=0.5cm]
\foreach \i \j \k in {a/0/0,b/3/0,c/6/0,al/0/1,ar/1/1,bl/3/1,br/4/1,cl/6/1,cr/7/1,at/0/2,bt/3/2,ct/6/2}
	\node[fill=black,circle,inner sep=0,minimum size=0.2cm] (\i) at (\j, \k) {};
\foreach \i \j in {a/al,at/al,b/bl,bt/bl,c/cl,ct/cl,at/bt,bt/ct,al/ar,bl/br,cl/cr}
	\draw (\i) -- (\j);
\node[fill=black,rectangle,minimum size=0.2cm, inner sep=0,label=below:$x$,draw] (adesc) at ($(a)$) {};
\node[fill=black,rectangle,minimum size=0.2cm, inner sep=0,label=below:$y$,draw] (bdesc) at ($(b)$) {};
\node[fill=black,rectangle,minimum size=0.2cm, inner sep=0,label=below:$z$,draw] (cdesc) at ($(c)$) {};
\end{tikzpicture}
}
\caption{Gadgets for the canonical template graph families~$\G^{K_3}$ and~$\G^{P_3}$. The vertices represented by circles are private to the gadget, while the three vertices~$x,y,z$ represented by squares are shared by other gadgets.}
\end{center}
\end{figure}

\begin{definition} \label{definition:canonicalgraphs}
For each positive integer~$n$, the $n$-th graph~$\G^{K_3}_n$ in the sequence of \emph{canonical template graphs for triangle packing}~$\G^{K_3}$ is obtained as follows:
\begin{itemize}
	\item Start with an independent set~$U$ of size~$n$. 
	\item For each set~$\{x,y,z\} \in \binom{U}{3}$, add a copy of the nine private vertices in the triangle gadget of Figure~\ref{fig:trianglegadget} to the graph and connect them to the vertices~$x,y,z \in U$ as in Figure~\ref{fig:trianglegadget}.
\end{itemize}
The $n$-th graph~$\G^{P_3}_n$ in the sequence of \emph{canonical template graphs for~$P_3$ packing} is obtained in the same way, using the path gadget of Figure~\ref{fig:pathgadget} instead of the triangle gadget.
\end{definition}

From this definition it easily follows that for each~$n$, the graphs~$\G^{K_3}_n$ and~$\G^{P_3}_n$ can be constructed in time polynomial in~$n$ and have size exactly~$n + 9\binom{n}{3}$. The graph families are called canonical because the following problems are NP-complete.

\problemdef{\TrianglePackingCanonical}
{An integer~$n$ encoded in unary and a subset~$S$ of the vertices of~$\G^{K_3}_n$.}
{Can the vertices of~$\G^{K_3}_n[S]$ be partitioned into triangles?}

\problemdef{\PathPackingCanonical}
{An integer~$n$ encoded in unary and a subset~$S$ of the vertices of~$\G^{P_3}_n$.}
{Can the vertices of~$\G^{P_3}_n[S]$ be partitioned into~$P_3$'s?}

\begin{lemma} \label{lemma:canonicalpacking:npcomplete}
\TrianglePackingCanonical and \PathPackingCanonical are NP-complete.
\end{lemma}
\begin{proof}
Membership in NP is trivial. Completeness follows from the fact that the existing NP-completeness reductions for triangle packing and~$P_3$ packing, which follow from more general results by Kirkpatrick and Hell~\cite{KirkpatrickH78}, construct induced subgraphs of the canonical graphs. In particular, in these two specific cases their NP-completeness proof reduces an instance of \XTC to instances of triangle packing or~$P_3$ packing. An instance of \XTC consists of a universe~$U$ among with a collection~$T \subseteq \binom{U}{3}$ of size-three subsets of the universe, and asks whether there is a subset~$T' \subseteq T$ such that each element of~$U$ is contained in exactly one set of~$T'$. To reduce such an instance~$(U,T)$ of \XTC to triangle packing or~$P_3$ packing, it suffices to perform the construction of Definition~\ref{definition:canonicalgraphs} but to only make gadgets for the size-three subsets of~$U$ that appear in~$T'$. Consequently, we can obtain instances of \TrianglePackingCanonical and \PathPackingCanonical that are equivalent to~$(U,T$) by letting~$S$ contain the vertices~$U$ together with those of the gadgets that would be created in the reduction.
\end{proof}

\subsubsection{One subdivided star and many triangles}

Before presenting the technical details, we describe the main intuition of the construction. We OR-cross-compose a sequence of instances of the NP-complete \TrianglePackingCanonical problem and choose a polynomial equivalence relation that enforces that all instances we are working with share the same value of~$n$ and have the same size set~$S$. The instance of the \kFSubgraphTest problem that we create contains the canonical graph~$\G^{K_3}_n$ for instances of size~$n$. For each vertex of~$\G^{K_3}_n$ we will add a private triangle of two new vertices that cannot be used in any other triangles; this private triangle will therefore be a \emph{cheap} way to get a triangle subgraph, as they consume only one vertex of~$\G^{K_3}_n$. The only other triangles in~$G$ will be within~$\G^{K_3}_n$, which are \emph{expensive} since they consume three vertices of~$\G^{K_3}_n$. We encode the input instances through the following mechanism. For each input instance~$i$, which is completely characterized by the set~$S_i$ for which it asks whether~$\G^{K_3}_n[S_i]$ can be partitioned into triangles, we add one possibility of realizing a large subdivided star to the host graph~$G$, such that realizing the subdivided star in this way blocks the cheap triangles using~$S_i$. As~$H$ will consist of one large subdivided star together with many disjoint triangles, to find an $H$-subgraph one must find a subdivided star together with many disjoint triangles. We will choose the number of triangles in such a way that, after a realization of a subdivided star is chosen that blocks the cheap triangles of~$S_i$, it will be optimal to take the cheap triangles for all vertices of~$\G^{K_3}_n - S_i$, effectively eliminating these vertices. The number of required triangles will be set such that, after taking these cheap triangles, one has to partition the remaining graph~$\G^{K_3}_n[S_i]$ into triangles to get a sufficient number. Armed with this intuition we present the formal proof. 


\restatekarponesubstarmanytriangles*

\begin{proof}
We will prove that \TrianglePackingCanonical OR-cross-composes into the problem of testing whether a graph~$G$ contains a subgraph of the form~$H \in \Fsubdivstar + \ell \cdot K_3$. Following the definition of OR-cross-composition we first define a polynomial equivalence relation~\eqvr on strings. We let all strings that do not encode a valid instance of the problem be equivalent, and we let valid instances be equivalent if they agree on the value of~$n$ and the size of the deleted set~$S$. As the size of~$\G^{K_3}_n$ is polynomial in~$n$, for each value of~$n$ there are only polynomially many choices for the size of the set~$S$ which easily implies that this is a polynomial equivalence relation.

In the remainder, it suffices to show how to OR-cross-compose a sequence of instances that is equivalent under~\eqvr. If the strings do not encode valid instances, then we output a constant-size \no-instance as the result of the OR-cross-composition. From now on we may therefore assume that~$x_1, \ldots, x_r$ is a series of strings encoding valid instances~$(n, S_1), \ldots, (n, S_r)$ whose sets~$S_i$ all have the same size~$m$. If~$n < 10$ or~$m < 10$, then we can solve all instances in constant time and output the appropriate answer; hence we may assume~$n,m \geq 10$. Let~$\G := \G^{K_3}_n$ be the canonical graph for the current value of~$n$. We construct graphs~$G$ and~$H$ such that~$H \subseteq G$ if and only if there is a \yes-instance among the inputs.
\begin{enumerate}
	\item We initialize~$G$ as a copy of~$\G$. For every vertex~$v \in V(\G)$, we add a \emph{dummy} vertex~$v'$ and an \emph{activator} vertex~$v''$ and turn~$\{v,v',v''\}$ into a triangle.
	\item For each instance number~$i \in [r]$, we create an \emph{instance selector}~$u_i$ and make it adjacent to~$\{v'' \mid v \in S_i\}$.
\end{enumerate}

This concludes the description of~$G$. 

\begin{observation}\label{observation:triangle:lowerbound}
Each triangle in~$G$ contains at least one vertex of~$V(\G)$. Each triangle in~$G$ that contains less than three vertices of~$V(\G)$ contains one vertex~$v \in V(\G)$ and the corresponding dummy vertex~$v'$ and activator vertex~$v''$.
\end{observation}

Define~$t := |V(\G)| - m + m/3$. We let~$H'$ be a subdivided star with~$m$ leaves and pick~$H := H + t \cdot K_3$ to complete the description of the \kFSubgraphTest instance~$(G,H)$. Recall that the parameter to the problem is the value~$k := |V(H)| = (2m + 1) + 3 t = (2m+1) + 3(|V(\G)| - m + m/3)$. As~$\G = \G^{K_3}_n$ has size polynomial in~$n$ and~$m$ measures the size of a vertex subset of~$\G$, it follows that~$k$ is polynomial in~$n$ and therefore in the size of the largest input instance. Hence~$k$ is bounded appropriately for an OR-cross-composition. It is easy to see that the instance~$(G,H)$ can be constructed in polynomial time. It remains to prove that it acts as the logical OR of the input instances. To this end, we first establish the following claim.

\begin{claim} \label{claim:triangle:tightbudget}
Any packing of~$t$ vertex-disjoint triangles in~$G$ that uses at most~$|V(\G)| - m$ activator vertices, uses exactly~$|V(\G)| - m$ activator vertices and all vertices of~$V(\G)$.
\end{claim}
\begin{claimproof}
First consider a packing using exactly~$|V(\G)| - m$ activator vertices. By Observation~\ref{observation:triangle:lowerbound}, every triangle uses at least one vertex of~$V(\G)$. For each activator vertex used, we can get a cheap triangle containing only one vertex of~$V(\G)$. The remaining~$t - (|V(\G)| - m) = m/3$ triangles each use three vertices from~$V(\G)$. Since~$m$ vertices of~$V(\G)$ are left after discarding those used in cheap triangles, these remaining~$m$ vertices must all be used in a triangle to get the additional~$m/3$ triangles needed to find~$t$ in total. Using similar arguments it is easy to show that if fewer than~$|V(\G)| - m$ activator vertices are used, one cannot obtain~$t$ in total. The claim follows.
\end{claimproof}

\begin{claim}
$H \subseteq G$ if and only if there is an~$i \in [r]$ such that~$\G^{K_3}_n[S_i]$ can be partitioned into triangles.
\end{claim}
\begin{claimproof}
($\Leftarrow$) We first prove the easy reverse direction. Assume that~$\G^{K_3}_n[S_i]$ can be partitioned into triangles. As~$|S_i| = m$, this implies that~$G^{K_3}_n[S_i]$ contains~$m/3$ vertex-disjoint triangles. We extend this packing of triangles to an $H$-subgraph in~$G$ as follows. For each vertex~$v \in V(\G) \setminus S_i$, add the triangle~$\{v,v',v''\}$ to the subgraph model, increasing the number of triangles to~$(|V(\G)| - m) + m/3$ in total. Observe that none of these triangles contains vertices~$\{v'' \mid v \in S_i\}$. Hence we can realize a subdivided star that is disjoint from this packing of triangles by centering the star at vertex~$u_i$ and using the~$2m$ vertices~$\{v', v'' \mid v \in S_i\}$ for the~$m$ legs of the star. Using the construction of~$G$, it is easy to verify that all the required edges are present. As the model of the subdivided star is disjoint from the model of the triangles, we find~$H$ as a subgraph in~$G$.

($\Rightarrow$) For the forward direction, assume that~$\phi$ is a full $H$-subgraph model in~$G$. Let~$c \in V(H)$ be the center of the subdivided star, and let~$\phi'$ be the model of the subdivided star obtained by restricting~$\phi$. We will first establish that~$\phi'(c) = \phi(c) \in \{u_1, \ldots, u_r\}$.

\begin{itemize}
	\item As dummy vertices of the form~$v'$ for~$v \in V(\G)$ have degree two in~$G$, they cannot model the center of a star of degree~$m \geq 10$.
	\item Suppose that~$\phi'(c) \in V(\G)$. Then the images of the~$m$ subdivider vertices of the star lie in~$N_G(\phi'(c))$. Observe that all neighbors of~$v \in V(\G)$ except~$v', v''$ belong to~$V(\G)$. Hence at least~$m - 2$ vertices of~$V(\G)$ are used in the model of the subdivided star. By Observation~\ref{observation:triangle:lowerbound}, each triangle in~$G$ uses at least one vertex of~$V(\G)$. As the models of the triangles are disjoint from the model of the subdivided star and there are at most~$|V(\G)| - (m - 2)$ vertices of~$V(\G)$ not used by the star, there can be at most~$|V(\G)| - (m - 2) < t$ triangles in the subgraph model. Hence~$\phi$ does not model~$H$; a contradiction.
	\item Suppose that~$\phi'(c)$ is an activator vertex of the form~$v''$ for~$v \in V(\G)$. Observe that the activator vertex~$v''$ has the unique $G$-neighbor~$v$ in the set~$V(\G)$, it has the dummy vertex~$v'$ as a neighbor, and it can have many instance selector vertices as $G$-neighbors. The vertices~$v$ and~$v'$ in~$G$ cannot be the images of vertices of two \emph{different} legs of the subdivided star: if~$v'$ is used in model~$\phi'$, then it either models a subdivider vertex (which means that its unique neighbor~$v$ distinct from~$\phi'(c)$ models the degree-one leaf of that leg of the star), or~$v'$ models a degree-one leaf which means that its only other neighbor~$v$ must model the subdivider vertex of the same leg of the star. Finally, if~$v'$ is not used in the model, then~$\{v, v'\}$ models part of (at most) one leg of the star. 
	
	From these observations, we deduce the following. Since the only $G$-neighbors of the activator vertex~$\phi'(c) = v''$ are~$v$,~$v'$, and instance selector vertices, at least~$m - 1$ subdivider vertices of the star are mapped to instance selectors by~$\phi'$. As the only neighbors of instance selector vertices in~$G$ are activator vertices, for each instance selector that is used by~$\phi'$, an activator vertex distinct from~$\phi'(c) = v''$ must also be used by~$\phi'$. Hence we find: if~$v \in V(\G)$ is not used in model~$\phi'$, then~$m$ instance selector vertices are used, implying that~$m$ activator vertices distinct from~$\phi'(c)$ are also used, giving a total number of~$m + 1$ activator vertices if we include~$\phi'(c) = v''$ itself. By Claim~\ref{claim:triangle:tightbudget} this shows that the subdivided star model~$\phi'$ does not leave enough activator vertices free to realize~$t$ disjoint triangles. On the other hand, if~$v \in V(\G)$ is used in the model~$\phi'$ of the subdivided star, still at least~$m$ activator vertices are used. But by Claim~\ref{claim:triangle:tightbudget}, if~$v \in V(\G)$ is not used to model a triangle, then we cannot pack~$t$ triangles using the~$|V(\G)| - m$ activator vertices that are left after discarding those used in the subdivided star. Again,~$\phi$ does not model~$H$ and we reach a contradiction.
\end{itemize}

As the vertices of~$G$ can be partitioned into instance selectors, dummies, activators, and vertices of~$V(\G)$, we find that if~$\phi$ is a full $H$-subgraph model in~$G$, then~$\phi(c)$ is an instance selector, say~$u_i$. As~$\deg_G(u_i) = \deg_H(c)$, all $G$-neighbors of~$u_i$ (which are activator vertices) are used to model vertices that subdivide the star. Hence there is a packing of~$t$ disjoint triangles in the graph~$G - N_G[u_i]$. Observe that for each vertex~$v \not \in S_i$, the graph~$G - N_G[u_i]$ contains~$v'$ and~$v''$. As these two vertices do not occur in any triangles except~$\{v,v',v''\}$, we may assume that the packing of triangles in~$G - N_G[u_i]$ contains all triangles~$\{v,v',v''\}$ for~$v \in V(\G) \setminus S_i$. This accounts for~$|V(\G)| - m$ of the~$t$ triangles. The remaining~$t - (|V(\G)| - m) = m/3$ triangles cannot use any dummy vertex, as dummy vertices only form triangles with activator vertices that have already been used fully. So the remaining~$m/3$ triangles are in fact triangles in~$G[S_i] = \G[S_i]$. As~$m = |S_i|$ this implies that~$\G[S_i] = \G^{K_3}_n[S_i]$ can be partitioned into triangles, which concludes the proof.
\end{claimproof}

\noindent As all criteria of an OR-cross-composition have been met, Theorem~\ref{theorem:karp:lowerbound:onesubstar:manytriangles} follows from Theorem~\ref{theorem:cross:composition:no:kernel}.
\end{proof}

\subsubsection{Two subdivided stars and many \texorpdfstring{$P_3$}{P3}'s}

The next kernelization lower bound concerns pattern graphs that contain two subdivided stars together with many~$P_3$'s. The following lemma employs a padding argument to prove that \PathPackingCanonical remains NP-complete under some degree restrictions. These degree restrictions will be useful to during the OR-cross-composition to reason about how the center of a large subdivided star can appear in the constructed host graph.

\begin{lemma} \label{lemma:canonicalpacking:larges:npcomplete}
\PathPackingCanonical remains NP-complete when restricted to instances~$(n, S)$ where, if~$n \geq 10$, we have~$|S| > \Delta (\G^{P_3}_n) + 1$.
\end{lemma}
\begin{proof}
Observe that for~$n \geq 10$ the maximum degree the graph~$\G^{P_3}_n$ is determined by the degree of the vertices~$U$ (which all have the same degree), since the private vertices of the gadgets have degree three. A vertex~$u \in U$ is adjacent in~$\G^{P_3}_n$ to one vertex of each gadget corresponding to a subset in~$\binom{U}{3}$ that contains~$u$. As there are~$\binom{|U| - 1}{2}$ such subsets, the degree of~$u \in U$ is~$\binom{n-1}{2}$. We can alter the NP-completeness transformation from \XTC to \PathPackingCanonical described in Lemma~\ref{lemma:canonicalpacking:npcomplete} as follows.

Given an input~$(U, T)$ of \XTC with~$|U| = n$, add~$3n^5$ new elements~$x_1, y_1, z_1, \ldots, x_{n^5}, y_{n^5}, z_{n^5}$ to the universe to obtain~$U'$ and add all size-three subsets of these new elements to~$T$ to obtain~$T'$. It is easy to see that instance~$(U,T)$ is equivalent to instance~$(U',T')$. Now reduce~$(U',T')$ to an instance~$(n + 3 n^5, S)$ of \PathPackingCanonical as described in Lemma~\ref{lemma:canonicalpacking:npcomplete}. The maximum degree of the canonical graph~$\G^{P_3}_{n + 3 n^5}$ is at most~$\binom{n + 3 n^5 - 1}{2}$. Recall that~$S$ contains the vertices of the universe of the exact cover instance together with the gadgets created for triples that are contained in~$T'$. As~$T'$ contains at least~$\binom{3n^5}{3}$ sets, the constructed set~$S$ contains the nine private vertices of at least~$\binom{3n^5}{3}$ gadgets. As~$n \geq 10$, a simple computation shows that~$\Delta(\G^{P_3}_{n + 3 n^5}) \leq \binom{n + 3 n^5 - 1}{2} < 9 \binom{3n^5}{3} - 1 \leq |S| - 1$. As we pad with polynomially many new triples, the running time remains polynomial. Hence the overall construction of padding and then performing the construction of Lemma~\ref{lemma:canonicalpacking:npcomplete} gives a polynomial-time transformation from \XTC to the restricted form of \PathPackingCanonical, which concludes the proof.
\end{proof}


\restatekarptwostarmanyps*

\begin{proof}
Using the same polynomial equivalence relation as in Theorem~\ref{theorem:karp:lowerbound:onesubstar:manytriangles} and by discarding the same trivial cases, it suffices to OR-cross-compose instances~$(n, S_1), \ldots, (n, S_r)$ of \PathPackingCanonical with~$n \geq 10$ whose sets~$S_i$ all have the same size~$m \geq 10$. By Lemma~\ref{lemma:canonicalpacking:larges:npcomplete}, we may assume that~$m = |S_i| > \Delta (\G^{P_3}_n) + 1$ for all~$i \in [r]$. Let~$\G := \G^{P_3}_n$ be the canonical graph for these instances and define~$n' := |V(\G)|$. The construction of~$G$ is a bit more involved because we have to give a construction that embeds two subdivided stars, rather than one. We proceed as follows. Label the vertices of~$\G$ as~$v_1, \ldots, v_{n'}$ in an arbitrary way.

\begin{itemize}
	\item Initialize~$G$ as a copy of the graph~$\G$.
	\item Add a matching~$M$ of size~$n'$ to the graph. For each index~$j \in [n']$ the $j$-th edge of~$M$ has the \emph{blocker vertex~$b_j$} as one endpoint and the \emph{propagator vertex~$p_j$} as the other endpoint. For each~$j \in [n']$, add the edge~$\{b_j, v_j\}$ to~$G$.
	\item Add another matching~$M'$ of size~$n'$ to the graph. For each index~$j \in [n']$ the $j$-th edge of~$M'$ has the \emph{dummy vertex~$d_j$} as one endpoint and the \emph{communicator vertex~$c_j$} as the other endpoint. For each~$j \in [n']$, add the edge~$\{c_j, p_j\}$ to~$G$.
	\item Add a special \emph{carving vertex}~$z^*$ and make it adjacent to the endpoints of matching~$M$.
	\item Add the \emph{carving matching}~$M''$ of size~$(n')^2$ to the graph and make one endpoint of each edge adjacent to~$z^*$. The carving matching ensures that there is a large subdivided star centered at the carving vertex, using~$M''$ and some other neighbors for $z^*$ for the legs of the star. This will force models of~$H$ to center one subdivided star at~$z^*$, as the constructed graph~$G$ will not contain other possibilities for the centers of such large subdivided stars.
	\item For each~$i \in [r]$, add an \emph{instance selector} vertex~$u_i$ to~$G$. For each vertex~$v_j \in S_i$ add the edge~$\{u_i, c_j\}$ to~$G$.
\end{itemize}

This concludes the description of~$G$. The graph~$H$ is defined as follows. It consists of~$m / 3$ copies of the graph~$P_3$, one subdivided star~$H_1$ with~$(2n' - m) + (n')^2$ leaves, and one subdivided star~$H_2$ with~$m$ leaves. Let~$y_1 \in V(H_1)$ be the center of~$H_1$ and let~$y_2 \in V(H_2)$ be the center of~$H_2$. From this choice it is clear that the parameter~$k$ of the constructed \kFSubgraphTest problem, which equals~$|V(H)|$, is polynomially bounded in the size of the largest input instance. It is easy to see that the construction can be performed in polynomial time. 

\begin{observation} \label{observation:twostars:lowerbound:firstcenter}
The carving vertex~$z^*$ is the unique vertex of~$G$ for which both the first neighborhood~$N_G(z^*)$ and the second neighborhood~$N_G(N_G(z^*))$ have size at least~$(2n' - m) + (n')^2$: the vertices of~$V(\G)$ have degree in~$G$ at most~$\Delta(\G) + 1 \leq n' + 1 < (2n' - m) + (n')^2$, the vertices in~$V(M'')$, and the blocker, dummy, and propagator vertices have degree at most four, the instance selector vertices have degree~$m \leq n'$, and a communicator vertex~$c_j$ has second neighborhood~$N_G(N_G(c_j)) \subseteq \{b_j, d_j, p_j, z^*\} \cup \{c_\ell \mid \ell \in [n']\}$ of size at most~$4 + n' < (2n' - m) + (n')^2$.
\end{observation}

\begin{observation} \label{observation:twostars:lowerbound:secondcenter}
As~$m > \Delta(\G) + 1 \geq \deg_G(v_j)$ for every $v_j \in V(\G)$, the only vertices in~$G$ that have degree at least~$m$ are the communicator vertices, the instance selector vertices, and~$z^*$.
\end{observation}

\begin{observation} \label{observation:twostars:wherearepaths}
The only~$P_3$ subgraphs in~$(G - \{z^*\}) - \{b_j, c_j \mid j \in [n']\}$ are contained in~$G[V(\G)]$, as removing these vertices turns~$M''$ into an isolated matching while the dummy vertices, instance selector vertices, and propagator vertices become isolated.
\end{observation}

To establish the correctness of the OR-cross-composition we have to prove that~$H \subseteq G$ if and only if there is a \yes-instance among the inputs. To illustrate how the construction is intended to work, we first prove the reverse direction.

\begin{claim}\label{claim:crosscomp:twostars:backwards}
If~$i \in [r]$ such that~$\G[S_i]$ can be partitioned into~$P_3$'s, then~$H$ is a subgraph of~$G$.
\end{claim}
\begin{claimproof}
Assume that~$\G[S_i]$ can be partitioned into~$P_3$'s. We construct a full $H$-subgraph model~$\phi$ in~$G$. As~$|S_i| = m$, a partition of~$\G[S_i]$ into~$P_3$'s is a packing of~$m/3$ vertex-disjoint~$P_3$'s. As~$\G$ is a subgraph of~$G$, there exists a packing of~$m/3$ vertex-disjoint~$P_3$'s in~$G[S_i]$. We show how to realize models of~$H_1$ and~$H_2$ that are disjoint from each other and from~$S_i$.

The model of~$H_2$ is centered at~$u_i$. Recall that~$N_G(u_i) = \{c_j \mid v_j \in S_i\}$. These~$|S_i| = m$ vertices model the subdivider vertices of the subdivided star~$H_2$. We use the dummy vertices~$\{d_j \mid v_j \in S_i\}$ as the degree-one endpoints of the star~$H_2$. We obtain a realization of~$H_2$ using~$\{u_i\} \cup \{d_j, c_j \mid v_j \in S_i\}$.

The model of~$H_1$ is centered at the carving vertex~$z^*$. Each of the~$(n')^2$ edges in the carving matching~$M''$ is used to realize one leg of the subdivided star. The remaining~$2n' - m$ legs are realized as follows. For each index~$j \in [n']$ such that~$v_j \in S_i$, we realize one leg of the subdivided star through the $j$-th edge~$\{p_j, b_j\}$ in the matching~$M$. For each~$j \in [n']$ such that~$v_j \not \in S_i$, we realize \emph{two} legs of the subdivided star: one leg~$\{p_j, c_j\}$ and one leg~$\{b_j, v_j\}$. From these definitions it is clear that the models of~$H_2$ and~$H_1$ are disjoint from each other and from~$S_i$.  As we realize one leg of the subdivided star~$H_2$ for each of the~$m$ indices~$j$ with~$v_j \in S_i$, while we realize two legs for the~$n' - m$ indices~$j$ with~$v_j \not \in S_i$, together these realize all~$(n')^2 + m + 2(n' - m) = (2n' - m) + (n')^2$ legs of the subdivided star~$H_2$. Hence~$G$ contains~$H$ as a subgraph.
\end{claimproof}

The following series of claims will establish the reverse of Claim~\ref{claim:crosscomp:twostars:backwards}: if~$H \subseteq G$ then the answer to some input instance is \yes.

\begin{claim} \label{claim:crosscomp:twostars:carvingcenter}
If~$\phi$ is a full subgraph model of~$H$ in~$G$, then~$\phi(y_1) = z^*$.
\end{claim}
\begin{claimproof}
If~$\phi$ is a full subgraph model of~$H$ in~$G$, then the degree of~$\phi(y_1)$ in~$G$ has to be at least~$\deg_H(y_1) = (2n' - m) + (n')^2$. By Observation~\ref{observation:twostars:lowerbound:firstcenter}, vertex~$z^*$ is the only possible choice.
\end{claimproof}

\begin{claim}\label{claim:crosscomp:twostars:manycommunicators}
Let~$\phi$ be a full subgraph model of~$H$ in~$G$ and let~$C := \{c_j \mid j \in[n'] \wedge c_j \in \phi(V(H_1))\}$. The following hold.
\begin{enumerate}
	\item $|C| \geq n' - m$.
	\item If~$|C| = n' - m$, then (a)~$\phi(V(H_1))$ contains all vertices~$\{v_j \mid c_j \in C\}$ and (b)~$\phi(V(H_1))$ contains all blocker vertices.
\end{enumerate}
\end{claim}
\begin{claimproof}
Let~$\phi$ be a full subgraph model of~$H$ in~$G$, which implies by the previous claim that~$\phi(y_1) = z^*$. Matching~$M''$ can realize~$(n')^2$ of the legs of~$H_1$ and, as~$\phi(y_1) = z^*$, it is easy to verify they cannot be useful in any other role. Consider how the remaining~$2n' - m$ legs of~$H_1$ are realized and observe that~$N_G(z^*) \setminus V(M'') = V(M)$. A priori, each vertex of~$N_G(z^*) \setminus V(M'')$ can be used to model a subdivider vertex of a leg of the star. However, for each~$j \in [n']$, if~$c_j \not \in \phi(V(H_1))$, then vertices~$\{p_j, b_j\}$ cannot both model subdivider vertex of different legs of the star: as~$N_G(p_j) = \{z^*, c_j, b_j\}$, the only way to utilize~$p_j$ to realize a leg without~$c_j$ is by using~$\{p_j, b_j\}$ as one leg, but then only one of the vertices of~$\{p_j, b_j\}$ models a subdivider vertex. It follows that if~$|C| <  n' - m$, then there are more than~$n' - m$ edges~$\{p_j, b_j\}$ of~$M$ through which only one leg of the subdivided star is realized. Through the remaining edges of~$M$, at most two legs can be realized. Hence the subdivided star can realize less than~$(n')^2 + m + 2(n' - m)$, and consequently~$\phi$ does not model the subdivided star with~$(2n' - m) + (n')^2$ leaves. We conclude that~$|C| \geq n' - m$, proving the first part of the claim.

For the proof of the second part, assume that~$|C| = n' - m$. Besides the~$(n')^2$ legs of~$H_1$ realized in~$M''$, there are~$(2n' - m) - |M'| = n' - m$ edges of~$M$ for which both endpoints of the edge realize subdivider vertices of different legs of the star. By the argument above, the vertices of an edge~$\{p_j, b_j\}$ of~$M$ can only model subdivider vertices of two different legs of the star if these legs are realized as~$\{p_j, c_j\}$ and~$\{b_j, v_j\}$. As each of these contributes one vertex to~$C$, while~$|C| = n' - m$, it follows that whenever~$c_j \in C$ the model must realize two legs of the star through~$\{p_j, b_j\}$ and therefore use~$\{b_j, v_j\}$ to model one leg of the star, forcing~$v_j \in \phi(V(H_1))$. This proves (a) of the second part of the claim. For~(b), observe that if~$b_j \not \in \phi(V(H_1))$, then either~$c_j \in \phi(V(H_1))$ while only one leg of the star is realized for the $j$-th edge of~$M$, or~$c_j \not \in \phi(V(H_1))$ and no legs of the star are realized for the $j$-th edge of~$M$. It follows that not all legs can be realized, a contradiction.
\end{claimproof}

\begin{claim} \label{claim:crosscomp:twostars:selectsinstance}
If~$\phi$ is a full subgraph model of~$H$ in~$G$, then~$\phi(y_2) \in \{u_1, \ldots, u_r\}$.
\end{claim}
\begin{claimproof}
Let~$\phi$ be a full subgraph model of~$H$ in~$G$. By Claim~\ref{claim:crosscomp:twostars:carvingcenter} we have~$\phi(y_1) = z^*$. By Observation~\ref{observation:twostars:lowerbound:secondcenter} and the fact that~$\deg_H(y_2) = m$, it follows that~$\phi(y_2)$ can only be an instance selector vertex or a communicator vertex. It remains to prove that the latter cannot happen.

Suppose that~$\phi(y_2) = c_j$ for some~$j \in [n']$. Consider how the legs of the star~$H_2$ are realized in~$H$. Observe that~$N_G(c_j) \subseteq \{d_j, p_j\} \cup \{u_1, \ldots, u_r\}$. As~$\deg_G(d_j) = 1$ vertex~$d_j$ cannot form the subdivider vertex of a leg of the star. Hence~$p_j$ can be one subdivider vertex of the star~$H_2$ and the remaining subdivider vertices of the star~$H_2$ must be instance selectors. For each instance selector~$u_j$ that models a subdivider vertex, a neighbor~$N_G(u_j) \subseteq \{c_\ell \mid \ell \in [n']\}$ is used to model the degree-one endpoint of the corresponding leg of the star. We distinguish two cases.

\begin{itemize}
	\item If~$p_j$ is not used to model a subdivider vertex of~$H_2$, then~$m$ instance selectors are used for this role, which leads to~$m$ communicator vertices other than~$\phi(y_2)$ being used in~$\phi(V(H_2))$. Including~$\phi(y_2)$ itself,~$\phi(V(H_2))$ contains~$m + 1$ communicator vertices. But then~$\phi(V(H_1))$, which is disjoint from~$\phi(V(H_2))$, realizes~$H_1$ using strictly less than~$n' - m - 1$ communicator vertices, a contradiction to Claim~\ref{claim:crosscomp:twostars:manycommunicators}.
	\item If~$p_j$ models a subdivider vertex of~$H_2$, then~$b_j$ is used for the degree-one endpoint of that leg of the star (as~$N_G(p_j) = \{z^*, c_j, b_j\}$ and~$z^* = \phi(y_1)$). Hence~$\phi(V(H_2))$ contains~$\{p_j, b_j\}$, the endpoints of the $j$-th edge of~$M$. The remaining~$m-1$ legs of the star are realized through instance selectors and their communicator neighbors, leading to an additional~$m-1$ communicator vertices being used in~$\phi(V(H_2))$. Hence~$\phi(V(H_2))$ contains~$m$ communicator vertices and the endpoints of one edge in~$M$. But then the model of~$H_1$, which is disjoint from~$\phi(V(H_2))$, contains at most~$n' - m$ communicator vertices and is disjoint from~$\{b_j, p_j\}$ for at least one~$j \in [n']$, contradicting Claim~\ref{claim:crosscomp:twostars:manycommunicators}.
\end{itemize}
This proves Claim~\ref{claim:crosscomp:twostars:selectsinstance}.
\end{claimproof}

\begin{claim}
If~$\phi$ is a full subgraph model of~$H$ in~$G$ with~$\phi(y_1) = z^*$ and~$\phi(y_2) = u_i$ for some~$i \in [r]$, then~$\G[S_i]$ can be partitioned into~$P_3$'s.
\end{claim}
\begin{claimproof}
Assume that the stated conditions hold. As~$N_G(u_i) = C_j := \{c_j \mid v_j \in S_i\}$ and the degree of~$u_i$ in~$G$ matches the degree of~$y_2$ in~$H_2$, it follows that all~$m$ vertices of~$C_j$ are used in~$\phi(V(H_2))$. Consider~$\overline{C_j} := \{c_j \mid j \in [n']\} \setminus C_j$. Since no vertex of~$C_j$ can be used in~$\phi(V(H_1))$, the model of~$H_1$ uses at most~$n' - m$ communicator vertices. By Claim~\ref{claim:crosscomp:twostars:manycommunicators} it follows that~$\phi(V(H_1))$ contains exactly~$n' - m$ communicator vertices, which must be those in~$\overline{C_j}$ since they are the only ones left. By the second part of Claim~\ref{claim:crosscomp:twostars:manycommunicators}, we find that~$\phi(V(H_1))$ contains~$\{v_j \mid c_j \in \overline{C_j}\}$ and all blocker vertices. Hence none of these vertices can be used to model~$P_3$'s. Similarly, no vertices of~$C_j \cup \overline{C_j}$ can be used to model a~$P_3$, nor can~$z^*$ be used since it is used as the center of~$H_1$. It follows that no vertices of~$\{z^*\} \cup \{c_j, b_j \mid j \in [n']\}$ are used to model~$P_3$'s, which implies by Observation~\ref{observation:twostars:wherearepaths} that all~$P_3$ models are contained in~$G[V(\G)]$. As they are disjoint from~$\phi(V(H_1))$, it follows that all~$P_3$ models are contained in~$G[V(\G)] - \{v_j \mid c_j \in \overline{C_j}\} = G[V(\G)] - \{v_j \mid v_j \not \in N_G(u_i)\} = G[V(\G)] - \{v_j \mid v_j \not \in S_i\} = G[S_i] = \G[S_i]$. Hence the~$m/3$ vertex-disjoint~$P_3$ models are contained in~$\G[S_i]$. As~$|S_i| = m$ this implies that~$\G[S_i]$ can be partitioned into~$P_3$'s.
\end{claimproof}

As the claims establish that instance~$(G,H)$ acts as the logical OR of the input instances, this concludes the proof of Theorem~\ref{theorem:karp:lowerbound:twosubstar:manyps}.
\end{proof}

\subsection{Upper bounds} \label{subsection:subgraph:karp:upperbounds}

We now show that the kernelization lower bounds of Section~\ref{subsection:subgraph:karp:lowerbounds} are fragile in the sense that, if we slightly change the considered graph classes~$\F$, then the resulting \kFSubgraphTest problem admits polynomial (many-one) kernels. Both results we present in this section rely on the same underlying idea. The graph classes for which we obtained kernel lower bounds in Section~\ref{subsection:subgraph:karp:lowerbounds} are $(3,0,2,2)$-splittable. Theorem~\ref{theorem:kernel:subgraph} therefore gives Turing kernels for the \kFSubgraphTest problem on these graph families, which are based on guessing the model of a vertex set that realizes the split. In the many-one kernelization setting, this strategy fails since we have to produce a single, small output instance and therefore cannot cover the~$|V(G)|^c$ different options for the model of a vertex set that realizes the split. In the two cases highlighted below, the split is realized by just a single vertex and the mentioned problem can be circumvented: by ad-hoc arguments we can compute a representative vertex set~$Y$ of size~$k^{\Oh(1)}$ such that, if~$G$ contains~$H$ as a subgraph, then~$G$ has an $H$-subgraph model where the split vertex is realized by a member of~$Y$. Since the size of~$Y$ is polynomially bounded in the parameter, this allows us to compute representative sets for all relevant models of the split vertex. For each choice, Lemma~\ref{lemma:kernel:generic} gives a representative set of size polynomial in~$k$. The union of the representative sets over all~$y \in Y$ then gives us a kernel.

\subsubsection{One subdivided star and many \texorpdfstring{$P_3$}{P3}'s} \label{section:subdividedstar:manypthrees}
Throughout Section~\ref{section:subdividedstar:manypthrees}, we will use the term subdivided star for any graph that can be obtained from a star by subdividing each edge \emph{at most} once. Observe that under this definition, any path on at most three vertices is a subdivided star. If~$H$ is a subdivided star and~$v \in V(H)$ is the unique maximum-degree vertex in~$H$, then we will call~$v$ the \emph{center} of the subdivided star and we say that~$H$ has a center. A subgraph model of a subdivided star~$H$ in~$G$ is \emph{centered at~$v \in V(G)$} if~$v$ is the image of the center. In the degenerate case that~$H$ is a path on at most five vertices (which can be obtained by subdividing all edges of~$K_{1,2}$), all connected components of the pattern graph have constant size which allows us to obtain a polynomial kernel through Lemma~\ref{lemma:kernel:generic}. In the intermediate lemmas leading up to the kernelization, we therefore restrict ourselves to the cases that there is a unique center.

We will need the following proposition for solving constrained weighted matching in polynomial time, which is possible by reducing it to an unconstrained minimum weight perfect matching computation.

\begin{proposition}[{\cite{Plesnik99}}] \label{proposition:constrained:weighted:matching}
There is a polynomial-time algorithm that, given an integer~$k$ and a graph~$G$ with nonnegative integer edge weights at most~$w_0$, computes in time polynomial in~$|V(G)| + |E(G)| + w_0$ a minimum-weight matching of cardinality~$k$, or determines that no matching of cardinality~$k$ exists.
\end{proposition}

Proposition~\ref{proposition:constrained:weighted:matching} allows us to find a subgraph model of a subdivided star in polynomial time, if one exists. While we could also perform this task using the randomized algorithm of Theorem~\ref{theorem:fsubgraphtest:alg}, the algorithm we present next has the advantage of being deterministic.

\begin{lemma} \label{lemma:subdividedstartest:alg}
There is a polynomial-time algorithm that, given a graph~$G$ and a subdivided star~$H$ with center~$c$ and a vertex~$v \in V(G)$, outputs a full $H$-subgraph model centered at~$v$ in~$G$ if one exists.
\end{lemma}
\begin{proof}
Let~$k_1$ be the number of one-vertex components in~$H - \{c\}$ and let~$k_2$ be the number of two-vertex components in~$H - \{c\}$. Let~$G_c$ be the subgraph of~$G$ containing only the edges that have at least one endpoint in~$N_G(v)$. Define the weight of an edge in~$G_c$ as the number of endpoints the edge has in~$N_G(v)$; then every weight will be one or two. We invoke Proposition~\ref{proposition:constrained:weighted:matching} to compute a minimum-weight matching~$M_c$ in~$G_c$ of cardinality exactly~$k_2$. Let~$w_c$ be the weight of~$M_c$.

\begin{claim}
There is a full $H$-subgraph model~$\phi$ with~$\phi(c) = v$ if and only if the matching~$M_c$ exists and~$\deg_G(v) - w_c \geq k_1$.
\end{claim}
\begin{claimproof}
For the forward direction, suppose that a matching of cardinality~$k_2$ exists in~$G_c$ and that~$M_c$ is a minimum-weight matching of this size of weight~$w_c$ satisfying~$\deg_G(v) - w_c \geq k_1$. We can realize a model of a subdivided star centered at~$v$ as follows: we use the edges in~$M_c$ for the two-vertex components of~$H - \{c\}$ (the legs of the star), which is possible since each edge has at least one endpoint that is adjacent to~$v$. By our choice of weight function we know that~$|N_G(v) \cap V(M_c)| = w_c$. By the degree requirement, we therefore find that~$N_G(v) \setminus V(M_c)$ consists of at least~$k_1$ vertices, which we can use to realize the size-one components of~$H - \{c\}$. Hence we obtain a full model of~$H$ in~$G$.

For the reverse direction, suppose that~$\phi$ is a full $H$-subgraph model in~$G$ centered at~$v$. Since all the two-vertex components of~$H - \{v\}$ have at least one endpoint adjacent to~$v$, as this edge of~$H$ has to be realized, the model of the two-vertex components is a matching~$M$ in~$G_c$ of cardinality~$k_2$. The models of the one-vertex components consist of vertices of~$N_G(v) \setminus V(M)$, hence~$|N_G(v) \setminus V(M)| \geq k_1$. As~$M$ is a matching of cardinality~$k$ in~$G_c$, while~$M_c$ is a matching of cardinality~$k$ in~$G_c$ that minimizes the number of vertices of~$N_G(v)$ it uses (by our choice of weight function), it follows that~$|N_G(v) \setminus V(M)| \geq \deg_G(v) - w_c \geq k_1$, which concludes the reverse direction of the proof.
\end{claimproof}

The claim shows how to extract a model centered at~$v$ from the matching~$M_c$, if it exists. As~$M_c$ can be computed in polynomial time, the claim follows.
\end{proof}

\begin{lemma} \label{lemma:starsandpaths:findrelevantcenters}
There is a polynomial-time algorithm with the following specifications. The input consists of a graph~$G$, a subdivided star~$H'$ with center~$c$, and a graph~$H''$ where each connected component is a path on at most three vertices. Let~$H := H' + H''$ and~$k := |V(H)|$. The output is a set~$Y \subseteq V(G)$ of size~$\Oh(k^2)$ such that if~$H \subseteq G$, then there is a full $H$-subgraph model in~$G$ in which the model of~$H'$ is centered at a vertex in~$Y$.
\end{lemma}
\begin{proof}
On input~$(G,H',H'')$, the algorithm proceeds as follows. We first invoke Lemma~\ref{lemma:subdividedstartest:alg} to~$G$ and~$H'$ for all possible choices~$v \in V(G)$ for the center to determine whether~$G$ contains~$H'$ as a subgraph. If not, then there is no $H$-subgraph in~$G$ and we may output~$Y := \emptyset$. In the remainder we assume that~$H' \subseteq G$ and that~$\phi_{H'}$ is a full $H'$-subgraph model in~$G$. Let~$S$ be the set of vertices that have degree at least~$3k+1$ in~$G$.

\begin{claim}
If~$|S| \geq k$, then there is a full $H$-subgraph model in~$G$ where the model of~$H'$ is centered at~$\phi_{H'}(c)$.
\end{claim}
\begin{claimproof}
Assume that~$|S| \geq k$. We construct an $H$-subgraph model~$\phi$ in~$G$, as follows. Define~$\phi(v) = \phi_{H'}(v)$ for all~$v \in V(H')$, which ensures that the center of the subdivided star is at~$\phi_{H'}(v)$. While there is a connected component~$C$ of~$H''$ that has not yet been assigned an image under~$\phi$, do the following. Let~$Z$ be the vertices in the current image of the partial $H$-subgraph model~$\phi$; then~$|Z| < k$ since the model is not yet complete. Hence there is at least one high-degree vertex~$s \in S \setminus Z$. As~$\deg_G(s) \geq 3k + 1$ and~$|Z| < k$, there are at least two vertices~$\{s_1, s_2\}$ in~$N_G(s) \setminus Z$. These three vertices form a path on three vertices in~$G$. As~$C$ is a path on at most three vertices, we can specify a model for~$C$ in~$\phi$ by mapping~$C$ to (a subset of)~$\{s, s_1, s_2\}$. By iterating this process we augment~$\phi$ to a full $H$-subgraph model in~$G$ with~$\phi(c) = \phi_{H'}(c) = v$, which proves the claim.
\end{claimproof}

The claim shows that if~$|S| \geq k$, then~$Y = \{\phi_{H'}(c)\}$ is a valid output for the procedure. In the remainder we therefore assume that~$|S| < k$. Let~$\P$ be a maximal packing of vertex-disjoint copies of~$P_3$ in the graph~$G-S$, which can be found by a greedy polynomial-time algorithm.

\begin{claim}
If~$\P$ contains at least~$k$ vertex-disjoint copies of~$P_3$, then there is a full $H$-subgraph model in~$G$ where the model of~$H'$ is centered at~$\phi_{H'}(c)$.
\end{claim}
\begin{claimproof}
Assuming~$|\P| \geq k$ we show how to construct an $H$-subgraph model~$\phi$ in~$G$ with~$\phi(c) = \phi_{H'}(c)$. Define~$\phi(v) = \phi_{H'}(v)$ for all~$v \in V(H')$. While there is a connected component~$C$ of~$H''$ that has not yet been assigned an image under~$\phi$, let~$Z$ be the vertices currently used in the image of~$\phi$ and observe that~$|Z| < k$. Let~$P \in \P$ be a $P_3$-subgraph disjoint from~$Z$, which exists as~$|\P| \geq k$. As~$C$ is a subgraph of~$P$, we can map component~$C$ under~$\phi$ to (a subgraph of)~$P$. Iterating the procedure results in a full $H$-subgraph model as required.
\end{claimproof}

The claim shows that if~$|\P| \geq k$ then~$Y = \{\phi_{H'}(c)\}$ is a valid output for the procedure. In the remainder we assume that~$|\P| < k$. Let~$T$ be the vertices that occur in a~$P_3$ subgraph in~$\P$; then the previous assumption implies that~$|T| < 3k$. By our choice of~$\P$ as a maximal packing of~$P_3$'s, every connected component of~$G - (S \cup T)$ has at most two vertices. Let~$C_T$ be the connected components of~$G - (S \cup T)$ that are adjacent to a vertex in~$T$. As each vertex of~$T \subseteq V(G - S)$ has degree at most~$3k$, each vertex of~$T$ is adjacent to at most~$3k$ components in~$C_T$ and therefore~$|C_T| \leq 3k \cdot |T| \leq 9k^2$. Let~$T'$ be the union of~$T$ and the vertices in~$C_T$. As each component of~$C_T$ has at most two vertices,~$|T'| \leq |T| + 2 \cdot |C_T| \leq 3k + 18k^2$. By our choice of~$C_T$, the only vertices of~$G$ that a connected component of~$G - (S \cup T')$ can be adjacent to, are those in~$S$. Let~$X := V(G) \setminus (S \cup T')$.

For each~$v \in X$ we invoke Lemma~\ref{lemma:subdividedstartest:alg} to determine whether there is an $H'$-subgraph model in~$G$ centered at~$v$. If no such subgraph model exists, then any $H$-subgraph in~$G$ must center the star at a vertex of~$V(G) \setminus X = S \cup T'$, so we may safely output~$Y := S \cup T'$ of size at most~$4k + 18k^2$. In the remainder of the proof we deal with the case that an $H'$-subgraph model can be centered in a vertex of~$X$. As the size of~$X$ is not bounded in~$k$, we cannot include all these vertices in the output set~$Y$. To deal with this issue, we will identify a single vertex~$x^* \in X$ such that if there is an $H$-subgraph model in~$G$ that centers the subdivided star at a member of~$X$, then there is such a model that centers the star at~$x^*$. The output set~$Y$ will then consist of~$S \cup T' \cup \{x^*\}$ of size at most~$4k + 18k^2 + 1$. It remains to define~$x^*$ and show that it has the claimed properties. We choose~$x^*$ as follows.
\begin{itemize}
	\item If~$k_1 = 0$ then we let~$x^*$ be an arbitrary vertex of~$X$ such that~$G$ has an~$H'$-subgraph model centered there, which can be tested using Lemma~\ref{lemma:subdividedstartest:alg}. We say that we found an \emph{small} candidate for the center.
	\item If~$k_1 > 0$ then we have to be slightly more careful. Recall that all connected components of~$G[X] = G - (S \cup T')$ have at most two vertices. If there is a connected component~$C = \{u,v\}$ in~$G[X]$ such that~$G$ has an $H'$-subgraph model centered at~$u$, then we set~$x^* := u$ and we say that we found a \emph{large} candidate for the center. If no $H'$-subgraph model in~$G$ can be centered in a two-vertex component of~$G[X]$, then we let~$x^*$ be an arbitrary vertex of~$X$ that can be the center of an $H'$-subgraph model in~$G$ and we say that we found a \emph{small} candidate for the center.
\end{itemize}

Vertex~$x^*$ can be identified in polynomial time. The algorithm outputs the set~$Y := S \cup T' \cup \{x^*\}$. The following series of claims proves that this choice of~$Y$ satisfies the requirements in the statement of the claim. To simplify the notation in the rest of the proof, let~$\delta := 1$ if we found a small candidate for the center and~$0$ otherwise.

\begin{claim} \label{claim:starsandpaths:model:uses:many:highdegree}
Let~$\phi_{H'}$ be a full $H'$-subgraph model in~$G$ centered at a vertex~$x \in X$. Then $|\phi_{H'}(V(H')) \cap S| \geq k_1 + k_2 - \delta$.
\end{claim}
\begin{claimproof}
We consider the three defining cases for~$x^*$ separately.

\begin{itemize}
	\item If~$k_1 = 0$, then all connected components of~$H' - \{c\}$ consist of two vertices. Since~$\phi_{H'}(c)$ is contained in a connected component~$C$ of~$G[X]$, which has at most two vertices, no connected component~$C'$ of~$H' - \{c\}$ can have its image entirely within~$C$ because~$C$ also contains the image of~$c$. Hence the image of every such connected component~$C'$ of~$H' - \{c\}$ contains a vertex in~$N_G(C)$. By our choice of~$T'$, connected components of~$G[X] = G - (S \cup T')$ only have neighbors in~$S$. Hence the image of every connected component~$C'$ of~$H' - \{c\}$ intersects~$S$. As the images are vertex-disjoint and there are~$k_2$ such components, we find~$|\phi_{H'}(V(H')) \cap S| \geq k_2 = k_1 + k_2$ (as~$k_1 = 0$), which proves the claim in this case.
	\item If~$k_1 > 0$, and we found a large candidate for the center, then there is at most one connected component of~$H' - \{c\}$ whose image under~$\phi_{H'}$ does not intersect~$S$: if~$\phi_{H'}(c)$ is contained in a connected component~$C$ of~$G[X]$ of two vertices, then the vertex of~$C$ unequal to~$\phi_{H'}(c)$ can model a one-vertex component of~$H' - \{c\}$. The images of all other connected components intersect~$S$, by the same argument as in the previous case. Hence~$|\phi_{H'}(V(H')) \cap S| \geq k_1 + k_2 - 1$.
	\item If~$k_1 > 0$, but we found a small candidate for the center, then the mechanism for defining~$x^*$ ensures that it is impossible to have a~$H'$-subgraph model centered in a two-vertex component of~$G[X]$ where the non-center vertex of the component is used to model a one-vertex component of~$H' - \{c\}$. Hence all connected components of~$H' - \{c\}$ intersect~$S$, implying that~$|\phi_{H'}(V(H')) \cap S| \geq k_1 + k_2$.
\end{itemize}
This concludes the proof of Claim~\ref{claim:starsandpaths:model:uses:many:highdegree}.
\end{claimproof}

\begin{claim} \label{claim:starsandpaths:centeredmodel:fewhighdegree}
For any set~$Z \subseteq T'$ of size at most~$k$ there is an $H'$-subgraph model~$\phi'$ in~$G - Z$ centered at~$x^*$ such that~$|\phi'(V(H')) \cap S| \leq k_1 + k_2 - \delta$.
\end{claim}
\begin{claimproof}
We first prove the claim when we found a small candidate for the center. Afterwards we show how to adapt the proof in case of a large candidate for the center.

\textbf{Small candidate.} Suppose we found a small candidate for the center. By the definition of~$x^*$, there is an $H'$-subgraph~$\phi'$ in~$G$ centered at~$x^*$. Recall that~$H - \{c\}$ has~$k_1 + k_2$ connected components. Each of the~$k_1$ one-vertex components can only contribute one to~$\phi'(V(H')) \cap S$. This implies that as long as $|\phi'(V(H')) \cap S| > k_1 + k_2$ holds, there is a two-vertex connected component~$C = \{u,v\}$ of~$H' - \{c\}$ with~$\{\phi'(u), \phi'(v)\} \subseteq S$. By the topology of a star, exactly one vertex of~$\{u,v\}$ is adjacent in~$H'$ to~$c$. Assume without loss of generality that this is~$v$, which implies that~$\{u,c\} \not \in E(H')$. Hence we may change the image of~$u$ without violating the validity of the $H'$-subgraph model, as long as the new image is disjoint from the rest of the model and is adjacent in~$G$ to~$\phi'(v)$. Since~$\phi'(v) \in S$ has degree more than~$2k$ in~$G$, there is a neighbor~$u'$ of~$\phi'(v)$ in~$G$ that is not used in the model~$\phi'$. By the preceding argument we may set~$\phi'(u) := u'$ to decrease the number of vertices of~$S$ used by the model by one, without violating the validity of the model. By iterating this argument we may assume that~$|\phi'(V(H')) \cap S| = k_1 + k_2$.

We will now show how to alter~$\phi'$ to turn it into a model in~$G - Z$, i.e., how to avoid using vertices of~$Z$ in the image of~$\phi'$. Assume that~$z \in \phi'(V(H'))$. Since~$Z \subseteq T'$ and~$x^* \in X$ is not adjacent in~$G$ to any vertex in~$T'$, it follows that~$z$ is the image of a degree-1 endpoint~$e$ of the subdivided star that is not adjacent in~$H'$ to~$c$. Let~$d$ be the unique neighbor of~$e$ in~$H'$. Since~$\phi'(d)$ is adjacent to both~$x^* \in X$ and~$z \in T'$, we must have~$\phi'(d) \in S$ by our choice of~$T'$. Hence the degree of~$\phi'(d)$ in~$G$ is at least~$3k+1$. Since~$|Z| \leq k$, the current model uses at most~$k$ vertices, and~$|S| \leq k$, there is a vertex~$d'$ in~$N_G(\phi'(d)) \setminus (Z \cup \phi'(V(H')) \cup S)$. Consequently, we may update the model~$\phi'$ by setting~$\phi'(d) := d'$ while preserving a valid~$H'$-subgraph model without increasing the number of vertices of~$S$ used by the model. By iterating this procedure we arrive at a subgraph model of~$H'$ disjoint from~$Z$ (which is therefore a model in~$G - Z$) using exactly~$k_1 + k_2$ vertices in~$S$.

\textbf{Large candidate.} The case where we found a large candidate for the center is similar; the main difference is that an $H'$-model can use one vertex of~$S$ less, by using the neighbor of~$x^*$ in~$G[X]$ as the image for a one-vertex component of~$H' - \{c\}$. In the case that we found a large candidate, the mechanism for defining~$x^*$ ensures that there is a~$H'$-subgraph~$\phi'$ in~$G$ centered at~$x^*$ such that~$x^*$ is contained in a two-vertex connected component~$C_{x^*} = \{x^*, y^*\}$ of~$G[X]$. We first show how to obtain a model~$\phi'$ where~$y^*$ is the image of a one-vertex component of~$H - \{c\}$. If~$y^* \not \in \phi'(V(H'))$, then we may take any one-vertex component of~$H - \{c\}$ (which exists since~$k_1 > 0$ if we found a large candidate) and map it to~$y^*$. Assume then that~$y^*$ is used as part of the image of a two-vertex component~$C_2$ of~$H - \{c\}$, and let~$C_1 = \{w\}$ be a one-vertex component of~$H - \{c\}$. Then~$w' := \phi'(w) \in S$ since~$N_G(x^*) \subseteq \{y^*\} \cup S$ while~$y^*$ is used for a different role in the model. Now we can swap the images of~$C_1$ and~$C_2$ under~$\phi$, as follows. We let~$\phi'(w) := y^*$. We map the vertex of~$C_2 \cap N_{H'}(c)$ to the vertex~$w' \in S$. Since~$w'$ has degree at least~$3k + 1$ in~$G$, it has a neighbor that is not yet used in the model; we use such a neighbor as the image of~$C_2 \setminus N_{H'}(c)$. 

We obtain a valid $H'$-subgraph model in~$G$ that is centered in~$x^*$ and in which the neighbor of~$x^*$ in~$G[X]$ is the image of a one-vertex component of~$H' - \{c\}$. While there is a connected component of~$H' - \{c\}$ whose image under~$\phi'$ contains two vertices in~$S$, we can update the model to reduce this number to one, just as in the previous case. Iterating this argument we obtain a model of~$H'$ centered at~$x^*$ where one component of~$H' - \{c\}$ does not use any vertex of~$S$, while the remaining components each use at most one vertex of~$S$. Hence this model~$\phi'$ satisfies~$|\phi'(H') \cap S| \leq k_1 + k_2 - 1$. Finally, just as in the previous case we may eliminate the vertices of~$Z$ from the model without increasing the number of~$S$-vertices that are used. This concludes the proof of Claim~\ref{claim:starsandpaths:centeredmodel:fewhighdegree}.
\end{claimproof}

\begin{claim} \label{claim:starsandpaths:centeratxprime}
If there is a full $H$-subgraph model~$\phi$ in~$G$ in which the subdivided star~$H'$ is centered at a vertex~$x \in X$, then there is a full $H$-subgraph model in~$G$ in which the subdivided star~$H'$ is centered at~$x^*$.
\end{claim}
\begin{claimproof}
We build a model~$\phi'$ of~$H$ in~$G$ that centers the subdivided star at~$x^*$. Let~$S' := \phi(V(H)) \cap S$. By Claim~\ref{claim:starsandpaths:model:uses:many:highdegree} we know that the~$H'$ submodel of~$\phi$ contributes at least~$k_1 + k_2 - \delta$ vertices to~$S'$. Hence there are at most~$|S'| - (k_1 + k_2 - \delta)$ vertices in~$S'$ that are the image of a vertex of a path in~$H$. Define~$\phi'$ as follows.

\begin{itemize}
	\item For every path~$P$ in~$H''$ such that~$\phi(V(P)) \cap S = \emptyset$, set~$\phi'(v) = \phi(v)$ for all~$v \in V(P)$. Let~$Z \subseteq V(G)$ be the vertices used in the image of the partial model constructed in this way.
	\item By Claim~\ref{claim:starsandpaths:centeredmodel:fewhighdegree}, there is a $H'$-subgraph model~$\phi_{H'}$ of~$H'$ in~$G - Z$ that is centered at~$x^*$ and uses at most~$k_1 + k_2 - \delta$ vertices from~$S$. Set~$\phi'(v) = \phi_{H'}(v)$ for all~$v \in V(H')$; as~$\phi_{H'}$ is a model of~$G - Z$ no vertex is used twice in the partial model constructed so far. Let~$Z'$ be the vertices used by the model after this step.
	\item It remains to define an image for the connected components of~$H''$ whose image under~$\phi$ contains a vertex of~$S$. There are at most~$|S'| - (k_1 + k_2 - \delta) \leq |S| - (k_1 + k_2 - \delta)$ of such components by the observation above. On the other hand, there are at least~$|S| - (k_1 + k_2 - \delta)$ vertices in~$S$ that are not used by the partial model~$\phi'$ constructed so far. As each vertex of~$S$ has degree at least~$3k+1$ and the partial model contains at most~$k$ vertices, we can choose for each~$s \in S \setminus Z$ two vertices~$s_1, s_2 \in N_G(s) \setminus Z$ such that the assigned pairs are disjoint over all~$s \in S \setminus Z$. Each resulting triple is a path on three vertices in~$G$ that can form the image of one of the~$|S| - (k_1 + k_2 - \delta)$ remaining connected components of~$H''$. Hence there are sufficient triples to realize all remaining components of~$H''$. As~$H = H' + H''$ this gives a $H$-subgraph model in~$G$ centered at~$x^*$, as required.
\end{itemize}
\end{claimproof}

As Claim~\ref{claim:starsandpaths:centeratxprime} shows that any model of~$H$ centered at a vertex of~$X$ can be transformed into a model centered at~$x^*$, the set~$Y := V(G) \setminus (X \setminus \{x^*\}) = S \cup T' \cup \{x^*\}$ satisfies the claimed requirements and may be used as the output. This concludes the proof of Lemma~\ref{lemma:starsandpaths:findrelevantcenters}.
\end{proof}

Using Lemma~\ref{lemma:starsandpaths:findrelevantcenters} we can prove the following theorem.

\restatekarpstarspaths*

\begin{proof}
On input~$(G,H)$, the kernelization algorithm proceeds as follows. If~$H$ is not of the correct form, which is easily determined in polynomial time, then we output a constant-size \no-instance. If~$H'$ does not have a center then the subdivided star~$H'$ has at most five vertices, implying that~$H$ is $5$-small. Hence we obtain a polynomial kernel using Theorem~\ref{theorem:kernel:packing}.


If~$H'$ has a unique center~$c$, then observe that all components of~$H'' := H - V(H')$ are paths on at most three vertices. We invoke Lemma~\ref{lemma:starsandpaths:findrelevantcenters} to compute a set~$Y$ of size~$\Oh(k^2)$ such that, if~$G$ contains~$H$ as a subgraph, then there is a subgraph model of~$H$ in~$G$ where~$c$ is mapped to a member of~$Y$. For each~$y \in Y$ construct a partial $H$-subgraph model~$\phi^{c \mapsto y}$ by setting~$\phi^{c \mapsto y}(c) := y$. Observe that every connected component of~$H - \{c\}$ has at most three vertices. Letting~$D := \{c\}$ the tuple~$(G,H,\phi^{c \mapsto y}, D)$ therefore satisfies the requirements of Lemma~\ref{lemma:kernel:generic} with~$(a,b,d) = (3, 0, 0)$. We may therefore invoke the lemma to compute a set~$X^{c \mapsto y}$ of size~$\Oh(k^{\Oh(a + b^2 + d)}) = \Oh(k^{\Oh(1)})$ such that if~$G$ has a full $H$-subgraph model extending~$\phi^{c \mapsto y}$, then~$G[X^{c\mapsto y}]$ contains a full $H$-subgraph model that extends~$\phi^{c \mapsto y}$.

Let~$X$ be the union of~$X^{c \mapsto y}$ for all~$v \in Y$. As~$|Y| \in \Oh(k^2)$, the size of~$X$ is polynomial in~$k$. Lemma~\ref{lemma:starsandpaths:findrelevantcenters} guarantees that if~$G$ contains~$H$ as a subgraph, then there is a $H$-subgraph model in~$G$ that extends~$\phi^{c \mapsto y}$ for some~$v \in Y$ and therefore~$G[X]$ contains~$H$ as a subgraph. It follows that the instance~$(G[X], H)$ of \kFSubgraphTest, whose size is polynomial in~$k$, is equivalent to~$(G,H)$. As the set~$X$ can be computed in polynomial time, this gives a valid kernelization algorithm of polynomial size.
\end{proof}

\subsubsection{One fountain and many three-vertex components}
In this section, we consider pattern graphs consisting of one (induced subgraph of a) fountain and many connected components of at most three vertices. It will be convenient to define a notion of center for a fountain, similarly as we did for subdivided stars. We will only concern ourselves with fountains whose cycle has length three; these are triangles where pendant vertices are attached to one vertex of the triangle. If such a fountain has a unique high-degree vertex, then this is the center of the fountain. In the degenerate case that there are no pendant vertices and the graph is two-regular, the graph has no center and the pattern graph~$H$ is $3$-small. As in the previous section we say that a subgraph model of a fountain~$H'$ in~$G$ is \emph{centered at~$v \in V(G)$} if the image of the center is~$v$. The main algorithmic tool will again be a lemma that identifies a small representative vertex set for the set of possible centers of $H$-subgraph models in~$G$.

\begin{lemma} \label{lemma:fountainandtriangles:findrelevantcenters}
There is a polynomial-time algorithm with the following specifications. The input consists of a graph~$G$, a graph~$H' \in \Ffountain{3}$ with center~$c$, and a graph~$H''$ where each connected component has at most three vertices. Let~$H := H' + H''$ and~$k := |V(H)|$. The output is a set~$Y \subseteq V(G)$ of size~$\Oh(k^3)$ such that if~$H \subseteq G$, then there is a full $H$-subgraph model in~$G$ in which the model of~$H'$ is centered at a vertex in~$Y$.
\end{lemma}
\begin{proof}
Let~$c$ be the center of~$H'$ and let~$d,e$ be the two other vertices on the unique triangle in~$H'$. (Note that we require~$H' \in \Ffountain{3}$ which ensures that~$H'$ has a unique triangle; induced subgraphs of fountains that do not have a triangle do not satisfy the preconditions and will be dealt with separately.) We apply the sunflower lemma in a similar fashion as in Lemma~\ref{lemma:compute:representative:set:constanth} to compute the desired set~$Y$. Define a system of sets~$\S$ containing all triples~$\{c',d',e'\} \in \binom{V(G)}{3}$ for which there is a $H'$-subgraph model~$\phi$ in~$G$ with~$\phi(c) = c', \phi(d) = d'$ and~$\phi(e) = e'$. We can compute~$\S$ in polynomial time, since~$\{c',d',e'\} \in \binom{V(G)}{3}$ can realize~$H'$ if and only if~$\deg_G{c'} \geq \deg_H(c)$ and the triple induces a triangle in~$G$.

We will compute a set~$\S' \subseteq \S$ with the following \emph{preservation property}: if~$G$ contains a full $H$-subgraph model, then~$G$ contains a full $H$-subgraph model~$\phi$ with~$\phi(\{c,d,e\}) \in \S'$. By our choice of~$\S$ it is clear that~$\S$ has the preservation property, so we initialize~$\S'$ as a copy of~$\S$.

\begin{claim} \label{claim:fountainsandtriangles:preservationproperty}
If~$\S' \subseteq \S$ has the preservation property and~$|\S'| \geq 6k^3$, we can identify a set~$S^* \in \S'$ in polynomial time such that~$\S' \setminus \{S^*\}$ also has the preservation property.
\end{claim}
\begin{claimproof}
Suppose that~$|\S'| \geq 6k^3$. By Lemma~\ref{lemma:sunflowers}, there is a sunflower in~$\S'$ consisting of at least~$k+1$ sets~$S_1, \ldots, S_{k+1}$ and this can be found in time polynomial in the size of the set family and the universe. Let~$C := \bigcap _{i=1}^{k+1} S_i$ be the core of the sunflower. We show that~$\S' \setminus \{S_1\}$ has the preservation property.

Suppose that~$\phi$ is a full $H$-subgraph model in~$G$. As~$\S'$ has the preservation property, there is a full $H$-subgraph model~$\phi'$ in~$G$ such that~$S_{\phi'} := \phi'(\{c,d,e\})$ is contained in~$\S'$. If~$S_{\phi'} \neq S_1$ then the set~$S_{\phi'}$ is contained in~$\S' \setminus \{S_1\}$, which establishes the preservation property. If~$S_1 = S_{\phi'}$, then we distinguish two cases depending on whether or not the core~$C$ is empty. 

\textbf{Empty core.} If~$C = \emptyset$, then sets~$S_1, \ldots, S_{k+1}$ induce vertex-disjoint triangles in~$G$ that each have a vertex of degree at least~$\deg_H(c)$. We show that there is an $H$-subgraph model~$\phi^*$ with~$\phi^*(c,d,e) = S_2$, thereby showing that~$S_1$ may be safely discarded. Let~$S_2 = \{c',d',e'\}$ and let~$c' \in V(G)$ satisfy~$\deg_G(c') \geq \deg_H(c)$. Define~$\phi^*(c) = c', \phi^*(d) = d'$, and~$\phi^*(e) = e'$. Let~$v_1, \ldots, v_{\deg_H(c)}$ be distinct vertices in~$N_G(c') \setminus \{c',d'\}$, which exist by the lower bound on the degree of~$c'$. We use~$v_1, \ldots, v_{\deg_H(c)}$ as the images under~$\phi^*$ of the pendant vertices attached to~$c$ in the fountain~$H'$. It remains to define images for the connected components in~$H''$. By the precondition to the lemma, each such component has at most three vertices. Observe that the set~$Z = \{c',d',e'\} \cup \{v_1, \ldots, v_{\deg_H(c)}\}$ has size at most~$|V(H')|$. As the sunflower has an empty core, the sets~$S_i$ are pairwise disjoint and therefore each vertex in~$Z$ intersects at most one set in the sunflower. Hence~$Z$ intersects less than~$|V(H')|$ sets in the sunflower. Consequently, there are at least~$k - |V(H')| = |V(H)| - |V(H')| = |V(H'')|$ sets among~$S_2, \ldots, S_{k+1}$ that are not intersected by~$Z$. As each set induces a triangle in~$G$, which is a supergraph of any connected component of~$H''$, each set not intersected by~$Z$ can form the image of one connected component of~$H''$. Hence there is a~$H$-subgraph model~$\phi^*$ in~$G$ with~$\phi^*(c,d,e) = S_2$, which shows that~$\S' \setminus \{S_1\}$ has the preservation property.

\textbf{Non-empty core.} Now we deal with the case that~$C \neq \emptyset$. We construct a full $H$-subgraph model~$\phi^*$ in~$G$ such that~$\phi'(\{a,b,c\}) \in \{S_2, \ldots, S_{k+1}\}$. Recall that~$\phi'$ is the $H$-subgraph model we obtained from~$\S'$. For every vertex~$v \in V(H'')$, define~$\phi^*(v) = \phi'(v)$. Let~$Z := \bigcup _{v \in V(H'')} \phi'(v)$, which has size~$|V(H'')| < |V(H)| = k$. Hence there is a set~$S_i \in \{S_2, \ldots, S_{k + 1}\}$ that is disjoint from~$Z$. Let~$S_i = \{c_i, d_i, e_i\}$. As the core~$C$ is non-empty, at least one of the vertices in~$S_i$ is contained in~$C$; let~$c_i \in C$ and define~$\phi^*(c) = c_i, \phi^*(d) = d_i$, and~$\phi^*(e) = e_i$. It remains to define an image for the pendant vertices attached to~$c$ in the fountain. We claim that~$\deg_G(c_i) \geq k + 1$. To see this, observe that~$c_i$ is contained in the core of the sunflower, while each of the petals~$S_1 \setminus C, \ldots, S_{k+1} \setminus C$ contain at least one vertex that occurs in a common triangle with~$c_i$ (by definition of the sets in~$\S$). Hence each set in~$S_1 \setminus C, \ldots, S_{k+1} \setminus C$ contains a neighbor of~$c_i$, and as these sets are disjoint by the definition of a sunflower we indeed have~$\deg_G(c_i) \geq k+1$. The number of pendant vertices attached to~$c$ in~$H'$ is exactly~$\deg_{H}(c) = |V(H')| - 1$. As~$Z \cup \{b_i, c_i\}$ has size~$|V(H'')| + 2$, there are at least~$k+1 - (|V(H''| + 2) = |V(H')| - 1$ vertices in~$N_G(c_i) \setminus (Z \cup \{b_i, c_i\})$. Letting each such vertex form the image of one pendant neighbor of~$c$ in~$H'$, we can extend~$\phi^*$ to a full $H$-subgraph model. As~$\phi^*(\{c,d,e\}) = S_i \in \S' \setminus \{S_1\}$, the latter set has the preservation property and we can safely omit~$S_1$. This proves Claim~\ref{claim:fountainsandtriangles:preservationproperty}.
\end{claimproof}

By iterating the argument above, we arrive at a set system~$\S' \subseteq \S$ with the preservation property that contains at most~$6k^3$ sets. As each set in~$\S' \subseteq \S$ has size three, the set~$Y := \bigcup _{S \in \S'} S$ has size at most~$18k^3$. The preservation property directly implies that if there is a $H$-subgraph model in~$G$, then there is such a model that uses one of the triangles in~$\S'$ as the image for the triangle in the fountain, and that therefore centers the fountain at a member of~$Y$. Hence~$Y$ is a valid output for the procedure. As each iteration can be done in time polynomial in the size of~$\S$ and~$G$, while the number of iterations is bounded by~$|\S|$ which is~$\Oh(|V(G)|^3)$, the procedure runs in polynomial time. This concludes the proof of Lemma~\ref{lemma:fountainandtriangles:findrelevantcenters}.
\end{proof}

Lemma~\ref{lemma:starsandpaths:findrelevantcenters} allows us to derive a polynomial many-one kernel.

\restatekarpfountaintriangles*

\begin{proof}
The main idea is the same as in the proof of Theorem~\ref{theorem:karp:subgraphkernel:starsandpaths}: find a representative set of centers of size polynomial in~$k$ and then invoke Theorem~\ref{lemma:kernel:generic} for each choice of center. As we are proving the theorem for the hereditary graph family~\F, we have to consider not just graphs~$H'$ that are contained in~$\Ffountain{3}$, but also their induced subgraphs. These are of two types: removing pendant vertices of a graph in \Ffountain{3} results in another graph in \Ffountain{3}. However, if we remove a vertex from the unique triangle of a graph in \Ffountain{3}, then this either reduces the graph to a star, or splits it into components with at most two vertices. Since Lemma~\ref{lemma:fountainandtriangles:findrelevantcenters} does not apply to stars, we have to deal with the case that~$H'$ is a member of \Ffountain{3} and that~$H'$ is a star separately.

The kernelization algorithm works as follows. On input~$(G,H)$ it first tests whether~$H$ has the right form, which is easy to do in polynomial time. If this is not the case, it outputs a constant-size \no-instance. The behavior on the remaining instances depends on the form of~$H'$, which is a connected induced subgraph of a member of \Ffountain{3}. Recall that the parameter~$k$ is defined as~$|V(H)|$.

If~$H' \in \Ffountain{3}$ and~$H'$ has a unique center~$c$ then we proceed as follows. We invoke Lemma~\ref{lemma:fountainandtriangles:findrelevantcenters} to compute a set~$Y \subseteq V(G)$ of size~$\Oh(k^3)$ such that, if~$G$ contains~$H$ as a subgraph, then there is a $H$-subgraph model that centers~$H'$ at a member of~$Y$. Let~$c$ be the center of~$H'$. For each~$y \in Y$ we create a partial $H$-subgraph model~$\phi^{c \mapsto y}$ with domain~$\{c\}$ that sets~$\phi^{c \mapsto y} := y$. Define~$D := \{c\}$. Just as in the proof of Theorem~\ref{theorem:karp:subgraphkernel:starsandpaths}, these parameters satisfy the requirements for Lemma~\ref{lemma:kernel:generic} with~$(a,b,d) = (3,0,0)$. We combine the resulting sets~$X^{c \mapsto y}$ over all choices of~$y$ into a set~$X$, and output the instance~$(G[X], H)$. Since~$|X| \in |Y| \cdot \Oh(k^{\Oh(a + b^2 + d)}) = \Oh(k^{\Oh(1)})$, the size of the reduced graph~$G[X]$ is bounded by a polynomial in~$k$. Correctness follows by exactly the same arguments as Theorem~\ref{theorem:karp:subgraphkernel:starsandpaths}. Hence we obtain a polynomial kernel if~$H' \in \Ffountain{3}$.

If~$H' \not \in \Ffountain{3}$, then~$H'$ is a star (if~$H' \not \in \Ffountain{3}$) or~$H'$ is a triangle (if it has no center), which is a $1$-thin/$3$-small graph. Hence we obtain a polynomial kernel through Theorem~\ref{theorem:kernel:packing}. 
\end{proof}

\bibliographystyle{abbrvurl}
\bibliography{Paper}

\begin{thebibliography}{10}

\bibitem{Abu-Khzam10}
F.~N. Abu-Khzam.
\newblock An improved kernelization algorithm for $r$-{S}et {P}acking.
\newblock {\em Inf. Process. Lett.}, 110(16):621--624, 2010.
\newblock \href {http://dx.doi.org/10.1016/j.ipl.2010.04.020}
  {\path{doi:10.1016/j.ipl.2010.04.020}}.

\bibitem{AlonYZ95}
N.~Alon, R.~Yuster, and U.~Zwick.
\newblock Color-coding.
\newblock {\em J. ACM}, 42(4):844--856, 1995.

\bibitem{AmbalathBHKMPR10}
A.~M. Ambalath, R.~Balasundaram, C.~R. H., V.~Koppula, N.~Misra, G.~Philip, and
  M.~S. Ramanujan.
\newblock On the kernelization complexity of colorful motifs.
\newblock In {\em Proc. 5th IPEC}, pages 14--25, 2010.
\newblock \href {http://dx.doi.org/10.1007/978-3-642-17493-3_4}
  {\path{doi:10.1007/978-3-642-17493-3_4}}.

\bibitem{AtminasLR12}
A.~Atminas, V.~V. Lozin, and I.~Razgon.
\newblock Linear time algorithm for computing a small biclique in graphs
  without long induced paths.
\newblock In {\em Proc. 13th SWAT}, pages 142--152, 2012.
\newblock \href {http://dx.doi.org/10.1007/978-3-642-31155-0_13}
  {\path{doi:10.1007/978-3-642-31155-0_13}}.

\bibitem{BeinekeS76}
L.~Beineke and A.~J. Schwenk.
\newblock On a bipartite form of the {R}amsey problem.
\newblock In {\em Proc. 5th British Combinatorial Conference}, volume~XV of
  {\em Utilitas Math.}, pages 17--22, 1976.

\bibitem{Binkele-RaibleFFLSV12}
D.~Binkele-Raible, H.~Fernau, F.~V. Fomin, D.~Lokshtanov, S.~Saurabh, and
  Y.~Villanger.
\newblock Kernel(s) for problems with no kernel: {On} out-trees with many
  leaves.
\newblock {\em ACM Trans. Algorithms}, 8(4):38, 2012.
\newblock \href {http://dx.doi.org/10.1145/2344422.2344428}
  {\path{doi:10.1145/2344422.2344428}}.

\bibitem{Bjorklund10}
A.~Bj{\"o}rklund.
\newblock Determinant sums for undirected {H}amiltonicity.
\newblock In {\em Proc. 51st FOCS}, pages 173--182, 2010.
\newblock \href {http://dx.doi.org/10.1109/FOCS.2010.24}
  {\path{doi:10.1109/FOCS.2010.24}}.

\bibitem{Bodlaender09}
H.~L. Bodlaender.
\newblock Kernelization: New upper and lower bound techniques.
\newblock In {\em Proc. 4th {IWPEC}}, pages 17--37, 2009.
\newblock \href {http://dx.doi.org/10.1007/978-3-642-11269-0_2}
  {\path{doi:10.1007/978-3-642-11269-0_2}}.

\bibitem{BodlaenderDFH09}
H.~L. Bodlaender, R.~G. Downey, M.~R. Fellows, and D.~Hermelin.
\newblock On problems without polynomial kernels.
\newblock {\em J. Comput. Syst. Sci.}, 75(8):423--434, 2009.
\newblock \href {http://dx.doi.org/10.1016/j.jcss.2009.04.001}
  {\path{doi:10.1016/j.jcss.2009.04.001}}.

\bibitem{BodlaenderFLPST09}
H.~L. Bodlaender, F.~V. Fomin, D.~Lokshtanov, E.~Penninkx, S.~Saurabh, and
  D.~M. Thilikos.
\newblock ({M}eta) {K}ernelization.
\newblock In {\em Proc. 50th FOCS}, pages 629--638, 2009.
\newblock \href {http://dx.doi.org/10.1109/FOCS.2009.46}
  {\path{doi:10.1109/FOCS.2009.46}}.

\bibitem{BodlaenderJK12c}
H.~L. Bodlaender, B.~M.~P. Jansen, and S.~Kratsch.
\newblock Kernel bounds for path and cycle problems.
\newblock {\em Theor. Comput. Sci.}, 511:117--136, 2013.
\newblock \href {http://arxiv.org/abs/1106.4141} {\path{arXiv:1106.4141}},
  \href {http://dx.doi.org/10.1016/j.tcs.2012.09.006}
  {\path{doi:10.1016/j.tcs.2012.09.006}}.

\bibitem{BodlaenderJK14}
H.~L. Bodlaender, B.~M.~P. Jansen, and S.~Kratsch.
\newblock Kernelization lower bounds by cross-composition.
\newblock {\em SIAM J. Discrete Math.}, 28(1):277--305, 2014.
\newblock \href {http://arxiv.org/abs/1206.5941} {\path{arXiv:1206.5941}},
  \href {http://dx.doi.org/10.1137/120880240} {\path{doi:10.1137/120880240}}.

\bibitem{BodlaenderTY11}
H.~L. Bodlaender, S.~Thomass{\'e}, and A.~Yeo.
\newblock Kernel bounds for disjoint cycles and disjoint paths.
\newblock {\em Theor. Comput. Sci.}, 412(35):4570--4578, 2011.
\newblock \href {http://dx.doi.org/10.1016/j.tcs.2011.04.039}
  {\path{doi:10.1016/j.tcs.2011.04.039}}.

\bibitem{CarnielliC1999}
W.~Carnielli and E.~M. Carmelo.
\newblock On the {R}amsey problem for multicolor bipartite graphs.
\newblock {\em Advances in Applied Mathematics}, 22(1):48--59, 1999.
\newblock \href {http://dx.doi.org/10.1006/aama.1998.0620}
  {\path{doi:10.1006/aama.1998.0620}}.

\bibitem{ChenLSZ07}
J.~Chen, S.~Lu, S.-H. Sze, and F.~Zhang.
\newblock Improved algorithms for path, matching, and packing problems.
\newblock In {\em Proc. 18th SODA}, pages 298--307, 2007.
\newblock URL: \url{http://doi.acm.org/10.1145/1283383.1283415}.

\bibitem{DBLP:conf/icalp/ChenTW08}
Y.~Chen, M.~Thurley, and M.~Weyer.
\newblock Understanding the complexity of induced subgraph isomorphisms.
\newblock In {\em Proc. 35th ICALP}, pages 587--596, 2008.
\newblock \href {http://dx.doi.org/10.1007/978-3-540-70575-8_48}
  {\path{doi:10.1007/978-3-540-70575-8_48}}.

\bibitem{CyganKPPW12}
M.~Cygan, S.~Kratsch, M.~Pilipczuk, M.~Pilipczuk, and M.~Wahlstr{\"o}m.
\newblock Clique cover and graph separation: New incompressibility results.
\newblock In {\em Proc. 39th ICALP}, pages 254--265, 2012.
\newblock \href {http://dx.doi.org/10.1007/978-3-642-31594-7_22}
  {\path{doi:10.1007/978-3-642-31594-7_22}}.

\bibitem{DBLP:journals/tcs/DalmauJ04}
V.~Dalmau and P.~Jonsson.
\newblock The complexity of counting homomorphisms seen from the other side.
\newblock {\em Theor. Comput. Sci.}, 329(1-3):315--323, 2004.
\newblock \href {http://dx.doi.org/10.1016/j.tcs.2004.08.008}
  {\path{doi:10.1016/j.tcs.2004.08.008}}.

\bibitem{DellM12}
H.~Dell and D.~Marx.
\newblock Kernelization of packing problems.
\newblock In {\em Proc. 23rd SODA}, pages 68--81, 2012.

\bibitem{DellM10}
H.~Dell and D.~van Melkebeek.
\newblock Satisfiability allows no nontrivial sparsification unless the
  polynomial-time hierarchy collapses.
\newblock In {\em Proc. 42nd STOC}, pages 251--260, 2010.
\newblock \href {http://dx.doi.org/10.1145/1806689.1806725}
  {\path{doi:10.1145/1806689.1806725}}.

\bibitem{DomLS09}
M.~Dom, D.~Lokshtanov, and S.~Saurabh.
\newblock Incompressibility through colors and {ID}s.
\newblock In {\em Proc. 36th ICALP}, pages 378--389, 2009.
\newblock \href {http://dx.doi.org/10.1007/978-3-642-02927-1_32}
  {\path{doi:10.1007/978-3-642-02927-1_32}}.

\bibitem{DowneyF13}
R.~G. Downey and M.~R. Fellows.
\newblock {\em Fundamentals of Parameterized Complexity}.
\newblock Texts in Computer Science. Springer, 2013.

\bibitem{Drucker12}
A.~Drucker.
\newblock New limits to classical and quantum instance compression.
\newblock In {\em Proc. 53rd FOCS}, pages 609--618, 2012.
\newblock \href {http://dx.doi.org/10.1109/FOCS.2012.71}
  {\path{doi:10.1109/FOCS.2012.71}}.

\bibitem{MR22:2554}
P.~Erd{\H{o}}s and R.~Rado.
\newblock Intersection theorems for systems of sets.
\newblock {\em J. London Math. Soc.}, 35:85--90, 1960.

\bibitem{FlumG06}
J.~Flum and M.~Grohe.
\newblock {\em Parameterized Complexity Theory}.
\newblock Springer-Verlag New York, Inc., 2006.

\bibitem{FominLMS12}
F.~V. Fomin, D.~Lokshtanov, N.~Misra, and S.~Saurabh.
\newblock Planar $\mathcal{F}$-{D}eletion: Approximation, kernelization and
  optimal {FPT} algorithms.
\newblock In {\em Proc. 53rd FOCS}, pages 470--479, 2012.
\newblock \href {http://dx.doi.org/10.1109/FOCS.2012.62}
  {\path{doi:10.1109/FOCS.2012.62}}.

\bibitem{DBLP:journals/jcss/FominLRSR12}
F.~V. Fomin, D.~Lokshtanov, V.~Raman, S.~Saurabh, and B.~V.~R. Rao.
\newblock Faster algorithms for finding and counting subgraphs.
\newblock {\em J. Comput. Syst. Sci.}, 78(3):698--706, 2012.
\newblock \href {http://dx.doi.org/10.1016/j.jcss.2011.10.001}
  {\path{doi:10.1016/j.jcss.2011.10.001}}.

\bibitem{FortnowS11}
L.~Fortnow and R.~Santhanam.
\newblock Infeasibility of instance compression and succinct {PCP}s for {NP}.
\newblock {\em J. Comput. Syst. Sci.}, 77(1):91--106, 2011.
\newblock \href {http://dx.doi.org/10.1016/j.jcss.2010.06.007}
  {\path{doi:10.1016/j.jcss.2010.06.007}}.

\bibitem{GareyJ79}
M.~R. Garey and D.~S. Johnson.
\newblock {\em Computers and Intractability, A Guide to the Theory of
  {NP}-Complete\-ness}.
\newblock W.H. Freeman and Company, New York, 1979.

\bibitem{GrahamGL95}
R.~L. Graham, M.~Gr\"{o}tschel, and L.~Lov\'{a}sz, editors.
\newblock {\em Handbook of Combinatorics (Vol. 2)}.
\newblock MIT Press, Cambridge, MA, USA, 1995.

\bibitem{1206036}
M.~Grohe.
\newblock The complexity of homomorphism and constraint satisfaction problems
  seen from the other side.
\newblock {\em J. ACM}, 54(1):1, 2007.
\newblock \href {http://dx.doi.org/10.1145/1206035.1206036}
  {\path{doi:10.1145/1206035.1206036}}.

\bibitem{380867}
M.~Grohe, T.~Schwentick, and L.~Segoufin.
\newblock When is the evaluation of conjunctive queries tractable?
\newblock In {\em STOC '01: Proceedings of the thirty-third annual ACM
  symposium on Theory of computing}, pages 657--666, New York, NY, USA, 2001.
  ACM Press.
\newblock \href {http://dx.doi.org/http://doi.acm.org/10.1145/380752.380867}
  {\path{doi:http://doi.acm.org/10.1145/380752.380867}}.

\bibitem{HermelinKSWW13}
D.~Hermelin, S.~Kratsch, K.~So{\l}tys, M.~Wahlstr{\"o}m, and X.~Wu.
\newblock A completeness theory for polynomial ({T}uring) kernelization.
\newblock In {\em Proc. 8th IPEC}, pages 202--215, 2013.
\newblock \href {http://dx.doi.org/10.1007/978-3-319-03898-8_18}
  {\path{doi:10.1007/978-3-319-03898-8_18}}.

\bibitem{HermelinW12}
D.~Hermelin and X.~Wu.
\newblock Weak compositions and their applications to polynomial lower bounds
  for kernelization.
\newblock In {\em Proc. 23rd SODA}, pages 104--113, 2012.

\bibitem{HopcroftK73}
J.~E. Hopcroft and R.~M. Karp.
\newblock An {$n^{5/2}$} algorithm for maximum matchings in bipartite graphs.
\newblock {\em SIAM J. Comput.}, 2(4):225--231, 1973.
\newblock \href {http://dx.doi.org/10.1137/0202019}
  {\path{doi:10.1137/0202019}}.

\bibitem{Jansen14}
B.~M.~P. Jansen.
\newblock Turing kernelization for finding long paths and cycles in restricted
  graph classes.
\newblock In {\em Proc. 22nd ESA}, pages 579--591, 2014.
\newblock \href {http://arxiv.org/abs/1402.4718} {\path{arXiv:1402.4718}}.

\bibitem{JansenB11}
B.~M.~P. Jansen and H.~L. Bodlaender.
\newblock Vertex cover kernelization revisited: Upper and lower bounds for a
  refined parameter.
\newblock In {\em Proc.\ 28th {STACS}}, pages 177--188, 2011.
\newblock \href {http://dx.doi.org/10.4230/LIPIcs.STACS.2011.177}
  {\path{doi:10.4230/LIPIcs.STACS.2011.177}}.

\bibitem{JansenK12}
B.~M.~P. Jansen and S.~Kratsch.
\newblock On polynomial kernels for structural parameterizations of odd cycle
  transversal.
\newblock In {\em Proc. 6th IPEC}, pages 132--144, 2011.
\newblock \href {http://dx.doi.org/10.1007/978-3-642-28050-4_11}
  {\path{doi:10.1007/978-3-642-28050-4_11}}.

\bibitem{JansenK13}
B.~M.~P. Jansen and S.~Kratsch.
\newblock Data reduction for graph coloring problems.
\newblock {\em Inform. Comput.}, 231:70--88, 2013.
\newblock \href {http://arxiv.org/abs/1104.4229} {\path{arXiv:1104.4229}},
  \href {http://dx.doi.org/10.1016/j.ic.2013.08.005}
  {\path{doi:10.1016/j.ic.2013.08.005}}.

\bibitem{DBLP:journals/tcs/KhotR02}
S.~Khot and V.~Raman.
\newblock Parameterized complexity of finding subgraphs with hereditary
  properties.
\newblock {\em Theor. Comput. Sci.}, 289(2):997--1008, 2002.
\newblock \href {http://dx.doi.org/10.1016/S0304-3975(01)00414-5}
  {\path{doi:10.1016/S0304-3975(01)00414-5}}.

\bibitem{KirkpatrickH78}
D.~G. Kirkpatrick and P.~Hell.
\newblock On the completeness of a generalized matching problem.
\newblock In {\em Proc. 10th STOC}, pages 240--245, 1978.
\newblock \href {http://dx.doi.org/10.1145/800133.804353}
  {\path{doi:10.1145/800133.804353}}.

\bibitem{KneisMRR06}
J.~Kneis, D.~M{\"o}lle, S.~Richter, and P.~Rossmanith.
\newblock Divide-and-color.
\newblock In {\em Proc. 32nd WG}, pages 58--67, 2006.
\newblock \href {http://dx.doi.org/10.1007/11917496_6}
  {\path{doi:10.1007/11917496_6}}.

\bibitem{Kratsch13}
S.~Kratsch.
\newblock Co-nondeterminism in compositions: a kernelization lower bound for a
  {R}amsey-type problem.
\newblock In {\em ACM Trans. Algorithms}, 2013.
\newblock To appear.

\bibitem{KratschPRR12}
S.~Kratsch, M.~Pilipczuk, A.~Rai, and V.~Raman.
\newblock Kernel lower bounds using co-nondeterminism: Finding induced
  hereditary subgraphs.
\newblock In {\em Proc. 13th SWAT}, pages 364--375, 2012.
\newblock \href {http://dx.doi.org/10.1007/978-3-642-31155-0_32}
  {\path{doi:10.1007/978-3-642-31155-0_32}}.

\bibitem{LewisY80}
J.~M. Lewis and M.~Yannakakis.
\newblock The node-deletion problem for hereditary properties is {NP}-complete.
\newblock {\em J. Comput. Syst. Sci.}, 20(2):219--230, 1980.

\bibitem{Lin14}
B.~Lin.
\newblock The parameterized complexity of $k$-{B}iclique.
\newblock In {\em Proc. 26th SODA}, 2014.
\newblock In press.
\newblock \href {http://arxiv.org/abs/1406.3700} {\path{arXiv:1406.3700}}.

\bibitem{LokshtanovMS12}
D.~Lokshtanov, N.~Misra, and S.~Saurabh.
\newblock Kernelization - {P}reprocessing with a guarantee.
\newblock In {\em The Multivariate Algorithmic Revolution and Beyond}, pages
  129--161, 2012.
\newblock \href {http://dx.doi.org/10.1007/978-3-642-30891-8_10}
  {\path{doi:10.1007/978-3-642-30891-8_10}}.

\bibitem{DBLP:journals/ipl/Marx04}
D.~Marx.
\newblock List edge multicoloring in graphs with few cycles.
\newblock {\em Inf. Process. Lett.}, 89(2):85--90, 2004.
\newblock \href {http://dx.doi.org/10.1016/j.ipl.2003.09.016}
  {\path{doi:10.1016/j.ipl.2003.09.016}}.

\bibitem{MarxP14}
D.~Marx and M.~Pilipczuk.
\newblock Everything you always wanted to know about the parameterized
  complexity of subgraph isomorphism (but were afraid to ask).
\newblock In {\em Proc. 31st STACS}, pages 542--553, 2014.
\newblock \href {http://dx.doi.org/10.4230/LIPIcs.STACS.2014.542}
  {\path{doi:10.4230/LIPIcs.STACS.2014.542}}.

\bibitem{Monien85}
B.~Monien.
\newblock How to find long paths efficiently.
\newblock {\em Annals of Discrete Mathematics}, 25:239--254, 1985.
\newblock \href {http://dx.doi.org/10.1016/S0304-0208(08)73110-4}
  {\path{doi:10.1016/S0304-0208(08)73110-4}}.

\bibitem{Moser09}
H.~Moser.
\newblock A problem kernelization for graph packing.
\newblock In {\em Proc. 35th SOFSEM}, pages 401--412, 2009.
\newblock \href {http://dx.doi.org/10.1007/978-3-540-95891-8_37}
  {\path{doi:10.1007/978-3-540-95891-8_37}}.

\bibitem{MulmuleyVV87}
K.~Mulmuley, U.~V. Vazirani, and V.~V. Vazirani.
\newblock Matching is as easy as matrix inversion.
\newblock {\em Combinatorica}, 7(1):105--113, 1987.
\newblock \href {http://dx.doi.org/10.1007/BF02579206}
  {\path{doi:10.1007/BF02579206}}.

\bibitem{Plesnik99}
J.~Plesn\'{i}k.
\newblock Constrained weighted matchings and edge coverings in graphs.
\newblock {\em Discrete Appl. Math.}, 92(2–3):229 -- 241, 1999.
\newblock \href {http://dx.doi.org/10.1016/S0166-218X(99)00052-9}
  {\path{doi:10.1016/S0166-218X(99)00052-9}}.

\bibitem{SchaferKMN12}
A.~Sch{\"a}fer, C.~Komusiewicz, H.~Moser, and R.~Niedermeier.
\newblock Parameterized computational complexity of finding small-diameter
  subgraphs.
\newblock {\em Optim. Lett.}, 6(5):883--891, 2012.
\newblock \href {http://dx.doi.org/10.1007/s11590-011-0311-5}
  {\path{doi:10.1007/s11590-011-0311-5}}.

\bibitem{ThomasseTV13}
S.~Thomass{\'e}, N.~Trotignon, and K.~Vuskovic.
\newblock Parameterized algorithm for weighted independent set problem in
  bull-free graphs.
\newblock arXiv 1310.6205, 2013.
\newblock \href {http://arxiv.org/abs/1310.6205} {\path{arXiv:1310.6205}}.

\bibitem{DBLP:journals/ipl/Williams09}
R.~Williams.
\newblock Finding paths of length $k$ in ${O}^*(2^{k})$ time.
\newblock {\em Inf. Process. Lett.}, 109(6):315--318, 2009.
\newblock \href {http://dx.doi.org/10.1016/j.ipl.2008.11.004}
  {\path{doi:10.1016/j.ipl.2008.11.004}}.

\bibitem{DBLP:journals/siamcomp/WilliamsW13}
V.~V. Williams and R.~Williams.
\newblock Finding, minimizing, and counting weighted subgraphs.
\newblock {\em SIAM J. Comput.}, 42(3):831--854, 2013.
\newblock \href {http://dx.doi.org/10.1137/09076619X}
  {\path{doi:10.1137/09076619X}}.

\bibitem{DBLP:journals/siamcomp/Yannakakis81a}
M.~Yannakakis.
\newblock Node-deletion problems on bipartite graphs.
\newblock {\em SIAM J. Comput.}, 10(2):310--327, 1981.
\newblock \href {http://dx.doi.org/10.1137/0210022}
  {\path{doi:10.1137/0210022}}.

\bibitem{Yap83}
C.-K. Yap.
\newblock Some consequences of non-uniform conditions on uniform classes.
\newblock {\em Theor. Comput. Sci.}, 26:287--300, 1983.
\newblock \href {http://dx.doi.org/10.1016/0304-3975(83)90020-8}
  {\path{doi:10.1016/0304-3975(83)90020-8}}.

\end{thebibliography}

\end{document}